\documentclass[11pt]{scrartcl}

\newif\ifcomment
\commenttrue

\usepackage[utf8]{inputenc}
\usepackage[T1]{fontenc}
\usepackage{lmodern}

\usepackage{ifthen}
\usepackage{amssymb}
\usepackage{microtype}
\usepackage{amsmath}
\usepackage{amssymb}
\usepackage{amsfonts}
\usepackage{mathtools}
\usepackage{xcolor}
\usepackage{xstring}

 \usepackage[hyphens]{url}
 \usepackage{amsthm}
 \usepackage{hyperref}
 \usepackage{cleveref}
 \usepackage{tikz}
 \usepackage{subcaption}

 \usepackage[inline]{enumitem}
 \usepackage{titling}
 \usepackage{xspace}
 \usepackage{dsfont}

\usetikzlibrary{shapes.geometric,hobby,calc} 
\usetikzlibrary{decorations}
\usetikzlibrary{decorations.pathmorphing}
\usetikzlibrary{decorations.text}
\usetikzlibrary{shapes.misc}
\usetikzlibrary{decorations,shapes,snakes}

\usepackage{macros}

\usepackage{stackengine}
\stackMath

\renewcommand{\tilde}{\widetilde}
\renewcommand{\emptyset}{\varnothing}
\newcommand{\dirtw}{\text{dtw}}
\newcommand{\dmw}{\text{dmw}}
\newcommand{\fv}{\text{fv}}
\newcommand{\ham}{\text{ham}}
\newtheorem{fact}{Fact}
{\begin{adjustwidth}{2em}{}\hspace{-2em}\textbf{Case #1.} }%
{\end{adjustwidth}}

\hypersetup{
	colorlinks=true,
	linkcolor=myBlue,
	citecolor=myBlue,
	urlcolor=myLightBlue,
	bookmarksopen=true,
	bookmarksnumbered,
	bookmarksopenlevel=2,
	bookmarksdepth=3
}

\title{Parametrised Algorithms for Directed Modular Width}
\predate{}
\date{}
\postdate{}

\preauthor{}
\DeclareRobustCommand{\authorthing}{
\begin{center}
\begin{tabular}{p{.45\textwidth}p{.45\textwidth}}
Raphael Steiner & Sebastian Wiederrecht\\
\url{steiner@math.tu-berlin.de} & \url{sebastian.wiederrecht@tu-berlin.de}\\
\noalign{\smallskip}
\multicolumn{2}{c}{Technische Universität Berlin\quad}
\end{tabular}
\end{center}}
\author{\authorthing}
\postauthor{}

\begin{document}
\maketitle

\begin{abstract}
	Many well-known NP-hard algorithmic problems on directed graphs resist efficient parametrisations with most known width measures for directed graphs, such as directed treewidth, DAG-width, Kelly-width and many others. 
	While these focus on measuring how close a digraph is to an oriented tree resp.\ a directed acyclic graph, in this paper, we investigate \emph{directed modular width} as a parameter, which is closer to the concept of clique-width. We investigate applications of modular decompositions of directed graphs to a wide range of algorithmic problems and derive FPT-algorithms for several well-known digraph-specific NP-hard problems, namely \emph{minimum (weight) directed feedback vertex set}, \emph{minimum (weight) directed dominating set}, \emph{digraph colouring}, \emph{directed Hamiltonian path/cycle}, \emph{partitioning into paths}, \emph{(capacitated) vertex-disjoint directed paths}, and the \emph{directed subgraph homeomorphism problem}.
	The latter yields a polynomial-time algorithm for detecting topological minors in digraphs of bounded directed modular width.
	Finally we illustrate that also other structural digraph parameters, such as \emph{directed pathwidth} and the \emph{cycle-rank} can be computed efficiently using directed modular width as a parameter.
	
	~
	
\end{abstract}
\section{Introduction}

Width measures for graphs have become a fundamental pillar in both structural graph theory and the field of parametrised complexity (see \cite{downey2012parameterized} for an introduction to the latter).
From an algorithmic point of view a \emph{good} width measure should provide three things:
(1) graph classes of bounded width should have a reasonably rich structure,
(2) a large number of different problems should be efficiently solvable on those classes of bounded width, and at last
(3) a decomposition witnessing the width of a graph should be computable in a reasonable amount of time.
For undirected graphs width measures have been extremely successful.
For sparse graph classes the notion of \emph{treewidth} \cite{robertson1986graph} and equivalent width measures like \emph{branchwidth} \cite{robertson1991graph} have lead to a plethora of different algorithmic approaches to otherwise untractable problems \cite{bodlaender1997treewidth}.
Among the methods utilising treewidth there are also so called \emph{algorithmic meta theorems} like Courcelle's Theorem \cite{courcelle1990monadic} stating that any MSO$_2$-definable problem can be solved on graphs of bounded treewidth in linear time.

For dense graphs the approach of \emph{clique-width} \cite{courcelle1993handle} and related parameters like \emph{rank-width} \cite{oum2005rank} or \emph{boolean-width} \cite{bui2011boolean} still leads to the tractability of several otherwise hard problems on classes of bounded width \cite{espelage2001solve}.
In fact, in the case of clique-width still a powerful algorithmic meta theorem exists \cite{courcelle1997expression,courcelle2000linear}.

In harsh contrast to the immense success of undirected graph width measures for algorithmic applications, the setting of \emph{directed graphs}, despite numerous attempts, has resisted all approaches that tried to replicate the notion of treewidth for directed graphs.
That said, this does only apply to the algorithmic side of things.
From the viewpoint of structural graph theory, especially regarding the butterfly minor relation, directed treewidth seems to be the right approach.
As a major argument in this direction one has to mention the Directed Grid Theorem \cite{kawarabayashi2015directed}.
Several digraph width measures, similar to the undirected treewidth, have been proposed.
The most popular among them are \emph{directed treewidth} \cite{johnson2001directed}, \emph{DAG-width} \cite{berwanger2006dag}, and \emph{Kelly-width} \cite{hunter2008digraph}.
While many routing problems including the general \emph{disjoint paths problem} can be solved in polynomial time on digraphs of bounded directed tree-width, digraphs of bounded DAG- or Kelly-width also admit polynomial time algorithms for solving parity games.
In the case of directed treewidth there also exists an algorithmic meta theorem in the spirit of Courcelle's Theorem \cite{de2016algorithmic}, however it is restricted to a small class of MSO$_2$-definable problems, mainly concerned with the routing of paths.
Even more restrictive digraph parameters like \emph{directed pathwidth} (see for example \cite{yang2008digraph}) or \emph{cycle-rank} \cite{eggan1963transition} do not seem to to perform any better.
One reason for the limited success of these directed width measures seems to be the fact that all of them are bounded from above by the size of a minimum \emph{directed feedback vertex set}, i.e.\@ a set of vertices hitting all directed cycles.
In particular this means that \emph{DAG}s, directed acyclic graphs, have bounded width for all those parameters.
Sadly many classical hard problems like the \emph{directed dominating set} problem remain NP-hard on DAGs or even classes of low \emph{DAG-depth} \cite{ganian2009digraph}.

On the other side the notion of clique-width has simultaneously been introduced for undirected and directed graphs \cite{courcelle1993handle}.
The same algorithmic meta result that was obtained for die undirected version of clique width can be obtained for the directed version \cite{courcelle2000linear} making directed clique width and its equivalent \emph{bi-rank-width} algorithmically much more successful digraph width measures than directed treewidth and its cousins.
However, from the point of structural digraph theory directed clique width and bi-rank-width are not as well behaved as neither of them is closed under subgraphs or butterfly minors. It is furthermore not entirely clear how well a directed clique-width expression of bounded width can be approximated in reasonable time.
The current technique is to compute a bi-rank-width decomposition of optimal width, say $k$, and then use this decomposition to obtain a directed clique-width expression of width at most $2^{k+1}-1$ in FPT-time \cite{courcelle2012graph}.
Additionally, there are elementary problems that still are $W[1]$-hard on graphs of bounded clique-width like the \emph{Hamiltonian cycle} problem \cite{fomin2010intractability}.

In their work on algorithmic meta arguments regarding directed width measures Ganian et.\@ al.\@ \cite{ganian2016there} introduce a concept to capture the algorithmic abilities of a directed width parameter and reach the conclusion that no directed width measure that is closed under subgraphs and \emph{topological butterfly minors} can be a powerful tool in terms of algorithmics.

So an interesting problem would be to find a digraph width measure that does not struggle with more global decision problems like Hamiltonian cycles, while at the same time offers a broad variety of different problems that become efficiently computable on digraphs of bounded width and uses a decomposition concept that can be computed, or at least approximated, within a reasonable span in FPT-time.

In the case of undirected graphs clique width has similar problems, which led Gajarsky et.\@ al.\@ to consider the more restrictive parameter of \emph{modular width} \cite{gajarsky2013parameterized}.
For undirected graphs modular width fills the very specific niche as it covers dense undirected graphs, allows for FPT-time algorithms for a broad spectrum of problems and an optimal decomposition can be computed efficiently \cite{mcconnell1999modular}.

For a more in-depth overview on the topic of directed width measures we recommend the chapter on digraphs of bounded width in \cite{bang2018classes}.

\paragraph{Contribution}

The main contribution of this paper is to consider the directed version of modular width as a structural parameter for directed graphs in terms of its usefulness in parametrised algorithms.
While \emph{directed modular width} is more restrictive than the digraph parameters discussed above, the advantage is a wealth of otherwise untractable problems that turn out to admit FPT-algorithms when parametrised with the directed modular width of the input digraph. Similar to the undirected case a directed modular decomposition of optimal width can be computed in polynomial time \cite{moduledecomposition}. Besides other classical hard problems we give FPT-algorithms for the problems \emph{digraph colouring}, \emph{Hamiltonian cycle}, \emph{partitioning into paths}, \emph{capacitated $k$-disjoint paths}, the \emph{directed subgraph homeomorphism problem}, and for computing the \emph{directed pathwidth} and the \emph{cycle rank}. We obtain polynomial-time algorithms for finding topological minors in digraphs of bounded modular width. 

The dynamic programming approach we take utilises the recursive nature of a directed modular decomposition, however, combining the dynamic programming tables of the children of some node in our recursion tree turns out to be a non-trivial problem on its own. The strategy we use for designing the algorithms is described in \cref{sec:strategy}.

\paragraph{Organisation of the Paper}

For better readability we provide a short overview on the organisation of this work.
\Cref{sec:prelims} contains the preliminaries together with an introduction to directed modular width.
We briefly discuss some structural properties of directed modular width an show that \emph{directed co-graphs} form exactly the class of directed modular width $2$.
Afterwards, in \cref{sec:strategy} we give an overview on the general approach to dynamic programming we take for directed modular width and explain our use of integer programming as part of this technique.

In \cref{sec:FVS} we discuss the \emph{minimum weight feedback vertex set} problem, while in \cref{sec:dominating} \emph{minimum weight dominating set} is discussed.
What follows are the \emph{dichromatic number} in \cref{digraphcolouring} and the \emph{hamiltonicity} problem in \cref{sec:hamilton}.

\Cref{vddp} is dedicated to the \emph{(capacitated) directed disjoint paths} problem for which we need to further develop intermediate theoretical tools to make our recursive approach work. In \cref{sec:homeo} we show how to apply the results from the previous section to the \emph{directed subgraph homeomorphism problem} and derive some related results.
At last \cref{sec:otherwidth} is concerned with the computation of other directed width parameters on graphs of bounded directed modular width.

\section{Preliminaries}\label{sec:prelims}
Digraphs in this paper are considered loopless and without parallel edges, but may have pairs of anti-parallel directed edges (called \emph{digons}). If $D$ is a digraph, we denote by $V(D)$ the set of its vertices and by $E(D)$ the set of its directed edges (sometimes also called \emph{arcs}). If $e \in E(D)$ is an edge which starts at a vertex $u$ and ends in a vertex $v$, we write $e=(u,v)$ and say that $u$ is the \emph{tail} of $e$, while $v$ is its \emph{head}. For a subset $X \subseteq V(D)$ of vertices, we use $D[X]$ to denote the digraph obtained by restricting to the set $X$ of vertices, and we let $D-X\coloneqq D[V(D) \setminus X]$. A subdigraph obtained in this way is called \emph{induced}. If $u \in V(D)$ is a vertex, we let $D-u\coloneqq D-\{u\}$ denote the digraph obtained by deleting $u$. We furthermore denote by $N_D^+(u), N_D^-(u)$ the sets of out- and in-neighbours, respectively, of $u$. 
\paragraph{Directed Modular Width}
Modules in graphs are vertex subsets which have the same relations to vertices outside the set. For digraphs, we have the following similar definition.
\begin{definition}
	Let $D$ be a digraph. A subset $\emptyset \neq M \subseteq V(D)$ of vertices is called a \emph{module}, if all the vertices in $M$ have the same sets of out-neighbours and the same sets of in-neighbours outside the module. Formally, we have $N_D^+(u_1) \setminus M = N_D^+(u_2) \setminus M$ and $N_D^-(u_1) \setminus M = N_D^-(u_2) \setminus M$ for all $u_1, u_2 \in M$.
	Consider \cref{fig:singlemodule} for an illustration.
\end{definition}

\begin{figure}[h!]
	\begin{center}
		\begin{tikzpicture}[scale=0.7]
		
		\pgfdeclarelayer{background}
		\pgfdeclarelayer{foreground}
		
		\pgfsetlayers{background,main,foreground}
		
		\begin{pgfonlayer}{main}
		
		\node (C) [] {};
		
		\node (C1) [v:ghost, position=180:25mm from C] {};
		
		\node (C2) [v:ghost, position=0:0mm from C] {};
		
		\node (C3) [v:ghost, position=0:25mm from C] {};

		

		
		
		\node (v1) [v:main,position=0:0mm from C2] {};
		\node (v2) [v:main,position=80:17mm from v1] {};
		\node (v3) [v:main,position=130:12mm from v1] {};
		\node (v4) [v:main,position=190:12mm from v1] {};
		\node (v5) [v:main,position=240:17mm from v1] {};
		\node (v6) [v:main,position=15:14mm from v1] {};
		\node (v7) [v:main,position=340:10mm from v6] {};
		\node (v8) [v:main,position=70:12mm from v6] {};
		\node (v9) [v:main,position=255:17mm from v6] {};
		\node (v10) [v:main,position=300:18mm from v6] {};
		
		\node (g1) [v:ghost,position=90:5mm from v6] {};
		\node (g2) [v:ghost,position=165:4mm from v6] {};
		\node (g3) [v:ghost,position=260:5mm from v6] {};
		\node (g4) [v:ghost,position=250:5mm from v7] {};
		\node (g5) [v:ghost,position=0:4mm from v7] {};
		\node (g6) [v:ghost,position=100:5mm from v7] {};

		
		
		

		

		
		
		\draw (v1) [e:main,->,bend left=15] to (v2);
		\draw (v1) [e:main,->,bend left=15] to (v3);
		\draw (v1) [e:main,->,bend left=15] to (v4);
		\draw (v1) [e:main,->,bend left=15] to (v5);
		\draw (v2) [e:main,->,bend left=15] to (v1);
		\draw (v3) [e:main,->,bend left=15] to (v1);
		\draw (v4) [e:main,->,bend left=15] to (v1);
		\draw (v5) [e:main,->,bend left=15] to (v1);
		
		\draw (v1) [e:main,->,bend left=20] to (v8);
		\draw (v1) [e:main,->,bend right=30] to (v9);
		\draw (v1) [e:main,->,bend right=15] to (v10);
		
		\draw (v2) [e:main,->,bend right=10] to (v3);
		\draw (v3) [e:main,->] to (v4);
		\draw (v3) [e:main,->,bend right=55] to (v5);
		\draw (v4) [e:main,->,bend right=5] to (v2);
		\draw (v5) [e:main,->] to (v4);
		
		\draw (v6) [e:main,->,bend left=15] to (v7);
		\draw (v7) [e:main,->,bend left=15] to (v6);
		
		\draw (v9) [e:main,->,bend right=15] to (v10);
		\draw (v10) [e:main,->,bend right=55] to (v8);
		
		\draw (v8) [e:main,->] to (v6);
		\draw (v8) [e:main,->] to (v7);
		\draw (v9) [e:main,->] to (v6);
		\draw (v9) [e:main,->] to (v7);
		\draw (v10) [e:main,->] to (v6);
		\draw (v10) [e:main,->] to (v7);
		
		\draw (v6) [e:main,->,bend right=10] to (v1);
		\draw (v7) [e:main,->,bend left=22] to (v1);
		
		\draw (g2) [e:coloredthin,color=magenta,closed,curve through={(g3) (g4) (g5) (g6)}] to (g1);
		
		

		
		
		\end{pgfonlayer}
		

		\begin{pgfonlayer}{background}
		
		\end{pgfonlayer}	
		
		\begin{pgfonlayer}{foreground}

		\end{pgfonlayer}
		\end{tikzpicture}
	\end{center}
	\caption{A digraph $D$ together with a module \textcolor{magenta}{$M$}.}
	\label{fig:singlemodule}
\end{figure}
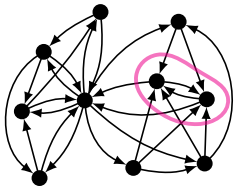

The following gives the precise recursive definition of directed modular width.
For an illustration see \cref{fig:moduledecomp}.

\begin{definition}[Directed Modular-Width]
	Let $k \in \mathbb{N}_0$, and let $D$ be a digraph. We say that $D$ has \emph{directed modular width at most} $k$, if one of the following holds:
	\begin{itemize}
		\item $|V(D)| \leq k$, or
		\item There exists a partition of $V(D)$ into $\ell \in \{2,\ldots,k\}$ modules $M_1, \ldots, M_\ell$ such that for every $i$, $D[M_i]$ has directed modular width at most $k$.
	\end{itemize}
	The least $k \ge 1$ for which a digraph $D$ has directed modular width at most $k$ is now defined to be the \emph{directed modular width}, denoted by $\dmw(D)$, of $D$.
\end{definition}

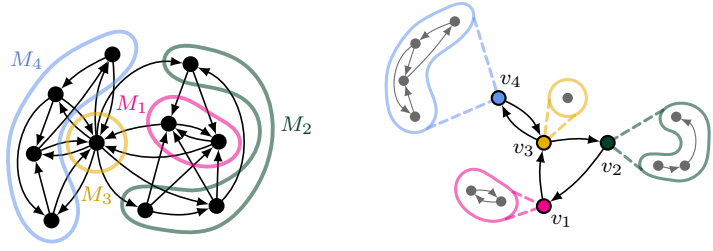
\begin{figure}[h!]
	\begin{center}
		\begin{tikzpicture}[scale=0.7]
		
		\pgfdeclarelayer{background}
		\pgfdeclarelayer{foreground}
		
		\pgfsetlayers{background,main,foreground}
		
		\begin{pgfonlayer}{main}
		
		\node (C) [] {};
		
		\node (C1) [v:ghost, position=180:42mm from C] {};
		
		\node (C2) [v:ghost, position=0:0mm from C] {};
		
		\node (C3) [v:ghost, position=0:42mm from C] {};

		
		
		\node (v1) [v:main,position=0:0mm from C1] {};
		\node (v2) [v:main,position=80:17mm from v1] {};
		\node (v3) [v:main,position=130:12mm from v1] {};
		\node (v4) [v:main,position=190:12mm from v1] {};
		\node (v5) [v:main,position=240:17mm from v1] {};
		\node (v6) [v:main,position=15:14mm from v1] {};
		\node (v7) [v:main,position=340:10mm from v6] {};
		\node (v8) [v:main,position=70:12mm from v6] {};
		\node (v9) [v:main,position=255:17mm from v6] {};
		\node (v10) [v:main,position=300:18mm from v6] {};
		
		\node (gp1) [v:ghost,position=90:5mm from v6] {};
		\node (gp2) [v:ghost,position=165:4mm from v6] {};
		\node (gp3) [v:ghost,position=260:5mm from v6] {};
		\node (gp4) [v:ghost,position=250:5mm from v7] {};
		\node (gp5) [v:ghost,position=0:4mm from v7] {};
		\node (gp6) [v:ghost,position=100:5mm from v7] {};
		\node (Lp) [v:ghost,position=140:4mm from gp2,font=\scriptsize] {\textcolor{magenta}{$M_1$}};
		
		\node (gg1) [v:ghost,position=70:3.5mm from v8] {};
		\node (gg2) [v:ghost,position=165:5mm from v8] {};
		\node (gg3) [v:ghost,position=270:3.7mm from v8] {};
		\node (gg4) [v:ghost,position=45:9mm from v7] {};
		\node (gg5) [v:ghost,position=0:7.5mm from v7] {};
		\node (gg6) [v:ghost,position=270:8mm from v7] {};
		\node (gg7) [v:ghost,position=25:9mm from v9] {};
		\node (gg8) [v:ghost,position=165:4.5mm from v9] {};
		\node (gg9) [v:ghost,position=300:4mm from v9] {};
		\node (gg10) [v:ghost,position=325:7mm from v9] {};
		\node (gg11) [v:ghost,position=270:4mm from v10] {};
		\node (gg12) [v:ghost,position=0:5mm from v10] {};
		\node (gg13) [v:ghost,position=0:10mm from v7] {};
		\node (gg14) [v:ghost,position=0:10mm from v8] {};
		\node (Lg) [v:ghost,position=12:15mm from v7,font=\scriptsize] {\textcolor{BritishRacingGreen}{$M_2$}};
		
		\node (gy1) [v:ghost,position=120:5.5mm from v1] {};
		\node (gy2) [v:ghost,position=180:5.5mm from v1] {};
		\node (gy3) [v:ghost,position=240:5.5mm from v1] {};
		\node (gy4) [v:ghost,position=300:5.5mm from v1] {};
		\node (gy5) [v:ghost,position=0:5.5mm from v1] {};
		\node (gy6) [v:ghost,position=60:5.5mm from v1] {};
		\node (Ly) [v:ghost,position=270:10mm from v1,font=\scriptsize] {\textcolor{Amber!80!black}{$M_3$}};
		
		\node (gb1) [v:ghost,position=330:4mm from v2] {};
		\node (gb2) [v:ghost,position=45:5mm from v2] {};
		\node (gb3) [v:ghost,position=120:4mm from v2] {};
		\node (gb4) [v:ghost,position=90:6mm from v3] {};
		\node (gb5) [v:ghost,position=180:5mm from v3] {};
		\node (gb6) [v:ghost,position=135:6mm from v4] {};
		\node (gb7) [v:ghost,position=225:6.5mm from v4] {};
		\node (gb8) [v:ghost,position=180:5.5mm from v5] {};
		\node (gb9) [v:ghost,position=280:4mm from v5] {};
		\node (gb10) [v:ghost,position=10:5mm from v5] {};
		\node (gb11) [v:ghost,position=180:7mm from v1] {};
		\node (gb12) [v:ghost,position=135:7mm from v1] {};
		\node (gb13) [v:ghost,position=90:9mm from v1] {};
		\node (Lb) [v:ghost,position=130:8.5mm from v3,font=\scriptsize] {\textcolor{CornflowerBlue!85!black}{$M_4$}};
		
		

		
		
		\node (u1) [v:main,position=0:0mm from C3,fill=Amber!90!black] {};
		\node (u2) [v:main,position=0:12mm from u1,fill=BritishRacingGreen] {};
		\node (u3) [v:main,position=270:12mm from u1,fill=magenta] {};
		\node (u4) [v:main,position=135:12mm from u1,fill=CornflowerBlue] {};

		\node (c4) [v:ghost,position=155:11mm from u4] {};
		
		\node(uv2) [v:small,position=80:10mm from c4] {};
		\node(uv3) [v:small,position=135:8mm from c4] {};
		\node(uv4) [v:small,position=190:8mm from c4] {};
		\node(uv5) [v:small,position=235:10mm from c4] {};
		
		\node (ugb1) [v:ghost,position=330:3mm from uv2] {};
		\node (ugb2) [v:ghost,position=45:4mm from uv2] {};
		\node (ugb3) [v:ghost,position=120:3mm from uv2] {};
		\node (ugb4) [v:ghost,position=90:3.5mm from uv3] {};
		\node (ugb5) [v:ghost,position=180:3mm from uv3] {};
		\node (ugb6) [v:ghost,position=135:4mm from uv4] {};
		\node (ugb7) [v:ghost,position=225:4.5mm from uv4] {};
		\node (ugb8) [v:ghost,position=180:4.5mm from uv5] {};
		\node (ugb9) [v:ghost,position=280:3mm from uv5] {};
		\node (ugb10) [v:ghost,position=10:4mm from uv5] {};
		\node (ugb11) [v:ghost,position=180:3.5mm from c4] {};
		\node (ugb12) [v:ghost,position=135:3.5mm from c4] {};
		\node (ugb13) [v:ghost,position=90:4.5mm from c4] {};
		
		\node (gugb2) [v:ghost,position=290:2mm from ugb2] {};
		\node (gugb9) [v:ghost,position=310:3.1mm from uv5] {};
		
		\node (c3) [v:ghost,position=45:12mm from u1] {};
		
		\node(uv1) [v:small,position=180:4mm from c3] {};
		
		\node (ugy1) [v:ghost,position=120:3.5mm from uv1] {};
		\node (ugy2) [v:ghost,position=180:3.5mm from uv1] {};
		\node (ugy3) [v:ghost,position=240:3.5mm from uv1] {};
		\node (ugy4) [v:ghost,position=300:3.5mm from uv1] {};
		\node (ugy5) [v:ghost,position=0:3.5mm from uv1] {};
		\node (ugy6) [v:ghost,position=60:3.5mm from uv1] {};
		
		\node (guy6) [v:ghost,position=180:3.4mm from uv1] {};
		\node (guy4) [v:ghost,position=300:3.4mm from uv1] {};
		
		\node (c1) [v:ghost,position=170:11mm from u3] {};
		
		\node(uv6) [v:small,position=160:3mm from c1] {};
		\node(uv7) [v:small,position=340:3mm from c1] {};
		
		\node (ugp1) [v:ghost,position=90:4mm from uv6] {};
		\node (ugp2) [v:ghost,position=165:3mm from uv6] {};
		\node (ugp3) [v:ghost,position=260:4mm from uv6] {};
		\node (ugp4) [v:ghost,position=250:4mm from uv7] {};
		\node (ugp5) [v:ghost,position=0:3mm from uv7] {};
		\node (ugp6) [v:ghost,position=100:4mm from uv7] {};
		
		\node (gup1) [v:ghost,position=45:2.8mm from uv7] {};
		\node (gup2) [v:ghost,position=280:3.5mm from uv7] {};
		
		\node (c2) [v:ghost,position=0:12mm from u2] {};
		
		\node(uv8) [v:small,position=80:5mm from c2] {};
		\node(uv9) [v:small,position=240:5mm from c2] {};
		\node(uv10) [v:small,position=300:5mm from c2] {};
		
		\node (ugg1) [v:ghost,position=90:2.5mm from uv8] {};
		\node (ugg2) [v:ghost,position=180:3mm from uv8] {};
		\node (ugg3) [v:ghost,position=230:3mm from uv8] {};
		\node (ugg4) [v:ghost,position=70:2mm from c2] {};
		\node (ugg6) [v:ghost,position=300:2mm from c2] {};
		\node (ugg7) [v:ghost,position=45:3mm from uv9] {};
		\node (ugg8) [v:ghost,position=165:3.5mm from uv9] {};
		\node (ugg9) [v:ghost,position=300:3mm from uv9] {};
		\node (ugg10) [v:ghost,position=325:5.05mm from uv9] {};
		\node (ugg11) [v:ghost,position=270:2.8mm from uv10] {};
		\node (ugg12) [v:ghost,position=0:3.7mm from uv10] {};
		\node (ugg13) [v:ghost,position=0:7mm from c2] {};
		\node (ugg14) [v:ghost,position=0:5.5mm from uv8] {};
		
		\node (gug1) [v:ghost,position=145:2.6mm from uv8] {};
		\node (gug2) [v:ghost,position=200:3mm from uv9] {};
		
		\node (Lu1) [v:ghost,position=200:4mm from u1,font=\scriptsize] {$v_3$};
		\node (Lu2) [v:ghost,position=280:5mm from u2,font=\scriptsize] {$v_2$};
		\node (Lu3) [v:ghost,position=320:4.4mm from u3,font=\scriptsize] {$v_1$};
		\node (Lu4) [v:ghost,position=55:4mm from u4,font=\scriptsize] {$v_4$};
		
		

		
		
		\draw (v1) [e:main,->,bend left=15] to (v2);
		\draw (v1) [e:main,->,bend left=15] to (v3);
		\draw (v1) [e:main,->,bend left=15] to (v4);
		\draw (v1) [e:main,->,bend left=15] to (v5);
		\draw (v2) [e:main,->,bend left=15] to (v1);
		\draw (v3) [e:main,->,bend left=15] to (v1);
		\draw (v4) [e:main,->,bend left=15] to (v1);
		\draw (v5) [e:main,->,bend left=15] to (v1);
		
		\draw (v1) [e:main,->,bend left=48] to (v8);
		\draw (v1) [e:main,->,bend right=30] to (v9);
		\draw (v1) [e:main,->,bend right=15] to (v10);
		
		\draw (v2) [e:main,->,bend right=10] to (v3);
		\draw (v3) [e:main,->] to (v4);
		\draw (v3) [e:main,->,bend right=55] to (v5);
		\draw (v4) [e:main,->,bend right=5] to (v2);
		\draw (v5) [e:main,->] to (v4);
		
		\draw (v6) [e:main,->,bend left=15] to (v7);
		\draw (v7) [e:main,->,bend left=15] to (v6);
		
		\draw (v9) [e:main,->,bend right=15] to (v10);
		\draw (v10) [e:main,->,bend right=55] to (v8);
		
		\draw (v8) [e:main,->] to (v6);
		\draw (v8) [e:main,->] to (v7);
		\draw (v9) [e:main,->] to (v6);
		\draw (v9) [e:main,->] to (v7);
		\draw (v10) [e:main,->] to (v6);
		\draw (v10) [e:main,->] to (v7);
		
		\draw (v6) [e:main,->,bend right=10] to (v1);
		\draw (v7) [e:main,->,bend left=22] to (v1);
		
		\draw (gp2) [e:coloredthin,color=magenta,closed,curve through={(gp3) (gp4) (gp5) (gp6)}] to (gp1);
		
		\draw (gg1) [e:coloredthin,color=BritishRacingGreen,closed,curve through={(gg2) (gg3) (gg4) (gg5) (gg6) (gg7) (gg8) (gg9) (gg10) (gg11) (gg12) (gg13)}] to (gg14);
		
		\draw (gy2) [e:coloredthin,color=Amber!90!black,closed,curve through={(gy3) (gy4) (gy5) (gy6)}] to (gy1);
		
		\draw (gb1) [e:coloredthin,color=CornflowerBlue,closed,curve through={(gb2) (gb3) (gb4) (gb5) (gb6) (gb7) (gb8) (gb9) (gb10) (gb11) (gb12) }] to (gb13);
		
		

		
		
		\draw (u1) [e:main,->,bend left=10] to (u2);
		\draw (u1) [e:main,->,bend left=20] to (u4);
		\draw (u4) [e:main,->,bend left=20] to (u1);
		\draw (u2) [e:main,->,bend left=10] to (u3);
		\draw (u3) [e:main,->,bend left=10] to (u1);
		
		
		
		\end{pgfonlayer}
		

		\begin{pgfonlayer}{background}
		
		\draw (uv2) [e:small,->,bend right=10] to (uv3);
		\draw (uv3) [e:small,->] to (uv4);
		\draw (uv3) [e:small,->,bend right=55] to (uv5);
		\draw (uv4) [e:small,->,bend right=5] to (uv2);
		\draw (uv5) [e:small,->] to (uv4);
		
		\draw (uv6) [e:small,->,bend left=25] to (uv7);
		\draw (uv7) [e:small,->,bend left=25] to (uv6);
		
		\draw (uv9) [e:small,->,bend right=15] to (uv10);
		\draw (uv10) [e:small,->,bend right=55] to (uv8);
		
		\draw (ugb1) [line width=1.15pt,opacity=0.55,color=CornflowerBlue,closed,curve through={(ugb2) (ugb3) (ugb4) (ugb5) (ugb6) (ugb7) (ugb8) (ugb9) (ugb10) (ugb11) (ugb12) }] to (ugb13);
		\draw (gugb2) [line width=1.15pt,opacity=0.55,color=CornflowerBlue,densely dashed,cap=round] to (u4);
		\draw (gugb9) [line width=1.15pt,opacity=0.55,color=CornflowerBlue,densely dashed,cap=round] to (u4);
		
		\draw (ugy2) [line width=1.15pt,opacity=0.55,color=Amber!90!black,closed,curve through={(ugy3) (ugy4) (ugy5) (ugy6)}] to (ugy1);
		\draw (guy4) [line width=1.15pt,opacity=0.55,color=Amber!90!black,densely dashed,cap=round] to (u1);
		\draw (guy6) [line width=1.15pt,opacity=0.55,color=Amber!90!black,densely dashed,cap=round] to (u1);
		
		\draw (ugp2) [line width=1.15pt,opacity=0.55,color=magenta,closed,curve through={(ugp3) (ugp4) (ugp5) (ugp6)}] to (ugp1);
		\draw (gup1) [line width=1.15pt,opacity=0.55,color=magenta,densely dashed,cap=round] to (u3);
		\draw (gup2) [line width=1.15pt,opacity=0.55,color=magenta,densely dashed,cap=round] to (u3);
		
		\draw (ugg1) [line width=1.15pt,opacity=0.55,color=BritishRacingGreen,closed,curve through={(ugg2) (ugg3) (ugg4) (ugg6) (ugg7) (ugg8) (ugg9) (ugg10) (ugg11) (ugg12) (ugg13)}] to (ugg14);
		\draw (gug1) [line width=1.15pt,opacity=0.55,color=BritishRacingGreen,densely dashed,cap=round] to (u2);
		\draw (gug2) [line width=1.15pt,opacity=0.55,color=BritishRacingGreen,densely dashed,cap=round] to (u2);
		
		\end{pgfonlayer}	
		
		\begin{pgfonlayer}{foreground}

		\end{pgfonlayer}
		\end{tikzpicture}
	\end{center}
	\caption{A digraph $D$ together with a decomposition into modules (left) and the corresponding module-digraph (right) together with the modules $M_i$ represented by vertices $v_i$.}
	\label{fig:moduledecomp}
\end{figure}

While the directed modular width can increase when taking subdigraphs, it is well-behaved with respect to induced subdigraphs. More precisely, we have the following statement.
\begin{fact} \label{inducedmonotonicity}
	Let $D$ be a digraph, and let $D'$ be an induced subdigraph of $D$. Then $$\dmw(D') \leq \dmw(D).$$
\end{fact}
\begin{proof}
	We prove the statement by induction on the number of vertices of $D$. The statement clearly holds true when $|V(D)|=1$, so suppose for the inductive step that $|V(D)|=n \ge 2$, and that the statement is true for all digraphs on less than $n$ vertices (and for all their induced subdigraphs).
	
	Let $D'=D[X]$ with $X \subseteq V(D)$ be a given induced subdigraph of $D$, and let $\omega\coloneqq \dmw(D), \omega'\coloneqq \dmw(D')$. If $|V(D)| \leq \omega$, then clearly, we also have $\omega' \leq |V(D')| \leq |V(D)| \leq \omega$, proving the claim. Otherwise, let $M_1,\ldots,M_\ell$ denote a partition of $V(D)$ into modules such that $2 \leq \ell \leq \omega$ and such that $\dmw(D[M_i]) \leq \omega$ for all $i \in [\ell]$. For every $i \in [\ell]$, $D'[X \cap M_i]=D[X \cap M_i]$ is an induced subdigraph of $D[M_i]$, and since $|V(D[M_i])|<|V(D)|$, the induction hypothesis tells us that $\dmw(D'[X \cap M_i]) \leq \dmw(D[M_i]) \leq \omega$ for all $i \in [\ell]$. Clearly, $X \cap M_i$ defines a module in $D'$ for each $i \in [\ell]$. Because $\ell \leq \omega$, the definition of directed modular width now implies
	$$\omega'=\dmw(D') \leq \max \Set{\ell,\dmw(D'[X \cap M_1]),\ldots,\dmw(D'[X \cap M_\ell])} \leq \omega,$$ which yields the claim also in this case.
\end{proof}
Given a digraph $D$ and a partition $M_1,\ldots,M_\ell$ of $V(D)$ into modules, we will frequently use $D_M$ to denote the \emph{module-digraph} of $D$ corresponding to the module-decomposition $\Set{M_1,\ldots,M_\ell}$: $D_M$ is obtained from $D$ by identifying $M_i, i \in [\ell]$ each into a single vertex $v_i \in V(D_M)$ and deleting parallel directed edges afterwards. Equivalently, an edge $(v_i,v_j)$ lies in $E(D_M)$ if and only if in $D$, there is at least one directed edge starting in $M_i$ and ending in $M_j$. Due to the modular property, this is equivalent to the fact that $(u,w) \in E(D)$ for all $u \in M_i, w \in M_j$.
For an example of a module-digraph see \cref{fig:moduledecomp}.

Throughout the paper, given a module-decomposition $M_1,\ldots,M_\ell$ of a directed graph $D$, we will denote by $\eta:V(D) \rightarrow V(D_M)$ the mapping defined by $\eta(z)\coloneqq v_k$ for all $z \in M_k$.

\paragraph{Classes of Bounded Directed Modular Width}

In the introduction we mentioned a result of Ganian et.\@ al.\@ \cite{ganian2016there} stating that an algorithmically powerful directed graph width measure cannot be closed under subgraphs or topological butterfly minors.
In the light of \Cref{inducedmonotonicity} we want to make the argument, that digraphs of bounded modular width still can be interesting from a structural point of view.
To do this we introduce a class of digraphs that recently got some attention: \emph{directed co-graphs}.

Directed co-graphs are defined recursively via three operations as seen in \cite{crespelle2006fully}.
In the following let $D_1,\dots,D_k$ be $k$ pairwise disjoint digraphs.
\begin{itemize}
\item The \emph{disjoint union} of $D_1,\dots D_k$, denoted by $D_1\oplus\dots\oplus D_k$, is the digraph with vertex set $\bigcup_{i=1}^k\Fkt{V}{D_i}$ and edge set $\bigcup_{i=1}^k\Fkt{E}{D_i}$.
\item The \emph{series composition} of $D_1,\dots,D_k$, denoted by $D_1\otimes\dots\otimes D_k$, is the digraph obtained by their disjoined union together with all possible edges between $D_i$ and $D_j$ for all $1\leq i<j\leq k$.
\item The \emph{order composition} of $D_1,\dots,D_k$, denoted by $D_1\oslash\dots\oslash D_k$, is the digraph obtained by their disjoined union together with all possible edges with tail in $D_i$ and head in $D_j$ for all $1\leq i<j\leq k$.
\end{itemize}

\begin{definition}
The class of directed co-graphs is recursively defined as follows.
\begin{enumerate}
	\item Every digraph consisting of a single vertex is a directed co-graph.
	\item If $D_1,\dots,D_k$ are directed co-graphs, then their disjoint union, their series composition and their order composition are directed co-graphs.
\end{enumerate}
\end{definition}

From a structural point of view the class of directed co-graphs is interesting because it can be characterised by a finite set of forbidden induced subgraph \cite{crespelle2006fully}.
In particular the set of forbidden induced subgraphs is closed under taking the complement and thus the complement of a directed co-graph is again a directed co-graph.
Also algorithmically directed co-graphs have seen some attention as they belong to the class of graphs of directed clique-width at most $2$ \cite{gurski2016directed} and their directed pathwidth as well as directed treewidth, though unbounded, can be computed in polynomial time \cite{gurski2018directed}.
Additionally, for every $k$ there exists a polynomial time algorithm for the \emph{weak $k$-disjoint path} problem \cite{bang2014arc} which is a relaxed version of the $k$-disjoined path problem we will discuss in \cref{vddp}.

Many of the above mentioned algorithms for directed co-graphs already utilise a module decomposition. If seen in this light, the results of our paper are generalisations of the algorithmic results already known especially for co-graphs.
For the sake of completeness, we provide a proof that directed co-graphs are exactly the digraphs of directed modular width at most $2$.

\begin{theorem}\label{thm:cographs}
Let $D$ be a digraph, then $\Fkt{\dmw}{D}\leq 2$ if and only if $D$ is a directed co-graph.
\end{theorem}

\begin{proof}
Every digraph on at most $2$ vertices has directed modular width at most $2$ and also is a directed co-graph.

Now let $D$ be a directed co-graph on more than $2$ vertices and let $D=D_1\circ\dots\circ D_k$ for $k\geq 2$ and $\circ\in\Set{\oplus,\otimes,\oslash}$  where $D_i$ is a directed co-graph for every $i\in[k]$.
Then $D'\coloneqq D_1\circ\dots\circ D_{k-1}$ is also a directed co-graph and $\Abs{\Fkt{V}{D'}}<\Abs{\Fkt{V}{D}}$ as well as $\Abs{\Fkt{V}{D_k}}<\Abs{\Fkt{V}{D}}$.
By induction now both $D_k$ and $D'$ are graphs of directed modular width at most $2$.
Moreover ${\Fkt{V}{D'}, \Fkt{V}{D_k}}$ is a partition of $\Fkt{V}{D}$ into two modules, both inducing subgraphs of directed modular width at most $2$.
Hence $D$ has directed modular width at most $2$.

For the reverse direction let $D$ be a digraph of directed modular width at most $2$ on at least $3$ vertices and let $\Set{M_1,M_2}$ be a decomposition of $\Fkt{V}{D}$ into non-empty modules such that $\Fkt{\dmw}{\InducedSubgraph{D}{M_i}}\leq 2$ for $i\in[2]$.
Since both modules are non-empty we may assume by induction that $D_i\coloneqq\InducedSubgraph{D}{M_i}$ is a directed co-graph for both $i$.
With $M_1$ and $M_2$ being modules of $D$, $D$ must be one of the following four graphs:
$D_1\oplus D_2$, $D_1\otimes D_2$, $D_1\oslash D_2$, or $D_2\oslash D_1$.
In all four cases $D$ can be obtained from $D_1$ and $D_2$ by one of the operations used in the definition of directed co-graphs and thus $D$ is a directed co-graph.
\end{proof}

\paragraph{Directed Modular Width and Directed Clique-Width}

When studying width parameters one usually is interested in their comparability to other parameters.
In this paragraph we discuss a bit about the place of directed modular width within the hierarchy of directed width measures.
In the undirected case, clique-width is a lower bound on the modular width and our goal is to establish the same relation for the two directed versions of these parameters.

For a digraph $D$ and a function $\operatorname{lab}\colon\Fkt{V}{D}\rightarrow\Set{1,\dots,k}$, the triple $\Brace{\Fkt{V}{D},\Fkt{E}{D},\operatorname{lab}}$ is called a \emph{$k$-labelled digraph}.
The function $\operatorname{lab}$ is called a \emph{labelling} of $D$, and for each $v\in\Fkt{V}{D}$, $\Fkt{\operatorname{lab}}{v}$ is called its \emph{label}.

\begin{definition}
For a positive integer, the class $CLW_k$ of $k$-labelled digraphs is recursively defined as follows.
\begin{enumerate}
	
	\item The digraph on a single vertex $v$ with label $i\in[k]$ is in $CLW_k$.
	
	\item Let $D_1=\Brace{V_1,E_1,\operatorname{lab}_1}\in CLW_k$ and $D_2=\Brace{V_2,E_2,\operatorname{lab}_2}\in CLW_k$ be two $k$-labelled digraphs on disjoint vertex sets.
	Let $D_1\oplus D_2$ be the disjoint union of $\Brace{V_1,E_1}$ with $\Brace{V_2,E_2}$ together with the labelling $\operatorname{lab}$ such that for all $v\in V_1\cup V_2$:
	\begin{align*}
		\Fkt{\operatorname{lab}}{v}\coloneqq\TwoCases{\Fkt{\operatorname{lab}_1}{v}}{v\in V_1}{\Fkt{\operatorname{lab}_2}{v}}{v\in V_2.}
	\end{align*}
	Then $D_1\oplus D_2\in CLW_k$.
	
	\item Let $D=\Brace{V,E,\operatorname{lab}}\in CLW_k$ be a $k$-labelled digraph and $i,j\in[k]$ be two distinct integers.
	Let $\Fkt{\rho_{i\rightarrow j}}{D}=\Brace{V,E,\operatorname{lab}'}$ where
	\begin{align*}
		\Fkt{\operatorname{lab}'}{v}\coloneqq\TwoCases{\Fkt{\operatorname{lab}}{v}}{\Fkt{\operatorname{lab}}{v}\neq i}{j}{\Fkt{\operatorname{lab}}{v}=i.}
	\end{align*}
	For every $v\in V$.
	Then $\Fkt{\rho_{i\rightarrow j}}{D}\in CLW_k$.
	
	\item Let $D=\Brace{V,E,\operatorname{lab}}\in CLW_k$ be a $k$-labelled digraph, and $i,j\in[k]$ be two distinct integers.
Let $\Fkt{\alpha_{i,j}}{D}$ be the digraph with labelling $\operatorname{lab}$ obtained from $D$ by adding all edges $\Brace{a,b}$ where $\Fkt{\operatorname{lab}}{a}=i$ and $\Fkt{\operatorname{lab}}{b}=j$.
Then $\Fkt{\alpha_{i,j}}{D}\in CLW_k$.
	
\end{enumerate}
The \emph{directed clique width} of a digraph $D$, denoted by $\Fkt{\operatorname{dcw}}{D}$, is the minimum integer $k$ such that there is a $k$-labelling $\operatorname{lab}$ of $D$ where $\Brace{\Fkt{V}{D},\Fkt{E}{D},\operatorname{lab}}\in CLW_k$.
\emph{Directed clique-width $k$-expressions} are expressions which recursively construct a digraph with the four operations defined above.
\end{definition}

Directed acyclic graphs can have arbitrary high directed clique width and thus there are classes of digraphs where directed treewidth and related parameters are bounded, but the directed clique width is not.
Moreover, bidirected complete graphs, i.e.\@ undirected complete graphs where we replace every undirected edge by a digon, are directed co-graphs and thus have bounded directed clique-width, however the directed treewidth and its cousins are unbounded on these graphs.
Hence directed treewidth and directed clique-width are incomparable.
However, it is straight forward to construct a directed clique width $\omega$-expression for digraphs of directed modular width $\omega$, which implies that directed clique-width acts as a lower bound on the directed modular width, just as in the undirected case.
For the sake of completeness we present a proof of this claim.

\begin{theorem}\label{thm:clwleqdmw}
Let $D$ be a directed graph, then $\Fkt{\operatorname{dcw}}{D}\leq\Fkt{\dmw}{D}$.
\end{theorem}

\begin{proof}
	First of all note that any digraph $D$  on $\omega$ vertices has directed clique width at most $\omega$, since we may assign a unique label to every vertex of $D$.
	
	Now let $D$ be a directed graph of directed modular width $\omega$ and let $\mathcal{M}=\Set{M_1,\dots,M_{\omega}}$ be a partitioning of $\Fkt{V}{D}$ into modules such that $\Fkt{\dmw}{\InducedSubgraph{D}{M_i}}\leq\omega$ for all $i\in[\omega]$.
	By induction we may assume that $\Fkt{\operatorname{dcw}}{\InducedSubgraph{D}{M_i}}\leq\omega$ for all $i\in[\omega]$ holds as well.
	We now show that we can construct $D$ from the digraphs induced by its modules using the operations from the definition of directed clique width.
	Suppose for every $i\in[\omega]$ that $D_i\coloneqq\Brace{M_i,\Fkt{E}{\InducedSubgraph{D}{M_i}},\operatorname{lab}_i}$ is a $\omega$-labelled digraph where $\operatorname{lab}_i\colon M_i\rightarrow[\omega]$.
	First for every $i\in[\omega]$ and every $j\in[\omega]\setminus\Set{i}$ apply the $\rho_{j\rightarrow i}$-operation to $D_i$ and denote the result in which all vertices of $D_i$ have label $i$ by $D_i'$.
	Then let $D'\coloneqq D_1'\oplus\dots\oplus D_{\omega}'$.
	Now consider the module digraph $D_M$ corresponding to $\mathcal{M}$.
	For every $\Brace{v_i,v_j}\in\Fkt{E}{D_M}$, $i,j\in[\omega]$ apply the $\alpha_{i,j}$ operation to $D'$.
	After all of those operations have been applied, since $\mathcal{M}$ is a module decomposition for $D$, we have constructed $D$ from the $D_i$ and thus $\Fkt{\operatorname{clw}}{D}\leq\omega$.
\end{proof}

\paragraph{Computing a non-trivial module-decomposition of a digraph.}
The most important tool, which is involved in all the algorithms proposed in this paper, is the ability to find a non-trivial decomposition of the vertex set of a given digraph into modules, in polynomial time. In fact, this task can be executed in a much stronger form. In \cite{moduledecomposition}, it was shown that a so-called \emph{canonical module-decomposition} of a given digraph can be obtained in linear time. For us, the following weaker form of their result will be sufficient.
\begin{theorem}[\cite{moduledecomposition}] \label{computedecomposition}
	There is an algorithm that, given a digraph $D$ on at least two vertices as input, returns a decomposition of $V(D)$ into $\ell \in \{2,\ldots,\dmw(D)\}$ modules. This algorithm runs in time $\mathcal{O}(n+m)$, where $n\coloneqq |V(D)|$ and $m\coloneqq |E(D)|$. 
\end{theorem}
To do the runtime-analysis of our algorithms, we will often use a rooted model-tree $T$ which resembles the structure of recursive calls in our algorithms. Every vertex $q \in V(T)$ has either no children or at least two. It furthermore admits a labelling of its vertices of the following kind: 

The root of the tree is labelled with a finite ground set $\Omega$ (in our case the vertex set of the considered digraph). Every other vertex $q \in V(T)$ is labelled with a subset $\emptyset \neq \Omega(q) \subseteq \Omega$, and for every branching vertex, the associated subset is the disjoint union of the subsets associated to its children. Finally, the leafs of the tree are labelled with the singletons $\{v\}, v \in \Omega$. A tree which admits a labelling of this type will be called a \emph{decomposition tree}. 
\begin{fact} \label{dectree}
	If $T$ is a decomposition tree with ground-set $\Omega$, then $|V(T)| \leq 2|\Omega|-1$.
\end{fact}
\begin{proof}
	First note that we can reduce to the case where $T$ is a rooted binary tree: If there is a branching vertex $q \in V(T)$ with $b \ge 3$ children $q_1,\ldots, q_b$, we can locally replace this branching by a binary tree with $b$ leafs, where instead of directly splitting $\Omega(q)$ into $\Omega(q_1),\ldots,\Omega(q_b)$, we first split-off $\Omega(q_1)$, then $\Omega(q_2)$, and so on. Clearly, successive application of this operation to every branching with more than two children yields a binary decomposition-tree $T'$ with ground set $\Omega$ and $|V(T')| \ge |V(T)|$. 
	
	Now if $T$ is a binary-tree, because the leafs of $T$ are labelled by the singletons of $\Omega$, $T$ has $|\Omega|$ leafs and therefore $2|\Omega|-1$ vertices.
\end{proof}
\paragraph{Integer programming with bounded number of variables.}
In the design of our algorithms, we frequently reduce the treated algorithmic problem to the same problem on an input digraph with a bounded number of vertices, but possibly equipped with additional information (such as weightings) of polynomial-size in the original input. In many cases, we will then make use of an integer program reformulation of the problem, in which we have a bounded number of constraints and variables, but possibly entries in the input matrices and vectors of polynomial size. We will therefore make use of the following powerful tool from the theory of Integer Programming, which shows that the feasibility of a given ILP can be decided in polynomial time when using the number of variables $p$ as a parameter.

\begin{theorem}[\cite{ipsolving}, Theorem 1]\label{ipfeasibility}
	There exists an algorithm that, given as input a matrix $A \in \mathbb{Z}^{n \times p}$ and a vector $b \in \mathbb{Z}^n$, decides whether there is a feasible solution to 
	$$Ax \geq b, x \in \mathbb{Z}^p$$ (and returns a solution if applicable) in time $\mathcal{O}(p^{2.5p+o(p)}L)$ where $L$ denotes the coding length of the input $(A,b)$. 
\end{theorem}
Solving an ILP can be easily reduced to checking the feasibility of several ILP-s using binary search. 
\begin{corollary}[\cite{ipsolving}, Theorem 12]\label{ipsolving}
	There exists an algorithm that, given as input a matrix $A \in \mathbb{Z}^{n \times p}$, vectors $c \in \mathbb{Z}^{p}, b \in \mathbb{Z}^n$, and some $U_1, U_2 \in \mathbb{Z}_+$, tests feasibility and if applicable outputs an optimal solution of the ILP
	\begin{align}
		\min c^Tx\\
		\text{  subj.\ to }&
		Ax \geq b, x \in \mathbb{Z}^p
	\end{align}
	in time $\mathcal{O}(p^{2.5p+o(p)}L \log (U_1U_2))$ where $L$ denotes the coding length of the input $(A,b,c)$. Here we assume that the optimal value of the program lies within $[-U_1,U_1]$ and that $U_2$ is an upper bound on the largest absolute value any entry in an optimal solution vector can take.
\end{corollary}

\section{Strategy}\label{sec:strategy}
Most of the FPT-algorithms presented in the following sections are based on a common general strategy, which shall be outlined in the following. 
\begin{itemize}
	\item In most cases, we consider a well-chosen generalisation of the original problem we want to solve. This often involves additional inputs, such as integer weights or capacities on the vertices. Although such a generalisation is not always necessary, it often acts naturally within the context of module-decompositions. Our methods allow us to also handle these more general settings.
	\item We derive auxiliary theoretical results, that deal with a given module-decomposition of a digraph and describe how the studied parameters or objects which shall be computed on the whole digraph interact with corresponding objects on the digraphs induced by the modules as well as the module-digraph. These theoretical results are at the core of the construction of these algorithms.
	\item We construct an algorithm that, given solutions to the considered problem on the modules and the module-digraph, constructs a solution to the problem for the whole digraph in polynomial time.
	\item To solve the problem on the module-digraph, we make use of the fact that for bounded directed modular width, we can bound the number of vertices of the module-digraph by a constant. We then reformulate the problem on the module-digraph as an Integer Linear Program, which has bounded number of variables. However, the additional inputs such as weights on the vertices may still have polynomial size. We then make use of \cref{ipsolving} to solve the problem on the module-digraph in FPT-time.
	\item Now we recurse until we end up solving the problem on digraphs consisting of single vertices. Because in each step, we further decompose an induced subdigraph into modules, the size of the recursion-tree, using Fact \ref{dectree}, has linear size in the number of vertices of the input digraph. Because the module-decompositions can be computed in polynomial time in each step, we manage to prove an upper bound on the run-time of the form $\mathcal{O}(f(\omega)p(n)q(\log \tau))$, where $f$ is some function, $p$ and $q$ are polynomials, $\omega$ denotes the directed modular width of the input digraph, $n$ the number of vertices of the input digraph, and $\tau$ bounds the additional information carried in the input (for instance an upper bound on the sum of the weights distributed on the digraph).
\end{itemize}

\section{Feedback Vertex Set}\label{sec:FVS}
In this section, we deal with the famous \emph{Feedback Vertex Set Problem} on directed graphs. A \emph{feedback vertex set} in a digraph $D$ is a subset $F \subseteq V(D)$ of vertices that meets every directed cycle. Equivalently, $D-F$ is an acyclic digraph. The size of a smallest feedback vertex set in $D$, denoted by $\fv(D)$ in the following. Feedback vertex sets play an important role in parametrised algorithmics. Especially for problems which are tractable on acyclic digraphs (resp. forests in the undirected case), small feedback vertex sets often allow for fast parametrisations of NP-hard problems. From a structural point of view, feedback vertex sets are also important, as $\fv(D)$ acts as an upper bound for many important parameters such as the directed treewidth or the cycle-rank. 
It is therefore desirable to find the size of a smallest or a sufficiently small feedback vertex set.  

\ProblemDefLabelled{Minimum Feedback Vertex Set}{FVS}{prob:FVS}
{A digraph $D$.}
{What is the value of $\fv(D)$? Find a feedback vertex set $F \subseteq V(D)$ with $|F|=\fv(D)$.}

The decision version of the above problem was among Karp's famous list of 21 NP-complete problems \cite{Karp1972}. It is known to be NP-complete even for restricted classes such as planar digraphs of maximum in- and out-degree $3$ \cite{garey}. 

On the positive side, the directed feedback vertex set problem is fixed-parameter tractable with respect to $\fv(D)$ itself. In fact, after a period of research concerning parametrisations of the problem on tournaments (see for example \cite{nieder}), it was shown in 2008 by Chen et al. \cite{chen2008fixed} that for digraphs $D$ with $\fv(D) \leq k$, a feedback vertex set of minimal size can be computed in time $\mathcal{O}((1.48k)^k\mathcal{O}(n^c))$, where $n\coloneqq |V(D)|$ and $c$ is some constant. This has motivated research towards a polynomial-size kernel for the problem, see for instance \cite{polynomialkernelFV}. 
Finally, it is known that the directed minimum feedback vertex set problem can be parametrised by the treewidth of the underlying graph \cite{bonamy}. However, the same is not known when using the directed treewidth of the digraph as a parameter instead. 

In the following, we present a simple FPT-algorithm for this problem with fixed parameter $\omega\coloneqq \dmw(D)$. In fact, the algorithm we propose recursively solves the following weighted generalisation of \ref{prob:FVS}. In order to keep control over the running time, we use an additional threshold-parameter $\tau \in \mathbb{N}$ as part of the input, which upper bounds the total weight distributed on the vertices. For a digraph $D$, a vertex-set $X \subseteq V(D)$ and a vertex-weighting $w:V(D) \rightarrow \mathbb{N}_0$, let us define $w(X)\coloneqq \sum_{z \in X}{w(z)}$. 

\ProblemDefLabelled{Minimum Weight Feedback Vertex Set}{wFVS}{prob:wFVS}
{A digraph $D$, a non-negative integer weighting $w:V(D) \rightarrow \mathbb{N}_0$ of the vertices, and a threshold $\tau \in \mathbb{N}$ such that $\sum_{z \in V(D)}{w(z)} \leq \tau$. }
{Find a feedback vertex set $F \subseteq V(D)$ of $D$ that minimises $w(F)$.}

The goal of this section is to prove the following. 
\begin{theorem} \label{FPTWeightedFVS}
	There exists an algorithm that given a digraph $D$, a vertex-weighting $w:V(D) \rightarrow \mathbb{N}_0$, and a corresponding bound $\tau \in \mathbb{N}$ as input, outputs a feedback vertex set of $D$ with minimum total weight. The algorithm runs in time $\mathcal{O}(n^3+\omega^2 2^\omega n^2\log \tau)$, where $n\coloneqq |V(D)|$ and $\omega\coloneqq \dmw(D)$. 
\end{theorem}

Setting $w(z)\coloneqq 1$ for all $z \in V(D)$, we see that \ref{prob:FVS} is a special case of \ref{prob:wFVS}, where we can put $\tau\coloneqq n$.

\begin{corollary}
	There exists an algorithm that, given as input a digraph $D$, outputs $\fv(D)$ and a feedback vertex set $F$ of minimum size in time $\mathcal{O}(n^3+\omega^2 2^\omega n^2\log n)$, where $n\coloneqq |V(D)|$ and $\omega\coloneqq \dmw(D)$. 
\end{corollary}

We prepare the proof of \cref{FPTWeightedFVS} with some auxiliary statements. 
\begin{lemma}
	Let $D$ be a digraph equipped with a vertex-weighting $w:V(D) \rightarrow \mathbb{N}_0$ and a partition $M_1,\ldots,M_\ell$ of the vertex set into modules. Let $D_M$ denote the corresponding module-digraph with vertex set $\{v_1,\ldots,v_\ell\}$. 
	
	Then for any set $F \subseteq V(D)$, the following statements are equivalent:
	\begin{itemize} \label{modulefvsets}
		\item $F$ is a feedback vertex set for $D$.
		\item $F \cap M_i$ is a feedback vertex set for $D[M_i]$ for all $i \in [\ell]$, and 
		$$F_M\coloneqq \CondSet{v_i}{i \in [\ell], M_i \subseteq F}$$ is a feedback vertex set for $D_M$.
	\end{itemize}
\end{lemma}
\begin{proof} 
	$\Longrightarrow$ Assume that $F$ is a feedback vertex set of $D$. Because $F$ intersects $V(C)$ for every directed cycle in $D$, $F \cap M_i$ intersects every directed cycle in $D[M_i]$ for each $i \in [\ell]$ and therefore defines a feedback vertex set. To prove that $F_M$ defines a feedback vertex set in $D_M$, assume towards a contradiction that there was a directed cycle $C$ with vertex sequence $v_{i_1}, v_{i_2}, \ldots,v_{i_m}=v_{i_1}$ in $D_M$ such that $V(C) \cap F_M=\emptyset$. By definition, this means that for every module $M_{i_r}$ with $1 \leq r \leq m-1$ we find a vertex $z_r \in M_{i_r}$ with $z \notin F$. However, by the definition of the module-digraph and the properties of the modules, we have $(z_{i_r},z_{i_{r+1}}) \in E(D)$ for all  $r \in [m-1]$. This implies that the cyclic vertex sequence $z_{i_1},z_{i_2},\ldots,z_{i_m}=z_{i_1}$ defines a directed cycle in $D$ whose vertex set is disjoint from $F$. This contradicts the fact that $F$ is a feedback vertex set. Therefore $F_M$ is a feedback vertex set of $D_M$ as claimed.
	
	$\Longleftarrow$ Assume that $F \cap M_i$ is a feedback vertex set in $D$ for all $i \in [\ell]$ and that $F_M$ is one for $D_M$. Now let $C$ be an arbitrary directed cycle in $D$ with vertex-sequence $u_1,u_2,\ldots,u_m=u_1$. We must show that $V(C) \cap F \neq \emptyset$. If it is completely contained in some module $M_i$, it must contain a vertex of $F \cap M_i \subseteq F$ and we are done. 
	
	Otherwise, consider the cyclical sequence $\eta(u_1),\eta(u_2),\ldots,\eta(u_m)=\eta(u_1)$ of vertices in $D_M$. For each $r \in [m-1]$, we either have $\eta(u_r)=\eta(u_{r+1})$ or $(\eta(u_r), \eta(u_{r+1})) \in E(D_M)$. Therefore, deleting all consecutive multiple occurrences of identical vertices from the sequence, we obtain a closed directed walk in $D_M$ which visits at least two different vertices, and therefore contains the vertex set of a directed cycle $C_M$ in $D_M$. Because $F_M$ was assumed to be a feedback vertex set, we conclude that there is an $r \in [m]$ and some $i \in [\ell]$ such that $v_i=\eta(u_r) \in V(C_M) \cap F_M$. Hence, $u_r \in M_i \subseteq F$, and thus $u_r \in V(C) \cap F$. Since $C$ contains a vertex from $F$, the claim follows.
\end{proof}
From the above we can easily reduce \ref{prob:FVS} on the digraph $D$ to corresponding instances for the digraphs $D[M_1],\ldots,D[M_\ell]$ and $D_M$. For any vertex-weighted digraph $(D,w)$, let us denote by $\fv(D,w)$ the minimum weight of a feedback vertex set for $D$.
\begin{lemma} \label{mfvrec}
	Let $(D,w)$ with $w:V(D) \rightarrow \mathbb{N}_0$ be a vertex-weighted digraph, and let $\Set{M_1,\ldots,M_\ell}$ be a module-partition of $V(D)$. Let $w_M:V(D_M) \rightarrow \mathbb{N}_0$ be defined according to 
	$$w_M(v_i)\coloneqq w(M_i)-\fv(D[M_i],w|_{M_i})$$ for any vertex $v_i$ corresponding to module $M_i$, $i \in [\ell]$. Let $F_M$ be a feedback vertex set in $D_M$ of minimum weight with respect to $w_M$, and for each $i \in [\ell]$, let $F_i$ be a feedback vertex set in $D[M_i]$ of minimum weight with respect to $w|_{M_i}$. Then
	$$F\coloneqq \left(\bigcup_{\substack{i \in [\ell],\cr
			v_i \in F_M}}{M_i}\right) \cup \left(\bigcup_{\substack{i \in [\ell],\cr
			v_i \notin F_M}}{F_i}\right)$$
	
	defines a feedback vertex set in $D$ of minimum weight with respect to $w$, namely
	$$\fv(D,w)=\fv(D_M,w_M)+\sum_{i=1}^{\ell}{\fv(D[M_i],w|_{M_i})}.$$
\end{lemma}
\begin{proof}
	Using \cref{modulefvsets}, it is readily verified that $F$ defines a feedback vertex set of $D$. It is furthermore clear from the definition of $w_M$ that
	\begin{align*}
	w(F)=~&\sum_{\substack{i \in [\ell],\cr
			v_i \in F_M}}{w(M_i)} + \sum_{\substack{i \in [\ell],\cr
			v_i \notin F_M}}{w(F_i)}=\sum_{\substack{i \in [\ell],\cr
			v_i \in F_M}}{\left(w(M_i)-w(F_i)\right)}+\sum_{i=1}^{\ell}{w(F_i)}\\
		=~&w_M(F_M)+\sum_{i=1}^{\ell}{w(F_i)}=\fv(D_M,w_M)+\sum_{i=1}^{\ell}{\fv(D[M_i],w|_{M_i})}.
		\end{align*}
	It therefore remains to show that $w(F') \ge w(F)$ for any other feedback vertex set $F'$ of $D$. So let $F'$ be arbitrary and define $F_i'\coloneqq F' \cap M_i$ for all $i \in [\ell]$ as well as $F'_M\coloneqq \{v_i|i \in [\ell], M_i \subseteq F'\}$. By \cref{modulefvsets}, these are feedback vertex sets and we conclude that $w(F'_i) \ge \fv(D[M_i],w|_{M_i}), i \in [\ell]$ and $w_M(F'_M) \ge \fv(D_M,w_M)$. We therefore have
	\begin{align*}
	w(F') \ge~& \sum_{\substack{i \in [\ell],\cr
			v_i \in F'_M}}{w(M_i)}+\sum_{\substack{i \in [\ell],\cr
			v_i \notin F'_M}}{\fv(D[M_i],w|_{M_i})}\\
			=~&\sum_{\substack{i \in [\ell],\cr
			v_i \in F'_M}}{\left(w(M_i)-\fv(D[M_i],w|_{M_i})\right)}+\sum_{i=1}^{\ell}{\fv(D[M_i],w|_{M_i})}\\
	=~&w_M(F'_M)+\sum_{i=1}^{\ell}{\fv(D[M_i],w|_{M_i})}\\ \ge~& \fv(D_M,w_M)+\sum_{i=1}^{\ell}{\fv(D[M_i],w|_{M_i})}.
	\end{align*}
	This verifies the minimality of $F$ and proves the claim.
\end{proof}
In order to compute a minimum weight feedback vertex set on the module digraph, we need the following.
\begin{lemma} \label{obvious1}
	Given a digraph $D$ on at most $\omega$ vertices, a vertex-weighting $w:V(D) \rightarrow \mathbb{N}_0$ and some $\tau \in \mathbb{N}$ such that $\sum_{z \in V(D)}{w(z)} \leq \tau$, $\fv(D,w)$ and a feedback vertex set of minimum weight can be computed in time $\mathcal{O}(\omega^2 2^\omega\log \tau)$.
\end{lemma}
\begin{proof}
	We use a simple brute-force approach. We first enumerate all at most $2^\omega$ feedback vertex sets of $D$. For this, we go trough the subsets $X \subseteq V(D)$ and test whether $D-X$ is acyclic, which can be done in time $\mathcal{O}(|E(D-X)|) \leq \mathcal{O}(\omega^2)$. For each feedback vertex set $F$, we compute $w(F)$ using at most $|V(F)|-1 \leq \omega$ arithmetic operations with pairs of numbers of size at most $\tau$. In the end, we select a feedback vertex set $F$ with minimum weight and output it. These steps in total require time at most $\mathcal{O}(2^\omega \omega^2+2^\omega \cdot \omega\log \tau+2^\omega \log \tau) \leq \mathcal{O}(\omega^2 2^\omega\log \tau)$.
\end{proof}
We are now ready to prove \cref{FPTWeightedFVS}. 
\begin{proof}[Proof of \cref{FPTWeightedFVS}.]
	Let $D$ be the input digraph with a vertex-weighting $w:V(D) \rightarrow \mathbb{N}_0$, an let $\tau$ be the given bound on the total weight. 
	
	If $V(D)=\{v\}$ for some vertex $v \in V(D)$, we simply return $F\coloneqq \emptyset$ and $\fv(D,w)=0$ as the solution to the problem. 
	
	Otherwise, if $|V(D)| \ge 2$, we apply the algorithm from \cref{computedecomposition} to $D$ and obtain a partition of $V(D)$ into modules $M_1,\ldots,M_\ell$, where $2 \leq \ell \leq \omega=\dmw(D)$. We next compute the module-digraph $D_M$ with vertex set $\{v_1,\ldots,v_\ell\}$. For each $i \in [\ell]$, we now recursively call the algorithm with instance $(D[M_i],w|_{M_i},\tau)$. Because of $$\sum_{z \in M_i}{w(z)} \leq \sum_{z \in V(D)}{w(z)} \leq \tau,$$ these are feasible instances. Let for each $i \in [\ell]$ $F_i$ be the minimum-weight feedback vertex set in $D[M_i]$ with respect to $w|_{M_i}$ obtained from the recursive call. We next compute for each $i \in [\ell]$ the weight
	$$w_M(v_i)\coloneqq w(M_i)-\fv(D[M_i],w|_{M_i})=\sum_{z \in M_i \setminus F_i}{w(z)}.$$
	Finally we apply the algorithm from \cref{obvious1} to the instance $(D_M,w_M,\tau)$ to obtain a minimum-weight feedback vertex set $F_M$ (note that $|V(D_M)|=l \leq \omega$). Again this instance is feasible, as we have
	$$\sum_{i=1}^{\ell}{w_M(v_i)} \leq \sum_{z \in V(D)}{w(z)} \leq \tau.$$ Finally we compute 
	$$F\coloneqq \left(\bigcup_{\substack{i \in [\ell],\cr
			v_i \in F_M}}{M_i}\right) \cup \left(\bigcup_{\substack{i \in [\ell],\cr
			v_i \notin F_M}}{F_i}\right)$$ and return $F$ as well as its weight. By \cref{mfvrec} we conclude that this indeed is a minimum-weight feedback vertex set for $D$ and that $\fv(D,w)=w(F)$.
	
	The running time of the described algorithm can be analysed using a rooted decomposition-tree $T$ which resembles the structure of the recursive calls appearing during the algorithm. The root of this tree corresponds to the input digraph $D$, while the remaining vertices correspond to all other digraphs appearing in a recursive call during the execution of the algorithm. Whenever we compute a module-decomposition $M_1',\ldots,M_s'$ of an induced subdigraph $D'$ of $D$ during a recursive call, the induced subdigraphs $D'[M_1'],\ldots,D'[M_s']$ correspond to the children of the vertex representing $D'$. Therefore, the leafs of this tree are labeled by the $n$ subdigraphs $D[\{v\}], v \in V(D)$ of $D$. Because $T$ is a decomposition tree with ground set $\Omega\coloneqq V(D)$, we conclude from \Cref{dectree} that $|V(T)| \leq 2n-1$. 
	
	Assume $|V(D)| \ge 2$ and let us first estimate the running time needed for the described computations corresponding to the root $D$, excluding the running time required by the recursive calls to $D[M_1],\ldots,D[M_\ell]$. We first have to apply the algorithm from \cref{computedecomposition}, which requires time $\mathcal{O}(|V(D)|+|E(D)|) \leq \mathcal{O}(n^2)$. Furthermore, we have to compute the weights $w_M(v_i)$ for all $i \in [\ell]$, the total time required here is bounded by $\mathcal{O}(\omega n \log \tau)$. Finally we apply the algorithm from \cref{obvious1} to obtain the minimum-weight FVS $F_M$ in time $\mathcal{O}(\omega^2  2^\omega \log \tau)$. Computing $F$ affords at most $\mathcal{O}(\omega n)$ elementary operations. In total, we conclude that these computations can be executed in time $\mathcal{O}(n^2+\omega^2 2^\omega n \log \tau)$.
	
	The same upper bound applies to any other branching vertex in the tree, as the corresponding induced subdigraph $D'$ of $D$ by Fact \ref{inducedmonotonicity} has less than $n$ vertices and directed modular width at most $\omega$ as well. After having reached a leaf of the tree, this branch of the algorithm terminates in constant time. Therefore the total run-time required for the execution of the algorithm can be upper bounded by
	$$\mathcal{O}(\underbrace{n}_{\text{leafs}}+\underbrace{|V(T)|(n^2+\omega^2  2^\omega n \log \tau)}_{\text{branching vertices}})=\mathcal{O}(n^3+\omega^2  2^\omega n^2\log \tau),$$ which proves the bound claimed in the Theorem.
\end{proof}

\section{Dominating Set}\label{sec:dominating}
For any digraph $D$ and a vertex subset $X \subseteq D$, let $$N_D^+(X)\coloneqq \CondSet{z \in V(D)}{\text{there exists}~ u \in X~\text{such that}~ (u,z) \in E(D)}$$ denote the \emph{out-neighbourhood} of $X$, and let $\coloneqq X \cup N_D^+(X)$ denote the \emph{closed out-neighbourhood}.

In the following, a vertex subset $X \subseteq V(D)$ shall be called \emph{(out-)dominating}, if $\InducedSubgraph{N_D^+}{X}=V(D)$, that is, every vertex in $V(D) \setminus X$ has an in-neighbour inside $X$. In this section, we deal with the well-known problem of finding a minimum dominating set in a given digraph. We denote by $\gamma(D)$ the \emph{directed domination number} of $D$, which is defined as the size of a smallest dominating vertex set. 

\ProblemDefLabelled{Minimum Dominating Set}{DS}{prob:DS}
{A digraph $D$.}
{Find $\gamma(D)$ and a dominating vertex set $X \subseteq V(D)$ with $|X|=\gamma(D)$.}

The minimum dominating set problem on digraphs (as well as its undirected counterpart) is a classical NP-complete problem \cite{garey}. It is known that the problem remains NP-hard on DAGs \cite{directedwidthmeasures}. We refer to \cite{dirdomsetprob} for some digraph-specific algorithms for this problem.

In this section, we construct an FPT-algorithm solving \ref{prob:DS} with respect to directed modular width as the fixed parameter. Again, we design an algorithm which solves the following more general weighted version of the problem.

\ProblemDefLabelled{Minimum Weight Dominating Set}{wDS}{prob:wDS}
{A digraph $D$, and non-negative integer weights $w:V(D) \rightarrow \mathbb{N}_0$ on the vertices, and a threshold $\tau \in \mathbb{N}$ such that $\sum_{z \in V(D)}{w(z)} \leq \tau$.}
{Find a dominating vertex set $X \subseteq V(D)$ of minimum weight $w(X)=:\gamma(D,w)$.}

\begin{theorem} \label{domsetweighted}
	There exists an algorithm that, given as instance a digraph $D$ with vertex-weights $w:V(D) \rightarrow \mathbb{N}_0$ and a corresponding bound $\tau$ on the total weight, computes $\gamma(D,w)$ and a dominating set in $D$ of minimum weight in time $\mathcal{O}(n^3+\omega 2^\omega n^2\log \tau)$, where $n\coloneqq |V(D)|$ and $\omega\coloneqq \dmw(D)$. 
\end{theorem}
Clearly, if we set all weights to $1$, we see that the Minimum Dominating Set Problem is a special case of its weighted version, where we can set $\tau\coloneqq |V(D)|$. We therefore obtain:
\begin{corollary}
	There exists an algorithm that, given as instance a digraph $D$ computes $\gamma(D)$ and a dominating set in $D$ of minimum weight in time $\mathcal{O}(n^3+\omega 2^\omega n^2\log n)$, where $n\coloneqq |V(D)|$ and $\omega\coloneqq \dmw(D)$. 
\end{corollary}
To prove \cref{domsetweighted}, we again start by analysing the structure of dominating sets on a digraph whose vertex set is partitioned into modules. 
\begin{lemma}\label{domsetsmodules}
	Let $D$ be a digraph and let $M_1,\ldots,M_\ell$ be a partition of $V(D)$ into modules. Let $D_M$ denote the corresponding module-digraph with vertex-set $\{v_1,\ldots,v_\ell\}$. Then for every vertex set $U \subseteq V(D)$, the following two statements are equivalent:
	\begin{itemize}
		\item $U$ is a dominating set for $D$.
		\item There exists $X \subseteq U$ such that each of the following holds:
		\begin{enumerate}[label=(\roman*)]
			\item $X_M\coloneqq \CondSet{v_i}{i \in [\ell], X \cap M_i \neq \emptyset}$ is a dominating set for $D_M$.
			\item For all $i \in [\ell]$ with $v_i \in X_M \cap N_{D_M}^+(X_M)$, we have $|X \cap M_i|=1$.
			\item For all $i \in [\ell]$ with $v_i \in X_M \setminus N_{D_M}^+(X_M)$, $X \cap M_i$ is an inclusion-wise minimal dominating set for $D[M_i]$.
		\end{enumerate}
	\end{itemize}
\end{lemma}
\begin{proof}
	$\Longrightarrow$ Let $U$ be a dominating set, and let $X \subseteq U$ be a dominating set which is additionally inclusion-wise minimal. To prove (i), suppose for a contradiction that there was an $i \in [\ell]$ such that $v_i \notin X_M$ and $v_i$ has no in-neighbour within the set $X_M$. This means that $X \cap M_i=\emptyset$ and, by the modular property, that there is no directed edge $(u,z) \in E(D)$ with $u \in X$ and $z \in M_i$, which clearly contradicts the fact that $X$ is a dominating set for $D$.
	
	For (ii), suppose towards a contradiction that there is $i \in [\ell]$ such that $v_i \in X_M \cap N_{D_M}^+(X_M)$ and $|X \cap M_i| \ge 2$. The first condition implies that there exists some module $M_j$ with $X \cap M_j \neq \emptyset$ such that $(v_j,v_i) \in E(D_M)$. Let $u \in X \cap M_j$ be some vertex. By the modular property and the definition of $D_M$, we have $N_D^+(u) \supseteq M_i$.
	
	Now choose some fixed vertex $x \in X \cap M_i$. We claim that also $X'\coloneqq (X \setminus M_i) \cup \{x\}$ is a dominating set for $D$. For that, because of $M_i \subseteq N_D^+(u) \subseteq \InducedSubgraph{N_D^+}{X'}$, it suffices to verify that $N_D^+(y) \setminus M_i \subseteq\InducedSubgraph{N_D^+}{X'}$ for every $y \in (X \cap M_i) \setminus \{x\}$. However, this again is a consequence of the modular property, as we have $N_D^+(y) \setminus M_i=N_D^+(x) \setminus M_i \subseteq \InducedSubgraph{N_D^+}{X'}$ for all $y \in (X \cap M_i) \setminus \{x\}$. Therefore $X'$ is a dominating set which is properly contained in $X$. This is a contradiction to the assumption that $X$ is inclusion-wise minimal. 
	
	To verify (iii), let $i \in [\ell]$ be given such that $v_i \in X_M \setminus N_{D_M}^+(X_M)$. The latter implies that for every vertex $v_j \in V(D_M), j \neq i$, $(v_j,v_i) \in E(D_M)$ implies that $X \cap M_j=\emptyset$. By the properties of the modules and the definition of $D_M$, this implies that no vertex within $M_i$ has an in-neighbour in $X \setminus M_i$. Because $X$ is a dominating set, this means that each vertex in $M_i$ either is contained in $X$ or has an in-neighbour in $X \cap M_i$. This clearly means that $X \cap M_i$ is a dominating set for $D[M_i]$. Finally, $X \cap M_i$ is inclusion-wise minimal: If there was a dominating set $Y \subsetneq X \cap M_i$ for $D[M_i]$, by the definition of a module we would have $N_D^+(Y) \setminus M_i=N_D^+(X \cap M_i) \setminus M_i$, and therefore $(X \setminus M_i) \cup Y$ would form a proper subset of $X$ which is dominating for $D$, again a contradiction to the assumption that $X$ is inclusion-wise minimal. Finally, this implies that (iii) holds.
	
	$\Longleftarrow$ Assume that (i)-(iii) are satisfied. We verify that $X$ is dominating, which clearly implies that the same holds true for $U \supseteq X$. For this purpose, let $z \in V(D) \setminus X$ be arbitrary, and let $i \in [\ell]$ be such that $z \in M_i \setminus X$. 
	
	Assume first $X \cap M_i=\emptyset$, that is, $v_i \notin X_M$. Using (i), we conclude that there is a vertex $v_j \in X_M, j \neq i$ such that $(v_j,v_i) \in E(D_M)$. Let $u \in X \cap M_j \neq \emptyset$ be some vertex. By the modular property, we have $z \in M_i \subseteq N_D^+(u) \subseteq \InducedSubgraph{N_D^+}{X}$ as desired, which concludes this case.
	
	Otherwise, we have that $X \cap M_i \neq \emptyset$ and therefore $v_i \in X_M$. 
	
	Then we either have $v_i \in X_M \cap N_{D_M}^+(X_M)$ or $v_i \in X_M \setminus N_{D_M}^+(X_M)$. In the first case, as above, we conclude the existence of a vertex $v_j \in X_M, j \neq i$ such that $(v_j,v_i) \in E(D_M)$. With the same argumentation as above, we conclude that $z \in \InducedSubgraph{N_D^+}{X}$, as desired.
	
	In the second case, we have $v_i \in X_M \setminus N_{D_M}^+(X_M)$, and so by (iii), $z \in N_{D[M_i]}^+(X \cap M_i) \subseteq \InducedSubgraph{N_D^+}{X}$, which again shows that $z$ is dominated by $X$. 
	
	Finally, this shows that each vertex in $D$ is dominated by $X$ and therefore verifies the claim.
\end{proof}
As a direct consequence, we obtain:
\begin{corollary} \label{optimumdomset}
	Let $D$ be a digraph with vertex-weighting $w:V(D) \rightarrow \mathbb{N}_0$ and a partition $M_1,\ldots,M_\ell$ of the vertex set into modules. Then there exists a dominating set $X_M$ in the module-digraph $D_M$ such that with 
	$$I_1\coloneqq \CondSet{i \in [\ell]}{v_i \in X_M \cap N_{D_M}^+(X_M)}, I_2\coloneqq \CondSet{i \in [\ell]}{v_i \in X_M \setminus N_{D_M}^+(X_M)}$$ we have
	$$\gamma(D,w)=|I_1|+\sum_{i \in I_2}{\gamma(D[M_i],w|_{M_i})}.$$
\end{corollary}
Using the above characterisation of dominating sets, we can reduce the computation of a minimum-weight dominating set to corresponding tasks on the modules. Here, as in the case of the minimum weight feedback vertex set problem, no integer program is required. We simply enumerate all inclusion-wise minimal dominating sets of the module-digraph $D_M$ and then, given such a set, which prescribes which modules are to be filled with vertices and which are to be left empty, compute a dominating set for the whole digraph with minimum weight according to these additional properties. In the end, we simply select the best dominating set obtained in this way.

\begin{proof}[Proof of \cref{domsetweighted}.]
	Assume we are given as instance a directed graph $D$, a weighting $w:V(D) \rightarrow \mathbb{N}_0$ of the vertices, and a bound $\tau \in \mathbb{N}$ on the total weight distributed on the vertices.
	
	If $D$ consists of a single vertex $v$, we output $X\coloneqq \{v\}$ as the unique dominating set. Otherwise, we apply the algorithm from \cref{computedecomposition} in order to obtain a non-trivial decomposition of $V(D)$ into modules $M_1,\ldots,M_\ell$, where $2 \leq \ell \leq \omega$. We compute the induced subdigraphs $D[M_1],\ldots,D[M_\ell]$ as well as the module-digraph $D_M$.
	
	Now we recursively apply the algorithm to each of the instances $(D[M_i],w|_{M_i}),i \in [\ell]$. For each $i \in [\ell]$, we thereby obtain a dominating set $X_i$ for $D[M_i]$ with minimum weight. For each $i \in [\ell]$, we arbitrarily select a vertex $x_i \in M_i$. 
	
	Next we go through all (at most $2^\ell \leq 2^\omega$) subsets of $V(D_M)$ and test whether they are dominating or not. For each dominating set $Z \subseteq V(D_M)$, we compute $$I_1(Z)\coloneqq \CondSet{i \in [\ell]}{v_i \in Z \cap N_{D_M}^+(Z)}, I_2(Z)\coloneqq \CondSet{i \in [\ell]}{v_i \in Z \setminus N_{D_M}^+(Z)}.$$  We now define a subset of $V(D)$ according to
	$$X(Z)\coloneqq \left(\bigcup_{i \in I_1(Z)}{\{x_i\}}\right) \cup \left(\bigcup_{i \in I_2(Z)}{X_i}\right).$$ It is easily seen from \cref{domsetsmodules} that $X(Z)$ is a dominating set for $D$, and clearly, we have
	$$w(X(Z))=|I_1(Z)|+\sum_{i \in I_2(Z)}{\gamma(D[M_i],w|_{M_i})}.$$
	\cref{optimumdomset} implies that $\min_{Z \subseteq V(D_M)}{w(X(Z))}=\gamma(D,w)$. Therefore, we simply select the dominating set $Z$ with minimum value $w(X(Z))$ and output $X(Z)$, which is a dominating set for the digraph $D$ with minimum weight $\gamma(D,w)$.
	
	It remains to analyse the running time of the described algorithm. 
	Let $T$ denote a rooted decomposition-tree which resembles the structure of the recursive calls appearing during the algorithm. Whenever we compute a module-decomposition $M_1',\ldots,M_s'$ of an induced subdigraph $D'$ of $D$ during a recursive call, the induced subdigraphs $D'[M_1'],\ldots,D'[M_s']$ correspond to the children of the vertex representing $D'$. We again conclude from Fact \ref{dectree} that $|V(T)| \leq 2n-1$, as the leafs of $T$ correspond to the $n$ singletons in $V(D)$.
	
	Let us estimate the running time needed for the described computations corresponding to the root $D$, excluding the running time required by the recursive calls to $D[M_1],\ldots,D[M_\ell]$. We assume $|V(D)| \ge 2$, otherwise the algorithm terminates in constant time. We first have to apply the algorithm from \cref{computedecomposition}, which requires time $\mathcal{O}(|V(D)|+|E(D)|) \leq \mathcal{O}(n^2)$. Furthermore, we have to compute $D_M$ and $D[M_1],\ldots,D[M_\ell]$, this certainly can be done in time $\mathcal{O}(n^2)$ as well. Next we enumerate the at most $2^\ell \leq 2^\omega$ subsets of $V(D_M)$ and for each $Z \subseteq V(D_M)$ test whether it is dominating. If so, we furthermore compute $I_1(Z), I_2(Z)$, $X(Z)$ and the sum $X(Z)$ of weights. In total, the number of elementary operations required here is bounded by $\mathcal{O}(2^\omega (\omega^2+\omega+\omega n+n \log \tau))$. Finally we output $X(Z)$, where $Z$ minimises $w(X(Z))$. Determining such a set needs no more time than $\mathcal{O}(2^\omega \log \tau)$. 
	
	Finally, we can execute all operations performed with respect to $D$ in the algorithm using $\mathcal{O}(n^2+2^\omega n \log \tau+\omega 2^\omega n) \leq \mathcal{O}(n^2+\omega 2^\omega n \log \tau)$ operations.
	
	The same upper bound applies to any other branching vertex in the tree, as the corresponding induced subdigraph $D'$ of $D$ by Fact \ref{inducedmonotonicity} has less than $n$ vertices and directed modular width at most $\omega$ as well. After having reached a leaf of the tree, this branch of the algorithm terminates in constant time. Therefore the total run-time required for the execution of the algorithm can be upper bounded by
	$$\mathcal{O}(\underbrace{n}_{\text{leafs}}+\underbrace{|V(T)|(n^2+\omega 2^\omega n \log \tau)}_{\text{branching vertices}})=\mathcal{O}(n^3+\omega 2^\omega n^2\log \tau),$$ which proves the bound claimed in the Theorem.
\end{proof}

\section{Dichromatic number}\label{digraphcolouring}
In this section, we investigate the complexity of computing the so-called \emph{dichromatic number} of a given directed graph, when modular directed width is used as a parameter. In an \emph{acyclic colouring} of a given digraph $D$, we assign colours to the vertices of $D$ such that there are no monochromatic directed cycles. The \emph{dichromatic number} of a digraph $D$, denoted by $\vec{\chi}(D)$, now is defined as the minimum number of colours required for an acyclic colouring of a digraph. This notion was introduced in 1982 by Neumann-Lara (\cite{neulara}), rediscovered by Mohar (\cite{mohar}), investigated in \cite{bokal2004circular}, and since then has attracted wide interest. It has become apparent that the dichromatic number acts as a natural directed counterpart of the chromatic number of an undirected graph. Numerous recent results (see \cite{perfect}, \cite{fractionalNL}, \cite{largesubdivisions}, \cite{dig4}, \cite{lists}, \cite{HARUTYUNYAN2019}, \cite{noneven}) support this claim.

Formally, we consider the following problem.

\ProblemDefLabelled{Digraph Colouring}{DC}{prob:DC}
{A digraph $D$.}
{What is the value of $\vec{\chi}(D)$?}

While graph colouring famously can be solved by an FPT-algorithm with respect to treewidth, it was shown in \cite{noneven} that even for bounded size of a directed feedback vertex set, deciding whether a directed graph has dichromatic number at most $2$ is NP-complete. This rules out efficient parameterisations by most known directed width parameters such as directed treewidth, DAG-width or Kelly-width, as all of these are upper bounded in terms of the size of a smallest feedback vertex set. In fact, up to date, only few classes of digraphs for which the dichromatic number can be evaluated in polynomial time are known. In the following, we show that using modular directed width as a parameter, there is an FPT-algorithm. To obtain this algorithm, we slightly generalise the problem of determining the dichromatic number to enable a recursion.
\begin{definition}
	Let $D$ be a digraph equipped with an assignment $N:V(D) \rightarrow \mathbb{N}$ of positive integers to the vertices.
	An \emph{$N$-colouring} with $k\in\mathbb{N}$ colours of $D$ is an assignment of lists $\Fkt{c}{v}\subseteq[k]$ of colours to every vertex $v\in\Fkt{V}{D}$ such that $\Abs{\Fkt{c}{v}}=\Fkt{N}{v}$for all $v\in\Fkt{V}{D}$, and moreover there is no directed cycle $C$ in $D$ such that $i \in c(v)$ for every $v \in V(C)$ and any $i \in [k]$.
%
	
	We define the \emph{$N$-dichromatic number} $\vec{\chi}_N(D)$ of a digraph $D$ to be the minimum $k$ such that an $N$-colouring with $k$ colours of $D$ exists.
\end{definition}
As an additional input for our generalised colouring problem, we also have a threshold $\tau \in \mathbb{N}$, which bounds the total size of the colour lists which have to be assigned. For bounded directed modular width, the proposed algorithm runs in polynomial time in $\tau$ and $n\coloneqq |V(D)|$.

\ProblemDefLabelled{Weighted Digraph Colouring}{wDC}{prob:wDC}
{A digraph $D$, a natural number $\tau \in \mathbb{N}$, and an assignment $N:V(D) \rightarrow \mathbb{N}$ such that $\sum_{v \in V(D)}{N(v)} \leq \tau$.}
{What is the value of $\vec{\chi}_N(D)$?}

\begin{theorem}\label{algdichromatic}
	There is an algorithm that, given a digraph $D$ on $n$ vertices and an assignment $N:V(D) \rightarrow \mathbb{N}$ of numbers to the vertices such that $\sum_{v \in V(D)}{N(v)} \leq \tau$, outputs the value of $\vec{\chi}_N(D)$ together with a certifying assignment of colour lists to the vertices. The running time of the algorithm is $\mathcal{O}(n^3+f(\omega)n\log^2 \tau+n\tau)$, where $n\coloneqq |V(D)|$, $\omega\coloneqq \dmw(D)$, and $f(\omega)=2^{\mathcal{O}(\omega 2^\omega)}$.
\end{theorem}
Clearly, for any digraph, the assignment $N(v)\coloneqq 1$ for all $v \in V(D)$ leads to $\vec{\chi}_N(D)=\vec{\chi}(D)$ and we can put $\tau\coloneqq n$. Therefore, the FPT-algorithm from \cref{algdichromatic} contains the computation of the dichromatic number as a special case.
\begin{corollary}
	The dichromatic number of a given digraph $D$ can be computed in time $\mathcal{O}(n^3+f(\omega)n\log^2 n)$, where $n\coloneqq |V(D)|$, $\omega\coloneqq \dmw(D)$ and $f(\omega)=2^{\mathcal{O}(\omega 2^\omega)}$.  
\end{corollary}
We prepare the proof of \cref{algdichromatic} with some auxiliary statements.
\begin{lemma} \label{InducedCycles}
	Let $D$ be a digraph equipped with a module-decomposition $\Set{M_1,\ldots,M_\ell}$ of the vertex set. If $C$ is an induced (that is, chordless) directed cycle in $D$, then either there is some $i \in [\ell]$ such that $C$ is contained in $D[M_i]$, or $D$ uses at most one vertex from each module.
\end{lemma}
\begin{proof}
	Assume towards a contradiction that there was a directed cycle $C$, such that for some $i \in [\ell]$ we have $|V(C) \cap M_i| \ge 2$, and $V(C) \setminus M_i \neq \emptyset$. Let $x \in V(C) \setminus M_i$ be some vertex, and let $y_1$ be the closest vertex after $x$ in the cyclic directed order along $C$ which is contained in $M_i$, and let $y_2$ be the closest vertex before $x$ contained in $M_i$ in the cyclic order. Because of $|V(C) \cap M_i| \ge 2$, we know that $y_1 \neq y_2$. Let $x_1 \in V(C) \setminus M_i$ be the predecessor of $y_1$ on $C$, and let $x_2 \in V(C) \setminus M_i$ be the successor of $y_2$ on $C$. This means that $(x_1,y_1), (y_2,x_2) \in E(C)$. By the modular property, this implies that also $(x_1,y_2), (y_1,x_2) \in E(D)$. Because $C$ was assumed to be chordless, this implies that $E(C)=\{(x_1,y_1), (y_1,x_2), (y_2,x_2),(x_1,y_2)\}$, contradicting the fact that $C$ is a directed cycle.
\end{proof}
\begin{lemma}\label{recursion}
	Let $D$ be a digraph, let $M_1,\ldots,M_\ell$ be a partition of $V(D)$ into modules, and let $D_M$ denote the corresponding module-digraph. 
	
	Let $N:V(D) \rightarrow \mathbb{N}$ be an assignment of numbers to the vertices.
	
	Denote by $v_i \in V(D_M)$ for every $i \in [\ell]$ the vertex of the module-digraph representing $M_i$, and define an assignment $N_M:V(D_M) \rightarrow \mathbb{N}$ according to $N_M(v_i)\coloneqq \vec{\chi}_{N|_{M_i}}(D[M_i])$ for each $i \in [\ell]$. 
	
	Then we have
	$$\vec{\chi}_N(D)=\vec{\chi}_{N_M}(D_M).$$
	Moreover, given an optimal $N_M$-colouring of $D_M$ (i.e., with a minimal total number of colours), we can construct an optimal $N$-colouring of $D$ in time $\mathcal{O}(\ell |V(D)|).$
\end{lemma}
\begin{proof}
	Let $k\coloneqq \vec{\chi}_N(D), k_M\coloneqq \vec{\chi}_{N_M}(D_M)$. We prove both inequalities $k \leq k_M$ and $k_M \leq k$ separately.
	
	To prove the first inequality, consider an assignment $c_M:V(D_M) \rightarrow 2^{[k_M]}$ of colour lists that define an optimal $N_M$-colouring of $D_M$. We therefore have $|c_M(v_i)|=\vec{\chi}_{N|_{M_i}}(D[M_i])$ for all $i \in [\ell]$. The latter implies that for every $i \in [\ell]$, the digraph $D[M_i]$ admits an $N|_{M_i}$-colouring $c_i$ using $|c_M(v_i)|$ colours in total. By relabeling, we may assume that in total the set of colours used by $c_i$ is exactly $c_M(v_i)$. Consider now the assignment $c:V(D) \rightarrow 2^{[k_M]}$ defined by $c(v)\coloneqq c_i(v)$ whenever $v \in M_i$. This assignment has the property that $|c(v)|=N|_{M_i}(v)=N(v)$ for all $v \in M_i \subseteq V(D), i \in [\ell]$. We claim that it defines a valid $N$-colouring of $D$. Assume towards a contradiction that there exists a directed cycle $C$ in $D$ such that $\bigcap_{v \in V(C)}{c(v)} \neq \emptyset$. W.l.o.g. we may assume that $C$ is chosen with $V(C)$ inclusion-wise minimal, i.e., $C$ is chordless (if there was a chord, we could find a directed cycle using a proper subset of the vertices). By \cref{InducedCycles} $C$ either is contained in $D[M_i]$ for some $i \in [\ell]$, or it uses at most one vertex from each module. In the first case, we obtain that $C$ is a directed cycle in $D[M_i]$ with $\bigcap_{v \in V(C)}{c_i(v)} \neq \emptyset$, which contradicts the choice of $c_i$ as a proper $N|_{M_i}$-colouring of $D[M_i]$. In the second case, $C$ in a natural way yields a directed cycle $C_M$ in the module-digraph $D_M$ such that for every vertex $v_i \in V(C_M)$, there is a unique corresponding vertex $w_i \in V(C) \cap M_i$ from the module. By the choice of the colourings $c_i$, we find that $c(w_i)=c_i(w_i) \subseteq \bigcup_{z \in M_i}{c_i(z)}=c_M(v_i)$ for all $i \in [\ell]$, and so we have
	$$\bigcap_{v_i \in V(C_M)}{c_M(v_i)} \supseteq \bigcap_{w_i \in V(C)}{c(w_i)} \neq \emptyset,$$ contradicting the choice of $c_M$ as a valid $N_M$-colouring of $D_M$. As all cases led to a contradiction, we conclude that indeed $c$ is a proper $N$-colouring of $D$ whose colour sets are contained in $[k_M]$, and therefore $k \leq k_M$ as claimed.
	
	To prove the second inequality, consider an optimal $N$-colouring $c:V(D) \rightarrow 2^{[k]}$ of $D$. Clearly, for any $i$, the restriction $c|_{M_i}$ defines a valid $N|_{M_i}$-colouring of $D[M_i]$, and therefore has to fulfil
	$$\left|\bigcup_{z \in M_i}{c(z)} \right| \ge \vec{\chi}_{N|_{M_i}}(D[M_i])=N_M(v_i).$$ 
	Now for each $i \in [\ell],$ choose some subset $L_i \subseteq \bigcup_{z \in M_i}{c(z)}$ with $|L_i|=N_M(v_i)$ and define an assignment $c_M:V(D_M) \rightarrow 2^{[k]}$ of colour lists to the vertices of $D_M$ according to
	$c_M(v_i)\coloneqq L_i.$
	We claim that this defines a proper $N_M$-colouring of $D_M$. Assume towards a contradiction that there exists a directed cycle $C$ in $D_M$ and a colour $\tilde{c} \in [k]$ such that $\tilde{c} \in L_i$ for every $v_i \in V(C)$. By the definition of $L_i$, this implies that for every $i$ such that $v_i \in V(C_i)$, we can find a vertex $w_i \in M_i$ such that $\tilde{c} \in c(w_i)$. By the properties of the modules, we now immediately conclude that $\CondSet{w_i}{i \in [\ell], v_i \in V(C)}$ forms the vertex set of a directed cycle in $D$ such that $\tilde{c}$ is contained in all colour sets of its vertices. This contradicts the fact that $c$ was chosen as a valid $N$-colouring of $D$. Finally, since $c_M$ uses at most $k$ colours in total, this shows $k_M \leq k$ and concludes the proof of the claimed equality.
	
	Concerning the (algorithmic) construction of an optimal $N$-colouring of $D$ given an optimal $N_M$-colouring of $D_M$ (using $k_M$ colours in total), we can simply compute the $N$-colouring of $D$ as defined in the proof of $k \leq k_M$. This requires at most $\mathcal{O}(|V(D)|)$ operations for each module, and so at most $\mathcal{O}(\ell |V(D)|)$ (rough estimate) in total. 
\end{proof}
\begin{lemma} \label{smallgraphs}
	Given a digraph $D$ on at most $\omega \in \mathbb{N}$ vertices, an assignment $N:~V(D) \rightarrow~\mathbb{N}$ and some $\tau \in \mathbb{N}$ such that $\sum_{v \in V(D)}{N(v)} \leq \tau$, $\vec{\chi}_N(D)$ and a corresponding optimal $N$-colouring can be computed in time $\mathcal{O}(f(\omega)\log^2 \tau)$, where $f(\omega)=2^{\mathcal{O}(\omega 2^\omega)}$.
\end{lemma}
\begin{proof}
	We reformulate the problem of determining the $N$-dichromatic number as a linear integer program to enable an application of \cref{ipsolving}. For this purpose, note that we can alternatively represent a colour-list assignment $c:V(D) \rightarrow 2^{[k]}$ by the collection of 'colour classes' $A_i\coloneqq \{v \in V(D)|i \in c(v)\}$ for all colours $i \in [k]$, each of which (by the definition of an $N$-colouring) induces an acyclic subdigraph of $D$.
	
	We can therefore associate an $N$-colouring with a set of variables $x_A, A \in \mathcal{A}(D)$, where $\mathcal{A}(D)$ is the collection of acyclic vertex sets in $D$, and $x_A$ counts the number of $i \in [k]$ such that $A=A_i$. The condition that the assigned colour lists are of the sizes required by $N$ can be formulated as a linear equation for each vertex. This shows that we can compute $\vec{\chi}_N(D)$ as the optimal value of the following ILP:
	\begin{alignat}{3}\label{primalfract}
		~&&~\min \sum_{A \in \mathcal{A}(D)}{x_A} ~&&~\\
		\text{  subj.\@ to }~&&~	\sum_{A \ni v}{x_A}=N(v) ~&&~ \text{ for all } v \in V(D) \nonumber\\
		~&&~  x_A \ge 0 ~&&~ \text{ for all } x_A \in \mathbb{Z}\nonumber
	\end{alignat}
	This ILP in canonical form has $p=\mathcal{O}(|\mathcal{A}(D)|)=\mathcal{O}(2^\omega)$ variables. The coding length $L$ of the matrix and the vectors describing this ILP is clearly bounded by $(\omega+1)2^\omega \log \tau$. Setting up the ILP requires enumerating the subsets $A \subseteq V(D)$ for which $D[A]$ is acyclic. As we can test whether $D[A]$ is acyclic in time $\mathcal{O}(|E(D)|) \leq \mathcal{O}(\omega^2)$, the ILP can be set up in time $\mathcal{O}(\omega^2 2^\omega +\omega 2^\omega \log \tau) \leq \mathcal{O}(\omega^2 2^\omega \log \tau).$
	
	It is readily verified that the optimal value $\vec{\chi}_N(D)$ of the program is bounded from above by $U_1\coloneqq \tau$, more generally, it holds that $\vec{\chi}_N(D) \leq \sum_{v \in V(D)}{N(v)}$ (assign disjoint colour sets to the different vertices). In an optimal solution to program \ref{primalfract}, we certainly have $x_A \leq \tau$ for all $A \in \mathcal{A}(D)$. Therefore we can put $U_2\coloneqq \tau$.  Application of \cref{ipsolving} now yields that there is an algorithm for determining $\vec{\chi}_N(D)$ in time $\mathcal{O}(f(\omega)\log \tau \log \tau^2)=\mathcal{O}(f(\omega)\log^2 \tau)$ for some function $f$. In fact, we may take $f(\omega)=p^{\mathcal{O}(p)}+\omega^2 2^\omega \leq 2^{\mathcal{O}(\omega 2^\omega)}$. This proves the claim.
\end{proof}
\begin{proof}[Proof of \cref{algdichromatic}]
	We follow a recursive approach which makes use of \cref{computedecomposition}, \cref{recursion} and \cref{smallgraphs}.
	
	Assume we are given a digraph $D$, the assignment $N:V(D) \rightarrow \mathbb{N}$ and a natural number $\tau \in \mathbb{N}$ such that $\sum_{v \in V(D)}{N(v)} \leq \tau$ as input. Let $\omega\coloneqq \dmw(D)$ be the directed modular width of $D$. 
	
	If $|V(D)|=1$, say $V(D)=\{v\}$, we return $\vec{\chi}_N(D)=N(v)$ and a colour-list of $N(v)$ different colours.
	
	If $|V(D)| \geq 2$, we first apply the algorithm from \cite{moduledecomposition} to $D$, in order to obtain a partition of $V(D)$ into modules $M_1,\ldots,M_\ell$, where $2 \leq \ell \leq \omega$. We compute the digraphs $D[M_1], \ldots, D[M_\ell]$ as well as the module-digraph $D_M$.
	
	We now recursively apply the algorithm to each of the instances $(D[M_1],N|_{M_1},\tau), \ldots,$ $(D[M_\ell],N|_{M_\ell},\tau)$. Each of the digraphs $D[M_i]$ (according to Fact \ref{inducedmonotonicity}) has directed modular width at most $\omega$ and less vertices than $D$.
	
	Now given the outputs $\vec{\chi}_{N|_{M_i}}(D), i=1,\ldots,\ell$ of these recursive calls (and the corresponding optimal colour-list assignments), as in \cref{recursion}, we define $N_M: V(D_M) \rightarrow \mathbb{N}$ according to $N_M(v_i)\coloneqq \vec{\chi}_{N|_{M_i}}(D)$ for all $i \in [\ell]$. Because of $|V(D_M)|=\ell \leq \omega$ we can now apply the algorithm from \cref{smallgraphs} to the instance $(D_M,N_M,\tau)$ in order to obtain the value of $\vec{\chi}_{N_M}(D_M)$ and a corresponding optimal $N_M$-colouring of $D_M$. Note that the instance $(D_M,N_M,\tau)$ is feasible, as we have the estimate
	$$\sum_{v \in V(D_M)}{N_M(v)}=\sum_{i=1}^{\ell}{\vec{\chi}_{N|_{M_i}}(D[M_i])} \leq \sum_{i=1}^{\ell}{\sum_{v \in M_i}{N(v)}}=\sum_{v \in V(D)}{N(v)} \leq \tau.$$
	By the procedure explained in (the proof of) \cref{recursion}, from an optimal colour-list assignment for $D_M$ with respect to $N_M$ we obtain an optimal colour-list assignment for $D$, and furthermore can calculate the $N$-dichromatic number of $D$ according to $\vec{\chi}_N(D)=\vec{\chi}_{N_M}(D_M)$. 
	
	It remains to argue for the correctness of the algorithm. First of all, the algorithm returns an optimal $N$-colouring of $D$ in finite time: In each of the recursive calls, the number of vertices of the digraphs in the instances is strictly smaller than the number of vertices in the current digraph. At some point, we therefore have reduced the task to such instances where the digraphs consist of a single vertex. In this case, the algorithm outputs a solution without further recursion. 
	
	For the runtime-analysis, we again consider a rooted tree $T$ corresponding to the execution of the algorithm, where the root vertex is identified with the digraph $D$ in the instance, and the remaining vertices are each identified with a different induced subdigraph $D'$ which appears in a recursive call during the execution. The children of a vertex corresponding to such a digraph $D'$ are associated with the induced subdigraphs $D'[M_1'],\ldots,D'[M_{s}']$ for the corresponding module-decomposition $\Set{M_1',\ldots,M_{s}'}$. The leafs of this tree correspond to single-vertex digraphs. 
	
	The runtime consumed by a vertex in $T$ corresponding to a call with instance $(D',N',\tau)$ (disregarding the time needed to execute the recursive calls corresponding to its successors) involves 
	\begin{itemize}
		\item Computing a directed modular decomposition $\Set{M_1',\ldots,M_s'}$ of $D'$ with the algorithm from \cite{moduledecomposition}, which takes time $\mathcal{O}(|V(D')|+|E(D')|) \leq \mathcal{O}(|V(D')|^2)$.
		\item Computing the digraphs $D'[M_1'], \ldots, D'[M_s']$ and the module-digraph $D'_M$. This certainly can be executed in time $\leq \mathcal{O}(|V(D')|^2)$.
		\item Applying the algorithm from \cref{smallgraphs} to the instance $(D'_M,N'_M,\tau)$, which can be executed ($D'_M$ has $s \leq \omega$ vertices) in time $\mathcal{O}(f(\omega)\log^2 \tau)$.
		\item Constructing an optimal $N'$-colouring of $D'$ given an optimal $N'_M$-colouring of $D'_M$ and optimal $N'|_{M_i'}$-colourings of $D'[M_i']$ for all $i \in [s]$. By \cref{recursion}, we get $\mathcal{O}(\omega |V(D')|)$ as an upper bound for the required runtime.
	\end{itemize}
	This yields an upper bound of $\mathcal{O}(|V(D')|^2+f(\omega)\log^2 \tau)=\mathcal{O}(n^2+f(\omega)\log^2 \tau)$ for the runtime needed for the computations corresponding to the non-leaf vertex $D'$ of the tree. For a leaf vertex corresponding to a digraph with unique vertex $v$, we only need to output $N(v) \leq \tau$ different colours, which requires linear time in $\tau$. The number of leafs of $T$ clearly is the number of vertices of $D$, which is $n$. Therefore, summing over all vertices of $T$, we conclude that the runtime of our algorithm is at most $$\mathcal{O}(\underbrace{|V(T)|(n^2+f(\omega)\log^2 \tau)}_{\text{branchings}}+\underbrace{n\tau}_{\text{leafs}}).$$ By Fact \ref{dectree}, we have $|V(T)| \leq 2n-1$. This finally yields the desired upper bound of $\mathcal{O}(n^3+f(\omega)n\log^2 \tau+n\tau)$ for the total run-time. As in \cref{smallgraphs}, we can bound $f$ by $f(\omega)=2^{\mathcal{O}(\omega 2^\omega)}$.
\end{proof}

\section{Hamiltonian Paths and Cycles}\label{sec:hamilton}
In this section we present an algorithm which solves the Path Partitioning Problem on directed graphs by an FPT-algorithm with respect to directed modular width. As simple consequences we obtain FPT-algorithms for the problems Directed Hamiltonian Cycle and Directed Hamiltonian Path. Our approach is similar to the one for the undirected case in \cite{gajarsky2013parameterized}. In the following, a \emph{path-partition} of a digraph $D$ shall be defined to be a collection of vertex-disjoint directed paths $P_1,\ldots,P_r$ in $D$ which together cover all vertices of $D$. In order to guarantee such a collection of paths to always exist, for the remainder of this section, we exceptionally allow a single vertex to be treated as a directed path as well. This will \emph{not} be the case in other sections such as in \cref{vddp}. In the following, we denote by $\ham(D)$ the minimum size of a path-partition of $D$.

\ProblemDefLabelled{Partitioning into Directed Paths}{PDP}{prob:PDP}
{A digraph $D$.}
{Find $\ham(D)$.}

Clearly, a digraph $D$ admits a directed Hamiltonian path iff $\ham(D)=1$. It is clear from that perspective that the above problem is NP-hard in general, even for planar digraphs and for digraphs with bounded degree \cite{P79}. However, on directed acyclic graphs, $\ham(D)$ can be computed in polynomial time \cite{stefan}. 

\begin{theorem} \label{hamlecker}
	There exists an algorithm that, given a digraph $D$, computes $\ham(D)$ in time $\mathcal{O}(n^3+f(\omega) n^2\log n)$, where $n\coloneqq |V(D)|$, $\omega\coloneqq \dmw(D)$ and $f(\omega)=2^{\mathcal{O}(\omega^2 \log \omega)}$.
\end{theorem}

To prove the above Theorem, we again use a mixture of dynamic programming and Integer Programming with a bounded number of variables. We prepare the proof with the following auxiliary statements.
\begin{lemma} \label{obsone}
	Let $D$ be a digraph, and for any integer $s \ge 1$, let $D \oplus [s]$ denote the digraph obtained from $D$ by adding an independent set of $s$ new vertices to $D$, labeled by the elements of $[s]=\{1,2,\ldots,s\}$, and connecting each new vertex to all vertices in $V(D)$ by digons. Then we have
	$$\ham(D)=\min\CondSet{s \ge 1}{D \oplus [s] \text{ is directed Hamiltonian}}.$$
\end{lemma}
\begin{proof}
	Let $P_1,\ldots,P_r$ be an optimal path-partition of $D$, i.e., $r=\ham(D)$. It is clear that now $\Brace{1,P_1,2,P_2,3,P_3,\ldots,r,P_r,1}$ defines a directed Hamiltonian cycle in $D \oplus [r]$. On the other hand, given a directed Hamiltonian cycle $C$ in $D \oplus [s]$ for some $s \ge 1$, it is readily verified that the weakly connected components of $C-(V(C) \cap [s])$ form a partition of $D$ into $s$ directed paths. This proves the claim.
\end{proof}
In the proof of our next Lemma we express the existence of a directed Hamiltonian cycle in a digraph $D$ equipped with a modular decomposition as the feasibility of a certain ILP with the number of variables being bounded in terms of the number of modules. 
\begin{lemma} \label{hamprogram}
	Let $D$ be a digraph, $M_1,\ldots,M_\ell$ a partition of $V(D)$ into modules, $D_M$ the corresponding module-digraph with vertex set $\{v_1,\ldots,v_\ell\}$. There exists an algorithm that, given as instance the $D_M$ and the numbers $\ham(D[M_1]),|M_1|,\ldots,\ham(D[M_\ell]),|M_\ell|$, decides whether $D$ is directed Hamiltonian in time $\mathcal{O}(f(\ell) \log n)$, where $n\coloneqq |V(D)|$ and $f(\omega)=2^{\mathcal{O}(\ell^2 \log \ell)}$.
\end{lemma}
\begin{proof}
	We consider a system of linear constraints on integer variables and prove that its feasible assignments yield directed Hamiltonian cycles in $D$ (and vice versa). In a digraph, for any set $X$ of vertices, let us denote by $\delta^+(X)$ the set of directed edges starting in $X$ and ending outside of $X$. The system of linear equations has a non-negative integer variable $x_e$ for every edge $e \in E(D_M)$. It is defined as follows:
	\begin{alignat*}{2}
		\sum_{t\in N_{D_M}^-(v_i)}{x_{\Brace{t,v_i}}}&=\sum_{t \in N_{D_M}^+(v_i)}{x_{\Brace{v_i,t}}} ~&&~ \text{ for all } i \in [\ell],\\
		\sum_{t \in N_{D_M}^+(v_i)}{x_{\Brace{v_i,t}}} &\leq |M_i| ~&&~ \text{ for all }i \in [\ell],\\
		\sum_{t \in N_{D_M}^+(v_i)}{x_{\Brace{v_i,t}}} &\geq \ham(D[M_i]) ~&&~ \text{ for all }i \in [\ell],\\
		\sum_{e \in \delta^+(X)}{x_e} &\ge 1 ~&&~ \text{ for all } \emptyset \neq X \subsetneq V(D_M),\\
		0 &\leq x_{e} \in \mathbb{Z} ~&&~ \text{ for all }e \in E(D_M).
	\end{alignat*}
	Given a directed Hamiltonian cycle $C$ in $D$, for any edge $e=(v_i,v_j) \in E(D_M)$, we define $x_e \ge 0$ as the number of edges on $C$ starting in $M_i$ and ending in $M_j$. The first two types of inequalities in the system are easily seen to be satisfied for any $i \in [\ell]$, because $C$ enters and leaves the module $M_i$ the same number of times, and clearly has to use a different vertex from $M_i$ each time it enters. Moreover, the weak components of $C-(V(C) \setminus M_i)$ form a partition of the digraph $D[M_i]$ into directed paths. The number of these paths is equal to the number $\sum_{t \in N_{D_M}^-(v_i)}{x_{\Brace{t,v_i}}}=\sum_{t \in N_{D_M}^+(v_i)}{x_{\Brace{v_i,t}}}$ of times the cycle $C$ enters or leaves the module, respectively. Therefore, this number must be lower-bounded by $\ham(D[M_i])$. Finally, because $C$ is a directed cycle that visits all modules, it follows that for every non-trivial partition of $V(D_M)$, it has to have edges in $D_M$ between the partition sets in both directions. This implies that the last inequality is satisfied. 
	
	The other way round, assume we are given a feasible variable-assignment $x_e \in \mathbb{N}_0$ of variables for the above system of linear inequalities. Let $D_M^\ast$ denote the multi-digraph obtained from $D_M$ by replacing every edge $e \in E(D_M)$ by $x_e$ parallel edges if $x_e \ge 1$, and deleting the edge $e$ if $x_e=0$. The first type of constraints shows that in this multi-digraph, in- and out-degree coincide for every vertex. The inequalities $\sum_{e \in \delta^+(X)}{x_e} \ge 1, \text{ for all } \emptyset \neq X \subsetneq V(D_M)$ ensure that $D_M^\ast$ is connected. Finally, this implies that $D_M^\ast$ admits a directed Eulerian tour, and therefore, there is a closed directed walk $W$ in $D_M$ which visits every edge $e \in E(D_M)$ exactly $x_e$ times. The remaining two conditions ensure that for any $i \in [\ell]$, the number $n_i\coloneqq \sum_{t \in N_{D_M}^-(v_i)}{x_{\Brace{t,v_i}}}=\sum_{t \in N_{D_M}^+(v_i)}{x_{\Brace{v_i,t}}}$ of times $W$ visits $v_i$ fulfils $\ham(D[M_i]) \leq n_i \leq |M_i|$. The latter implies that there exists a partition of $D[M_i]$ into exactly $n_i$ directed paths. 
	It is now easy to construct a directed Hamiltonian cycle. We start with some module, say $M_1$, and visit the modules in the order of traversal of the modules as given by $W$. Between entering a module $M_i$ and leaving it again, we have to move on a directed path within $D[M_i]$. This path will be chosen as follows: If we enter the module $M_i$ for the $j$-th time, where $1 \leq j \leq n_i$, we move along the $j$-th directed path in the path-partition of $D[M_i]$ selected above and then leave the module again. Note that during the process, whenever we leave a module $M_i$ and enter a module $M_j$, we have $(v_i,v_j) \in D_M$, and therefore, by the definition of $D_M$, all directed edges $(z_1,z_2)$, $z_1 \in M_i, z_2 \in M_j$ exist in $D$. This finally implies that the described traversal of the vertices in $D$ indeed defines a directed Hamiltonian cycle. 
	
	We have proven that checking Hamiltonicity of $D$ is equivalent to deciding feasibility of the above system of linear constraints for integer variables. This system has $p=|E(D_M)|=\mathcal{O}(\ell^2)$ variables and $\mathcal{O}(2^\ell)$ constraints. The coding length $L$ of the matrix and the vector describing the linear constraints is therefore bounded by $\mathcal{O}(\ell^2 2^\ell \log n)$ (note that all absolute values in the matrix and the vector are bounded by $n$). Hence, according to \cref{ipfeasibility}, this task can be executed in time $\mathcal{O}(f(\ell) \log n)$, as claimed. Here we have $f(\ell)=p^{O(p)} \cdot \ell^2 2^\ell=2^{\mathcal{O}(\ell^2 \log \ell)}$.
\end{proof}
\begin{corollary}\label{hamrecursion}
	If we are given a digraph $D$, a module-decomposition $M_1,\ldots,M_\ell$, the module-digraph $D_M$ and the numbers $\ham(D[M_1]),\ldots,\ham(D[M_\ell])$, we can compute $\ham(D)$ in time $\mathcal{O}(f(\ell)n \log n)$, where $n\coloneqq |V(D)|$, $\omega\coloneqq \dmw(D)$ and $f(\ell)=2^{\mathcal{O}(\ell^2 \log \ell)}$.
\end{corollary}
\begin{proof}
	Using the algorithm from \cref{hamprogram}, for any fixed number $s \in [n]$, we can test whether $D \oplus s$ is Hamiltonian in time $\mathcal{O}(f(\ell) \log n)$. For this we add the set $M_{\ell+1}\coloneqq [s]$ of extra vertices as the $(\ell+1)$-the module to the module-decomposition of $D$, to obtain a module-decomposition of $D \oplus [s]$. We clearly know that $\ham((D \oplus [s])[M_{\ell+1}])=s$. Using \cref{obsone}, we therefore can compute $\ham(D)$ in time $n\mathcal{O}(f(\ell) \log n)=\mathcal{O}(f(\ell) n\log n)$, as claimed.
\end{proof}
We are now ready for the proof of \cref{hamlecker}. 
\begin{proof}[Proof of \cref{hamlecker}.]
	Let $D$ be the input digraph of directed modular width $\omega$. If $|V(D)|=1$, we simply return $\ham(D)=1$. Otherwise, we apply \cref{computedecomposition} in order to obtain a module-decomposition $M_1,\ldots,M_\ell$ of $D$ where $2 \leq \ell \leq \omega$ in time $\mathcal{O}(|V(D)|+|E(D)|)=\mathcal{O}(n^2)$. We now make a recursive call to $D[M_i]$ to obtain the value $\ham(D[M_i])$ for $i=1,\ldots,\ell$. Using \cref{hamrecursion} we can now compute $\ham(D)$ in time $\mathcal{O}(f(\omega) \log n)$. In total, these steps (except for the time needed to execute the recursive calls) require time $\mathcal{O}(n^2+f(\omega) \log n)$. The same upper bound applies to the running time consumed by any other recursive call during the algorithm (because each induced subdigraph of $D$ has modular width at most $\omega$). Again, the structure of recursive calls corresponds to a decomposition-tree $T$ on at most $2n-1$ vertices. The computation corresponding to one of the $n$ leafs of the tree terminates after constant time, and so we get an upper bound of 
	$$\mathcal{O}(|V(T)|(n^2+f(\omega) n\log n)+n)=\mathcal{O}(n^3+f(\omega) n^2\log n)$$ on the required runtime of the algorithm. This concludes the proof of the Theorem. 
\end{proof}

We conclude by noting the following two direct consequences of Theorem \cref{hamlecker} and \cref{hamprogram}.

\begin{corollary}
There exists an algorithms for testing the existence of a Hamiltonian path or of a Hamiltonian cycle in a digraph $D$ in time $\mathcal{O}(n^3+f(\omega) n^2\log n)$ where $n\coloneqq |V(D)|$, $\omega\coloneqq \dmw(D)$ and $f(\omega)=2^{\mathcal{O}(\omega^2 \log \omega)}$. 
\end{corollary}

\section{Disjoint Directed Paths} \label{vddp}
An important computational problem for digraphs is the so-called \emph{disjoint directed paths problem}. In this problem, as an instance, we are given a digraph $D$ and pairwise disjoint pairs $(s_1,t_1), \ldots, (s_r,t_r)$ of distinct vertices. Our task is to determine whether there exist pairwise vertex-disjoint directed paths $P_1,\ldots, P_r$ in $D$ such that $P_i$ starts in $s_i$ and ends in $t_i$ for all $i \in [r]$, and, if they exist, to find such a collection of paths. 

\ProblemDefLabelled{$r$ Vertex-Disjoint Directed Paths}{$r$-VDDP}{prob:rVDDP}
{A digraph $D$, disjoint pairs $(s_1,t_1), \ldots, (s_r,t_r)$ of vertices.}
{Does there exist a collection $P_1, \ldots, P_r$ of pairwise vertex-disjoint directed paths such that $P_i$ starts in $s_i$ and ends in $t_i$, for all $i \in [r]$? 
	
	If so, find one.}

The disjoint directed paths problem has long been known to be very hard algorithmically even for $r=2$. In the following we cite some important negative and positive results from literature. For an overview of the topic, we refer to the articles \cite{frank} and \cite{robseym}. 
\begin{itemize}
	\item $1$-VDDP is polynomial-time solvable for arbitrary digraphs\footnote{Reachability in directed graphs is a simple special case of the max-flow problem and can be solved using one of the well-known polynomial algorithms for this task (see for instance \cite{edmonds}).}. 
	\item For general digraphs, the existence version of \ref{prob:rVDDP} is NP-complete, for any fixed $r \ge 2$ (\cite{fortune}).
	\item There is an FPT-algorithm solving $r$-VDPP with respect to parameter $r$ for planar digraphs \cite{planardisjointpaths}.
	\item The \ref{prob:rVDDP} is polynomial-time solvable on DAGs (\cite{fortune}), for any fixed $r \ge 1$. More generally, for any fixed $r \ge 1$, there exists an XP-algorithm for \ref{prob:rVDDP} when parametrising with the directed treewidth $\dirtw(D)$ of the digraph \cite{dirtreewidth}. No FPT-algorithms with respect to parameters $r$ and $\dirtw$ are known to exist.
\end{itemize}
As the main contribution of this section, we show that \ref{prob:rVDDP} on arbitrary digraphs is fixed-parameter tractable with respect to the parameters $r$ and directed modular width. In fact what we prove implies a more general result, see \cref{sec:homeo}. Similar to \cref{digraphcolouring}, we consider a well-chosen weighted generalisation of \ref{prob:rVDDP} which allows a recursive approach. 

In the following, given a digraph $D$ and vertices $s, t \in V(D)$, we will use the term \emph{$s$-$t$-path} in order to refer to a usual directed path starting in $s$ and ending in $t$ in the case of $s\neq t$, and to a directed cycle with a designated 'beginning' and 'end' at the vertex $s=t$ otherwise.

For the generalised problem, it will be necessary to allow that in the given pairs $(s_1,t_1), \ldots, (s_r,t_r)$ of vertices, several vertices coincide, i.e., that we have $s_i=t_j$, $s_i=s_j$, or $t_i=t_j$ for some $i,j \in [r]$. Our goal will be to find an $s_i$-$t_i$-path $P_i$ in $D$ for every $i \in [r]$. However, in general, we do not require the paths to be disjoint any more, they are allowed to share vertices. Instead, we bound the number of times a vertex can be traversed by the paths in total by a respective non-negative integer-capacity. In the special case where all the capacities are equal to $1$, we get disjoint paths (as in the usual setting). For an $s$-$t$-path with $s=t$, we count two traversals of the vertex $s=t$ (at the beginning and at the end).
For an illustration of this generalised version of \ref{prob:rVDDP} together with a feasible solution in which one of the $\Brace{s_i,t_i}$ pairs has to be connected by a directed cycle (namely $\Brace{s_2,t_2}$ in this example) consider \cref{fig:rVDDPCexample}.

The last important alteration of the original problem is that in addition to testing whether a collection of $s_i$-$t_i$-paths as required exists, in case it does, we want to find one that minimises the \emph{sum of the sizes} of the paths. In the following, the \emph{size} of an $s$-$t$-path $P$ shall be defined as $|P|\coloneqq |E(P)|+1$. $|P|$ is defined such that (independent of whether $s=t$ or $s \neq t$), it counts the total number of traversals of vertices by $P$.

For the purpose of applying the method from \cref{ipsolving}, we again use a threshold $\tau \in \mathbb{N}$ as an input, which bounds the sum of the vertex-capacities.

To formulate our generalised problem properly, we need the following terminology.

Given a list $\mathcal{S}=[(s_1,t_1), \ldots, (s_r,t_r)]$ of vertex-pairs equipped with vertex-capacities $w:V(D) \rightarrow \mathbb{N}_0$ of a digraph $D$, let us say that a collection $P_1,\ldots,P_r$ of $s_i$-$t_i$-paths in $D$ is \emph{feasible} with respect to $(w, \mathcal{S})$, if any vertex $z \in V(D)$ is traversed at most $w(z)$ times in total by the collection. Furthermore, let us define $W(D,w,\mathcal{S})$ to be the minimum of $$W(P_1,\ldots,P_r)\coloneqq \sum_{i=1}^{r}{|P_i|}$$ over all feasible collections $\mathcal{P}=(P_1,\ldots,P_r)$ of $s_i$-$t_i$-paths in $D$. If such a collection of $s_i$-$t_i$-paths does not exist, by convention, we put $W(D,w,\mathcal{S})=\infty$. A feasible collection $\mathcal{P}$ for the pair $(w,\mathcal{S})$ is called \emph{optimal}, if it attains the minimum, i.e., $W(\mathcal{P})=W(D,w,\mathcal{S})$.

\begin{figure}[h!]
	\begin{center}
		\begin{tikzpicture}[scale=1]
		
		\pgfdeclarelayer{background}
		\pgfdeclarelayer{foreground}
		
		\pgfsetlayers{background,main,foreground}
		
		\begin{pgfonlayer}{main}
		
		\node (C) [] {};
		
		\node (C1) [v:ghost, position=180:25mm from C] {};
		
		\node (C2) [v:ghost, position=0:0mm from C] {};
		
		\node (C3) [v:ghost, position=0:25mm from C] {};

		

		
		
		\node (v1) [v:main,position=0:0mm from C2,fill=white,scale=1] {$2$};
		\node (v2) [v:main,position=120:15mm from v1,fill=white,scale=1] {$2$};
		\node (v3) [v:main,position=180:15mm from v1,fill=white,scale=1] {$3$};
		\node (v4) [v:main,position=240:15mm from v1,fill=white,scale=1] {$2$};
		\node (v5) [v:main,position=0:15mm from v1,fill=white,scale=1] {$2$};
		\node (v6) [v:main,position=270:15mm from v5,fill=white,scale=1] {$1$};
		\node (v7) [v:main,position=0:15mm from v5,fill=white,scale=1] {$1$};
		\node (v8) [v:main,position=90:15mm from v5,fill=white,scale=1] {$1$};
		
		\node (Ls1) [v:ghost,position=135:4mm from v2,font=\scriptsize] {$s_1$};
		\node (Ls2t2) [v:ghost,position=270:3.5mm from v4,font=\scriptsize] {$s_2=t_2$};
		\node (Lt1) [v:ghost,position=270:3.5mm from v6,font=\scriptsize] {$t_1$};
		\node (Ls3) [v:ghost,position=225:4mm from v3,font=\scriptsize] {$s_3$};
		\node (Lt3) [v:ghost,position=0:4mm from v7,font=\scriptsize] {$t_3$};
		
		
		
		

		

		
		
		\draw (v1) [e:main,->,bend right=15] to (v2);
		\draw (v1) [e:main,->] to (v5);
		\draw (v1) [e:main,->,bend right=15] to (v6);
		
		\draw (v2) [e:main,->,bend right=15] to (v1);
		\draw (v2) [e:main,->] to (v3);
		\draw (v2) [e:main,->] to (v8);
		
		\draw (v3) [e:main,->] to (v1);
		\draw (v3) [e:main,->] to (v4);
		
		\draw (v4) [e:main,->] to (v1);
		\draw (v4) [e:main,->] to (v6);
		
		\draw (v5) [e:main,->] to (v7);
		\draw (v5) [e:main,->,bend right=15] to (v8);
		\draw (v5) [e:main,->] to (v6);
		
		\draw (v6) [e:main,->,bend right=15] to (v1);
		\draw (v6) [e:main,->] to (v7);
		
		\draw (v7) [e:main,->] to (v8);
		
		\draw (v8) [e:main,->] to (v1);
		\draw (v8) [e:main,->,bend right=15] to (v5);

		

		
		
		\end{pgfonlayer}
		

		\begin{pgfonlayer}{background}
		
		\draw (v1) [e:marker,color=CornflowerBlue,bend right=15] to (v2);
		\draw (v2) [e:marker,color=CornflowerBlue] to (v3);
		\draw (v3) [e:marker,color=CornflowerBlue] to (v4);
		\draw (v4) [e:marker,color=CornflowerBlue] to (v1);
		
		\draw (v2) [e:marker,color=magenta] to (v8);
		\draw (v8) [e:marker,color=magenta,bend right=15] to (v5);
		\draw (v5) [e:marker,color=magenta] to (v6);
		
		\draw (v3) [e:marker,color=Amber!85!black] to (v1);
		\draw (v1) [e:marker,color=Amber!85!black] to (v5);
		\draw (v5) [e:marker,color=Amber!85!black] to (v7);
		
		\end{pgfonlayer}	
		
		\begin{pgfonlayer}{foreground}

		\end{pgfonlayer}
		\end{tikzpicture}
	\end{center}
	\caption{An example of an \ref{prob:rVDDP-C} instance together with a feasible solution.}
	\label{fig:rVDDPCexample}
\end{figure}
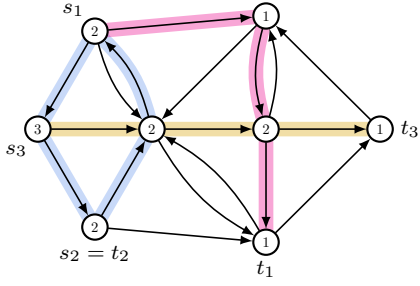

\ProblemDefLabelled{$r$ Vertex-Disjoint Directed Paths (Capacity Version)}{$r$-VDDP-C}{prob:rVDDP-C}
{A digraph $D$, a list $\mathcal{S}=[(s_1,t_1), \ldots, (s_r,t_r)]$ of not necessarily disjoint pairs of not necessarily distinct vertices in $D$, non-negative capacities $w:V(D) \rightarrow \mathbb{N}_0$ on the vertices and a threshold $\tau \in \mathbb{N}$ such that $\sum_{z \in V(D)}{w(z)} \leq \tau$.}
{Does there exist a feasible collection $P_1,\ldots,P_r$ compatible with $(w,\mathcal{S})$? 
	
	If so, output the value of $W(D,w,\mathcal{S})$ and a corresponding optimal collection.}
\begin{theorem}\label{algdisjpaths}
	There exists an algorithm solving the capacity-version of \ref{prob:rVDDP} which runs in time $\mathcal{O}\left(n^3+f(r,\omega)(n^2+n\log \tau)\right)$, where $n\coloneqq |V(D)|$ and $\omega\coloneqq \dmw(D)$, and $f(r,\omega)=2^{\mathcal{O}(r \log r \cdot 2^\omega \omega)}$.
\end{theorem}
Putting $w(v)\coloneqq 1, v \in V(D)$ we can take $\tau=n$, and obtain:
\begin{corollary}
	There exists an algorithm that solves \ref{prob:rVDDP} and runs in time $\mathcal{O}\left(n^3+f(r,\omega)n^2\right)$, where $n$ is the number of vertices in the input digraph, $\omega$ denotes its directed modular width, and $f(r,\omega)=2^{\mathcal{O}(r \log r \cdot 2^\omega \omega)}$ is some function.
\end{corollary}
The rest of this section is devoted to the proof of \cref{algdisjpaths}. We again aim at recursively reducing the problem to the subdigraphs induced by the modules. To do so however requires a number of quite technical auxiliary statements. The final verification of \cref{algdisjpaths} comes with the proof of \cref{subcollections}, where we construct an algorithm that solves a slightly more general version of \ref{prob:rVDDP-C} by recursively computing optimal collections of paths and cycles not only for one, but simultaneously for a set of $2^r$ well-chosen adapted inputs, including the original input. 

We start with some simple observations concerning the structure of $s$-$t$-paths with respect to a module-decomposition.

In the following, we call an $s$-$t$-path $P$ in a digraph $D$ \emph{reduced}, if there is no $s$-$t$-path $P'$ in $D$ such that $V(P') \subsetneq V(P)$. 

If a collection $\mathcal{P}=(P_1,\ldots,P_r)$ of paths in a digraph $D$ is optimal with respect to a pair $(w,\mathcal{S})$, it is easy to see that every $s_i$-$t_i$-path $P_i$ has to be \emph{reduced}. For if there was an $i \in [r]$ and an $s_i$-$t_i$-path $P'$ in $D$ such that $V(P') \subsetneq V(P_i)$, we could replace $P_i$ by $P'$ in the collection $\mathcal{P}$ to obtain a collection $\mathcal{P}'=(P_1, \ldots, P_{i-1},P',P_{i+1}, \ldots, P_r)$ which clearly still is feasible for the pair $(w,\mathcal{S})$ but has $W(\mathcal{P}')-W(\mathcal{P})=|P'|-|P_i|<0$, a contradiction to the assumed optimality. In the following, we call a collection $\mathcal{P}$ of $s_i$-$t_i$-paths \emph{reduced}, if every member of $\mathcal{P}$ is reduced.

\begin{lemma} \label{reducedpaths}
	Let $D$ be a digraph and let $M_1,\ldots,M_\ell$ be a partition of $V(D)$ into modules. Let $s, t \in V(D)$ be arbitrary, and let $P$ be an $s$-$t$-path in $D$ such that at least one of the following holds:
	\begin{itemize}
		\item $P$ is reduced, or
		\item $s, t \in M_j$ for some $j \in [\ell]$, $V(P) \setminus M_j \neq \emptyset$ and there is no $s$-$t$-path $P'$ in $D$ with $V(P') \subsetneq V(P)$ which uses at least one vertex outside the module $M_j$.
	\end{itemize}
	
	Then one of the following situations appears:
	\begin{itemize}
		\item $P$ is fully contained in $D[M_i]$ for some $i \in [\ell]$. 
		\item $P$ uses at most one vertex from each module.
		\item $s \neq t$, $V(P) \cap M_i=\{s,t\}$ for some $i \in [\ell]$, and $|V(P) \cap M_k| \leq 1$ for all $k \in [\ell], k \neq i$. 
	\end{itemize}
\end{lemma}
\begin{proof}
	Assume that $P$ is not fully contained in any module, but still there exists some module $M_i$ that hosts at least two different vertices of $P$. We want to prove that now the third of the cases in the Lemma applies.
	
	\textbf{Claim 1: $s \neq t$.}

	Assume towards a contradiction that $s=t$. Now consider the module $M_j$ containing $s=t$. $M_j$ cannot contain any further vertices from $P$: If $|V(P) \cap M_j| \ge 2$, we could find vertices $x \in M_j \setminus \{s\}$ and $y \not\in M_j$ such that $(x,y) \in E(P)$ or $(y,x) \in E(P)$. By the modular property, this would imply that $(s,y) \in E(D) \setminus E(P)$ respectively $(y,s) \in E(D) \setminus E(P)$. However, this means that $P$ would have a chord between $s=t$ and another vertex on $P$. This directly yields that there is an $s$-$t$-path using a proper subset of the vertices of $P$ and which contains a vertex outside $M_j$ (namely $y$), a contradiction.
	
	Now that we know that $V(P) \cap M_j=\{s\}$, it follows that $M_i \neq M_j$. Let $x$ be the first vertex in $M_i$ that is met when traversing $P$ in forward-direction starting with $s$, and let $y \in (V(P) \cap M_i) \setminus \{x\}$ be some vertex distinct from $x$. Let $x' \in V(P) \setminus M_i$ be the predecessor of $x$ on $P$. By the modular property, $(x',x) \in E(D)$ implies that also $(x',y) \in E(D)$, and hence, replacing the directed subpath of $P$ from $x'$ to $y$ by the short-cut-edge $(x',y)$ yields an $s$-$t$-path in $D$ using a proper subset of the vertices of $P$ and which contains a vertex outside $M_j$ (namely $y$), again, a contradiction. This contradiction shows that indeed, under the initial assumptions, we must have $s \neq t$.

	\textbf{Claim 2: $s, t \in M_i$.}

	Assume towards a contradiction that at least one of $s, t$ is not contained in $M_i$. By symmetry (reversal of all edges in $D$), we may assume that $s \notin M_i$. At least two vertices of $P$ are contained in $M_i$. Because of $s \not \in M_i$, the first vertex $x \in M_i$ that is met when traversing $P$ starting with $s$ has a predecessor $x'$ not contained in $M_i$. Let $y \in (V(P) \cap M_i)\setminus \{x\}$ some other element (which has to lie on the subpath $P[x,t]$ of $P$). As above, by the modular property, $(x',x) \in E(D)$ implies that also $(x',y) \in E(D)$. This is a short-cut which proves the existence of an $s$-$t$-path using a strict subset of the vertices of $P$, and again, we have a contradiction.

	\textbf{Claim 3: $|V(P) \cap M_k| \leq 1$ for all $k \neq i$.}

	Assume towards a contradiction there was a module $M_k, k \neq i$ containing two different vertices from $P$. Let $y_1$ be the first vertex in $M_k$ that is met when traversing $P$ starting at $s$, and let $y_2$ be the first vertex in $M_k$ that is met when traversing $P$ in backward-direction starting at $t$. Because of $|V(P) \cap M_k| \ge 2$, we conclude that $y_1 \neq y_2$. Because $y_1 \neq s$ and $y_2 \neq t$, there is a predecessor $x_1$ of $y_1$ on $P$ and a successor $x_2$ of $y_2$ on $P$. We have $x_1, x_2 \notin M_k$ and $(x_1,y_1), (y_2,x_2) \in E(P)$. By the modular property, this implies that also $(x_1,y_2),(y_1,x_2) \in E(D)$. Clearly, now replacing the directed subpath $P[x_1,y_2]$ by the edge $(x_1,y_2)$ defines an $s$-$t$-path $P'$ using a proper subset of the vertices of $P$, which uses vertices belonging to different modules (namely $x_1,y_2$), a contradiction. 
	
	\textbf{Claim 4: $V(P) \cap M_i=\{s,t\}$.}
	
	Assume towards a contradiction that $(V(P) \cap M_i)\setminus \{s,t\} \neq \emptyset$. Because there exist vertices on $P$ outside the module, we conclude that there are vertices $x \in M_i \setminus \{s,t\}, y \notin M_i$ such that $(x,y) \in E(P)$ or $(y,x) \in E(P)$. In the first case, we conclude that $(s, y) \in E(D)$, and in the second that $(y,t) \in E(D)$. In each case, we can take the corresponding edge as a short-cut in order to obtain a directed $s$-$t$-path $P'$ whose vertex set is properly contained in $V(P)$. In addition, $P'$ contains vertices from different modules, namely we have $y \notin M_i, s$ and $t \in M_i$. This contradiction verifies the claim.
	
	All in all, we have shown that if neither of the first two cases in the Lemma holds, we have to be in the situation of the third case. This concludes the proof. 
\end{proof}

\begin{lemma}\label{disjpathrecursion}
	Let $D$ be a digraph equipped with a partition $M_1, \ldots, M_\ell$ of $V(D)$ into modules. Let $D_M$ be the corresponding module-digraph and let $w:V(D) \rightarrow \mathbb{N}$ be a given assignment of non-negative integer-capacities to the vertices. Let $\mathcal{S}=[(s_1,t_1), \ldots, (s_r,t_r)]$ be a given list of vertex-tuples. Assume that $s_i, t_i$ lie in the same module for all $1 \leq i \leq r'$ and in different modules for all $i>r'$, where $r' \in \{0,1,\ldots,r\}$. 
	
	For a fixed subset $B \subseteq \{1,\ldots,r'\}$ let $\mathcal{S}_M^B\coloneqq [(\eta(s_i),\eta(t_i))~|~i \in B \cup \{r'+1,\ldots,r\}]$ denote an associated list of vertex-tuples in $D_M$. 
	For a vertex $v_k \in V(D_M)$ corresponding to module $M_k$, let $\mathcal{S}_k^B\coloneqq [(s_i,t_i)~|~s_i, t_i \in M_k, i \notin B]$ be a list of vertex-tuples within the module. 
	
	For every $k \in [\ell]$, define integer-capacities on the vertices of $D[M_k]$ by
	$$w_k^B(z)\coloneqq w(z)-|\CondSet{i \in B \cup \{r'+1,\ldots,r\}}{s_i=z}|-|\CondSet{i \in B \cup \{r'+1,\ldots,r\}}{t_i=z}|$$ for all $z \in M_k$.
	Furthermore, define integer-capacities on the vertices of $D_M$ according to 
	$$w_M^B(v_k)\coloneqq \sum_{z \in M_k}{w(z)}-W(D[M_k],w_k^B,\mathcal{S}_k^B)$$
	for all $k \in [\ell]$. 
	
	Then for any fixed set $B \subseteq [r']$, the following two statements are equivalent:
	\begin{itemize}
		\item There exists a collection $\mathcal{P}=(P_1,\ldots,P_r)$ of directed paths and cycles in $D$ feasible for $(w,\mathcal{S})$, such that $P_i$ stays within the module containing $s_i$ and $t_i$ for all $i \in [r'] \setminus B$, while it has a vertex outside the module if $i \in B$.
		\item We have $w_k^B(z) \ge 0$ for all $z \in M_k$, $W(D[M_k],w_k^B,\mathcal{S}_k^B)<\infty$ and $w_M^B(v_k)\ge 0$ for all $k \in [\ell]$, and $W(D_M,w_M^B,\mathcal{S}_M^B)<\infty$. 
	\end{itemize}
	In the case that both of the above statements hold true, the minimum value $W_B(D,w,\mathcal{S})$ of $W(\mathcal{P})$ among all collections $\mathcal{P}$ of directed cycles and paths in $D$ which are feasible for the pair $(w,\mathcal{S})$ and satisfy the additional properties described in the first item is given by
	$$W_B(D,w,\mathcal{S})=W(D_M,w_M^B,\mathcal{S}_M^B)+\sum_{k=1}^{\ell}{W(D[M_k],w_k^B,\mathcal{S}_k^B)}.$$
\end{lemma}
\begin{proof}
	We start with proving both directions of the claimed equivalence. 
	
	$\Longrightarrow$ Assume for the first implication that there exists a collection $\mathcal{P}=(P_1,\ldots,P_r)$ of directed paths and cycles in $D$ feasible for $(w,\mathcal{S})$, such that $P_i$ stays within the module containing $s_i$ and $t_i$ for all $i \in [r'] \setminus B$, while it has a vertex outside the module if $i \in B$. Let us furthermore choose $\mathcal{P}$ such that it minimises $W(\mathcal{P})$ among all collections with these additional properties. This assumption implies that $P_i$ is reduced for every $i>r'$ and every $i \in [r'] \setminus B$, and that for every $i \in B$, there is no directed $s_i$-$t_i$-path $P'$ in $D$ such that $V(P') \subsetneq V(P)$ which also has at least one vertex outside the module containing $s_i$ and  $t_i$. Therefore, by \cref{reducedpaths}, $P_i$ uses at most one vertex from each module if $i>r'$. If $i \in [r'] \setminus B$, then $P_i$ uses at most one vertex from each module not containing $s_i, t_i$ and the module including $s_i,t_i$ has no further vertices on $P_i$.
	
We have to prove that $w_k^B(z) \ge 0$ for all $z \in M_k$ and that there exist feasible collections $\mathcal{P}_k^B$ in $D[M_k]$ for the pairs $(w_k^B,\mathcal{S}_k^B)$ and that $w_M^B(v_k) \ge 0$ for all $k \in [\ell]$. Furthermore, we must show that there exists a feasible collection $\mathcal{P}_M^B$ in $D_M$ for the pair $(w_M^B,\mathcal{S}_M^B)$. 
	
	First of all, it is readily verified that we must have $$|\CondSet{i \in B \cup \{r'+1,\ldots,r\}}{s_i=z}|+|\CondSet{i \in B \cup \{r'+1,\ldots,r\}}{t_i=z}| \leq w(z)$$ for every $z \in V(D)$, for otherwise, the $s_i$-$t_i$-paths $P_i$ with $i \in B$ and $s_i=z$ or $t_i=z$ together would traverse $z$ more than $w(z)$ times, a contradiction to the feasibility of $\mathcal{P}$ with respect to $(w,\mathcal{S})$. This shows that $w_k^B$ is non-negative for all $k \in [\ell]$.
	
	Now let $k \in [\ell]$ be arbitrary but fixed, and consider a vertex-tuple $(s_i,t_i) \in \mathcal{S}_k^B$. By definition, we have that $i \in [r'] \setminus B$. The latter implies that the path $P_i$ stays within the module $M_k$, and therefore defines an $s_i$-$t_i$-path in $D[M_i]$. Let $$\mathcal{P}_k^B\coloneqq (P_i~|~i \in [r'] \setminus B, s_i, t_i \in M_k).$$ This is a feasible collection of directed paths and cycles for $(w_k^B,\mathcal{S}_k^B)$. To see this, consider an arbitrary vertex $z \in M_k$. It is traversed at least $$|\CondSet{i \in B \cup \{r'+1,\ldots,r\}}{s_i=z}|+|\CondSet{i \in B \cup \{r'+1,\ldots,r\}}{t_i=z}|$$ times by the paths $P_i$ with $i \in B \cup \{r'+1,\ldots,r\}$, and at most $w(z)$ times in total. Therefore, the collection $\mathcal{P}_k^B$ does not traverse $z$ more than $w_k^B(z)$ times.
	
	To prove that $w_M^B(v_k) \ge 0$ for all $k \in [\ell]$, we verify that $W(D[M_k],w_k^B,\mathcal{S}_k^B) \leq W(\mathcal{P}_k^B) \leq \sum_{z \in M_k}{w(z)}$. However, this follows directly by double-counting the total number of traversals of vertices in $D[M_k]$ by members of the collection $\mathcal{P}_k^B$: In total, each $s_i$-$t_i$-path $P_i \in \mathcal{P}_k^B$ contributes exactly $|P_i|$ traversals of vertices, but every vertex $z \in M_k$ is traversed at most $w(z)$ times. 
	
	For the last part, we have to construct a collection $\mathcal{P}_M^B$ of directed paths and cycles in $D_M$ which is feasible for $(w_M^B, \mathcal{S}_M^B)$. For this purpose, let $i \in B \cup \{r'+1,\ldots,r\}$ be arbitrary. 
	
If $i \in B$, then we know that $s_i,t_i$ lie together in some module $M_k$ and are connected via the $s_i$-$t_i$-path $P_i \in \mathcal{P}$ in $D$, which uses at least one vertex outside $M_k$. As observed at the beginning of the proof, $P_i$ uses at most one vertex from each module different from $M_k$ and has no further vertices in $M_k$. Let $s_i=u_1,u_2,\ldots,u_p=t_i$ be the sequence of vertices in $P_i$. We now define $\eta(P_i)$ to be the $\eta(s_i)$-$\eta(t_i)$-path in $D_M$ described by the sequence $\eta(s_i)=\eta(u_1),\eta(u_2),\ldots,\eta(u_p)=\eta(t_i)$ of vertices. 

On the other hand, if $i>r'$, as observed at the beginning of the proof, $P_i$ is reduced and by \cref{reducedpaths} uses at most one vertex from each module. Thus, we may define the $\eta(s_i)-\eta(t_i)$-path $\eta(P_i)$ in $D_M$ by replacing each vertex $x \in V(D)$ by $\eta(x)$.
	
	We now claim that the so-defined collection $\mathcal{P}_M^B\coloneqq (\eta(P_i)~|~i \in B \cup \{r'+1,\ldots,r\})$ does the job. We must prove that for any $k \in [\ell]$, the vertex $v_k \in M_k$ is traversed at most $w_M^B(v_k)$ times by the paths $\eta(P_i)$. For this purpose, we again double-count the number of traversals of vertices within the module $M_k$ by paths in $\mathcal{P}$. By the feasibility of $\mathcal{P}$, this number can be at most $\sum_{z \in M_k}{w(z)}$. On the other hand, the paths $P_i \in \mathcal{P}_k^B$ together contribute at least $\sum_{P \in \mathcal{P}_k^B}{|P|} \ge W(D[M_k],w_k^B,\mathcal{S}_k^B)$ traversals, and thus, the paths $P_i \in \mathcal{P}$ with $i \in B \cup \{r'+1,\ldots,r\}$ cannot contribute more than $\sum_{z \in M_k}{w(z)}-W(D[M_k],w_k^B,\mathcal{S}_k^B)=w_M^B(v_k)$ traversals in total. Because $\eta(P_i)$ visits $v_k$ exactly as often as $P_i$ visits vertices in $M_k$, for all $i \in B \cup \{r'+1,\ldots,r\}$, this verifies the last claim. This concludes the proof of the first implication.
	
	At this point, we furthermore note that the existence of the constructed path-collections $\mathcal{P}_k^B, k\in[\ell]$ and $\mathcal{P}_M^B$ imply that
	$$\sum_{k=1}^{\ell}{W(D[M_k],w_k^B,\mathcal{S}_k^B)} \leq \sum_{k=1}^{\ell}{W(\mathcal{P}_k^B)}=\sum_{i \in [r'] \setminus B}{|P_i|}$$ and
	$$W(D_M,w_M^B,\mathcal{S}_M^B) \leq W(\mathcal{P}_M^B)=\sum_{i \in B \cup \{r'+1,\ldots, r\}}{|\eta(P_i)|}=\sum_{i \in B \cup \{r'+1,\ldots, r\}}{|P_i|}.$$
	Finally, this implies that
	$$W(D_M,w_M^B,\mathcal{S}_M^B)+\sum_{k=1}^{\ell}{W(D[M_k],w_k^B,\mathcal{S}_k^B)} \leq \sum_{i \in [r]}{|P_i|}=W(\mathcal{P})=W_B(D,w,\mathcal{S}),$$ where the last equality holds because $\mathcal{P}$ was chosen in such a way that $W(\mathcal{P})$ is minimised among all feasible collections for $(D,w,\mathcal{S})$ with the additional properties prescribed by the set $B$.
	
	\bigskip
	
	$\Longleftarrow$ Assume that $w_k^B(z) \ge 0$ for all $z \in M_k$, $W(D[M_k],w_k^B,\mathcal{S}_k^B) <\infty$ and $w_M^B(v_k) \ge 0$ for all $k \in [\ell]$. Furthermore assume that $W(D_M,w_M^B,\mathcal{S}_M^B)<\infty$. 
	
	Let $\mathcal{P}_k^B$ be an optimal collection of paths for $(w_k^B,\mathcal{S}_k^B)$ for all $k \in [\ell]$ and $\mathcal{P}_M^B$ an optimal collection for $(w_M^B, \mathcal{S}_M^B)$. For any $i \in B \cup \{r'+1,\ldots,r\}$, let $Q_i$ denote the $\eta(s_i)$-$\eta(t_i)$-path in the collection $\mathcal{P}_M^B$. 
	
	In the following, we define a collection $\mathcal{P}=(P_i~|~i \in [r])$ of $s_i$-$t_i$-paths in $D$ which is feasible with respect to $(w,\mathcal{S})$ by successively adding new elements to the collection. In order to ensure that in the end, every vertex $z \in V(D)$ is traversed at most $w(z)$ times, during the process we keep track of the capacities on the vertices which are still left for the remaining paths that have to be routed. 
	
	Formally, we start with the initial capacities $w:V(D) \rightarrow \mathbb{N}_0$ on the vertices and an empty collection. Whenever a new $s_i$-$t_i$-path $P_i$ is added to the collection, we update the capacities on the vertices by reducing the capacity of each vertex $x \in V(P_i)$ by $1$ (resp. $2$, if $x=s_i=t_i$). We will then show that we can add a new path $P_i$ to the collection and route it through $D$ such that certain invariants keep satisfied. The invariants ensure that the reduced capacities are always non-negative, and therefore, in the end, we obtain a feasible collection in which no vertex $z \in V(D)$ is traversed more than $w(z)$ times. 
	
	We start by going through the different modules $M_k, k=1,\ldots,\ell$, and add all the directed paths and cycles contained in $\mathcal{P}_k^B$, which are also directed paths and cycles in $D$, to the collection.
	
	At this stage, all the pairs $(s_i,t_i)$ with $i \in [r'] \setminus B$ are connected. For every $k \in [\ell]$, the sum of the remaining capacities of vertices in $M_k$ is $$\sum_{z \in M_k}{w(z)}-\sum_{P \in \mathcal{P}_k^B}{|P|}=w_M^B(v_k) \ge 0.$$ Furthermore, by the definition of $w_k^B$, we know that for every vertex $z \in V(D)$, the remaining capacity is at least $$|\CondSet{i \in B \cup \{r'+1,\ldots,r\}}{s_i=z}|+|\CondSet{i \in B \cup \{r'+1,\ldots,r\}}{t_i=z}|.$$
	
	Now we go through the remaining pairs $(s_{i'},t_{i'}), i' \in B \cup \{r'+1,\ldots,r\}$ and successively have to add a new $s_i$-$t_i$ path $P_i$ to the collection. Let at each stage $I \subseteq B \cup \{r'+1,\ldots,r\}$ denote the set of indices $i'$ corresponding to $s_{i'}$-$t_{i'}$-paths $P_{i'}$ that are already contained in the current collection, and let $J\coloneqq (B \cup \{r'+1,\ldots,r\}) \setminus I$ correspond to the paths that still have to be routed. Denote by $R_{(I,J)}(z) \ge 0$ for every $z \in V(D)$ the remaining capacity of $z$ at this state. For every $k \in [\ell]$, we count by
	$$\text{trav}_{(I,J)}(v_k)\coloneqq 2|\CondSet{i' \in I}{\{s_{i'},t_{i'}\} \subseteq M_k}|+|\CondSet{i' \in I}{v_k \in V(Q_{i'}) \text{ and }\{s_{i'}, t_{i'}\} \not\subseteq M_k }|$$ the number of traversals of the vertex $v_k \in V(D_M)$ by the paths $(Q_{i'}|i' \in I)$. 
	
	During the whole process, we maintain the following invariants:
	
	\emph{For every $k \in [\ell]$, we have
		$$\sum_{z \in M_k}{R_{(I,J)}(z)} \ge w_M^B(v_k)-\text{trav}_{(I,J)}(v_k).$$ 
		Furthermore, the remaining capacity at any vertex $z \in V(D)$ is at least $$R_{(I,J)}(z) \ge |\CondSet{i' \in J}{s_{i'}=z}|+|\CondSet{i' \in J}{t_{i'}=z}|.$$}
	By the above, these properties are fulfilled in the beginning when $I=\emptyset, J=B \cup \{r'+1,\ldots,r\}$. As long as $J \neq \emptyset$, we continue to choose some $i \in J$, whose corresponding path is to be added next to the collection. Let $\eta(s_i)=v_{k_1},v_{k_2},\ldots,v_{k_j}=\eta(t_i)$ be the vertex-sequence of $Q_i$. The $s_i$-$t_i$-path $P_i$ we will add has a vertex-sequence of the form $s_i=u_1,u_{2},\ldots,u_j=t_i$, where $u_1 \in M_{k_1}, u_{2} \in M_{k_2}, \ldots, u_{j} \in M_{k_j}$. By the definition of the module-digraph, this sequence defines an $s_i$-$t_i$-path in $D$, independent of which vertex $u_{j'}$ is selected from the module $M_{k_{j'}}$, $j'\in\Set{2,\ldots,j-1}$. 
	
	The following claim describes the property according to which we select the vertices.
	\begin{claim} 
		\textit{In each module $M_{j'}, 2 \leq j' \leq j-1$, there exists a vertex $u_{j'}$ such that
			$$R_{(I,J)}(u_{j'})>|\CondSet{i' \in J \setminus \{i\}}{s_{i'}=u_{j'}}|+|\CondSet{i' \in J \setminus \{i\}}{t_{i'}=u_{j'}}|.$$}
	\end{claim}
	\begin{proof}
	For simpler notation, assume for the proof of this claim that $k_{j'}=j', 2 \leq j' \leq j-1$ (if necessary, we relabel).
		Let $j' \in \{2,\ldots,j-1\}$ be fixed. To see that a selection of $u_{j'}$ as claimed is possible, let us consider the sum
		$$\sum_{z \in M_{j'}}{\left(R_{(I,J)}(z)-|\CondSet{i' \in J \setminus \{i\}}{s_{i'}=z}|-|\CondSet{i' \in J \setminus \{i\}}{t_{i'}=z}|\right)}$$ and prove that it is strictly positive. 
		
		The feasibility of the collection $\mathcal{P}_M^B$ of paths in $D_M$ with respect to $(w_M^B, \mathcal{S}_M^B)$, implies that $v_{j'} \in V(D_M)$ is traversed at most $w_M^B(v_{j'})$ times in total by the paths $Q_{i'}$, $i' \in B \cup \{r'+1,\ldots,r\}$. The number of traversals by the sub-collection $(Q_{i'}|i' \in I)$ is counted by $\text{trav}_{(I,J)}(v_{j'})$. Clearly,
		$$\sum_{z \in M_{j'}}{\left(|\CondSet{i' \in J \setminus \{i\}}{s_{i'}=z}|+|\CondSet{i' \in J \setminus \{i\}}{t_{i'}=z}| \right)}$$ is a lower bound on the number of traversals of the vertex $v_{j'}$ by the sub-collection $(Q_{i'}|i' \in J \setminus \{i\})$. Also $Q_i$ traverses $v_{j'}$ once. All in all, this means that
		$$w_M^B(v_{j'}) \ge \text{trav}_{(I,J)}(v_{j'})+\sum_{z \in M_{j'}}{\left(|\CondSet{i' \in J \setminus \{i\}}{s_{i'}=z}|+|\CondSet{i' \in J \setminus \{i\}}{t_{i'}=z}|\right)}+1.$$
		
		By the first invariant, we know that 
		\begin{align*}
		\sum_{z \in M_{j'}}{R_{(I,J)}(z)} &\geq w_M^B(v_{j'})-\text{trav}_{(I,J)}(v_{j'})\\
		&> \sum_{z \in M_{j'}}{\left(|\CondSet{i' \in J \setminus \{i\}}{s_{i'}=z}|+|\CondSet{i' \in J \setminus \{i\}}{t_{i'}=z}|\right)}.
		\end{align*}
		Therefore, the sum we initially considered is positive, and the claim follows.
	\end{proof}
	For each $j' \in \{2,\ldots,j-1\}$, we search for a vertex in $M_{j'}$ with the property described in the claim, and thereby obtain the $s_i$-$t_i$-path $P_i$ in $D$, which we add to the current collection. It now remains to verify that the invariants are maintained.
	The first part of the invariant is readily verified: By the definition of $P_i$, for any $k \in [\ell]$, the total remaining capacity of the vertices in the module $M_k$ which is consumed by the addition of $P_i$ remains unchanged for all $k \in [\ell] \setminus \{k_1,\ldots,k_j\}$, decreases by $1$ if $k \in \{k_2,k_3,\ldots,k_{j-1}\}$, decreases by $1$ if $k_1 \neq k_j$ and $k \in \{k_1,k_j\}$, and finally by $2$, if $k=k_1=k_j$. These values are identical with the number of traversals of $v_k$ by the $v_{k_1}$-$v_{k_j}$-path $Q_i$ in $D_M$, and therefore, the inequality claimed in the invariant stays valid.
	
	For the second part of the invariant, let $z \in V(D)$ be arbitrary. If $z \notin V(P_i)$, both sides of the claimed inequality in the invariant remain unaffected by the addition of $P_i$ to the collection, and therefore the inequality stays valid. If $z=u_{j'}$ with $2 \leq j' \leq j-1$, the selection of the vertices $u_{j'}$ according to the property described in the claim ensures that
	$$R_{(I\cup \{i\},J \setminus \{i\})}(z)=R_{(I,J)}(z)-1 \ge |\CondSet{i' \in J \setminus \{i\}}{s_{i'}=z}|+|\CondSet{i' \in J \setminus \{i\}}{t_{i'}=z}|,$$ and so also in this case, the inequality keeps satisfied. If $z \in \{u_1,u_j\}=\{s_i,t_i\}$ and $s_i \neq t_i$, both sides of the inequality are reduced by exactly $1$ by the addition of $P_i$ to the collection, so the inequality remains valid. If instead $z=u_1=u_j=s_i=t_i$, both sides of the inequality are reduced by $2$ by the addition of $P_i$, again maintaining validity of the invariant. 
	
	Finally, we conclude that we can continue the described process and maintain the invariant properties until we end up with $I=B \cup \{r'+1,\ldots,r\}, J=\emptyset$, and therefore a collection $\mathcal{P}=(P_i~|~i \in [r])$ of $s_i$-$t_i$-paths in $D$ for which the remaining capacity of each vertex $z \in V(D)$, by the second invariant property, is non-negative. This means that $\mathcal{P}$ is compatible with $(w_M^B, \mathcal{S}_M^B)$. 
	
	For any $i \in [r'] \setminus B$, we have added the corresponding path $P_i$ in the first step of the process, namely, $P_i$ is a member from $\mathcal{P}_k^B$, where $s_i, t_i \in M_k$, and therefore $V(P_i) \subseteq M_k$. If $i \in B$, let $k \in [\ell]$ be such that $s_i, t_i \in M_k$. $P_i$ was defined by looking at the $v_k$-$v_k$-path $Q_i$ in $D_M$, and selecting for each vertex traversed by $Q_i$ a vertex on $P_i$ from the corresponding module. Because $Q_i$ visits at least two different vertices, this proves that $P_i$ uses at least one vertex outside the module $M_k$. Finally, this verifies that $\mathcal{P}$ has the additional properties as claimed, and concludes the proof of the equivalence.
	
	Note that it follows directly from the described construction that we have $|P_i|=|Q_i|$ for all $i \in B \cup \{r'+1,\ldots,r\}$. We therefore conclude
	\begin{align*}
	W_B(D,w,\mathcal{S}) \leq W(\mathcal{P})&=\sum_{i \in B \cup \{r'+1,\ldots,r\}}{|Q_i|}+\sum_{k=1}^{\ell}{W(\mathcal{P}_k^B)}=W(\mathcal{P}_M^B)+\sum_{k=1}^{\ell}{W(\mathcal{P}_k^B)}\\
	&=W(D_M,w_M^B,\mathcal{S}_M^B)+\sum_{k=1}^{\ell}{W(D[M_k],w_k^B,\mathcal{S}_k^B)},
	\end{align*}
	where in the last step we have used the optimality of the collections $\mathcal{P}_M^B$ and $\mathcal{P}_k^B, k \in [\ell]$. Putting this together with the inverse inequality obtained at the end of the proof of the $\Longrightarrow$-direction, we obtain the equality claimed at the end of the Lemma.
\end{proof}
It is not hard to convert the constructions described in the proof of the previous Lemma into polynomial-time algorithms.
\begin{corollary} \label{recursivecall}
	With the notation of \cref{disjpathrecursion}, there exists an algorithm that, given as instance
	\begin{itemize}
		\item The digraphs $D,D_M$, the capacities $w:V(D) \rightarrow \mathbb{N}_0$, a threshold $\tau \in \mathbb{N}$ such that $\sum_{z \in V(D)}{w(z)} \leq \tau$,
		\item The subset $B \subseteq [r']$, the collections $\mathcal{S},\mathcal{S}_M^B,\mathcal{S}_k^B, k \in [\ell]$, the module decomposition $M_1,\ldots,M_\ell$,
		\item For every $k \in [\ell]$, a collection $\mathcal{P}_k^B$ of directed paths and cycles in $D[M_k]$ which is optimal for $(w_k^B,\mathcal{S}_k^B)$,
		\item A collection $\mathcal{P}_M^B$ of directed paths and cycles in $D_M$ which is optimal for $(w_M^B,\mathcal{S}_M^B)$,
	\end{itemize}
	outputs a collection $\mathcal{P}=(P_1,\ldots,P_r)$ of directed paths and cycles in $D$ feasible for $(w,\mathcal{S})$, such that $P_i$ stays within the module containing $s_i$ and $t_i$ for all $i \in [r'] \setminus B$, while it has a vertex outside the module if $i \in B$. Additionally, $\mathcal{P}$ minimises $W(\mathcal{P})$ among all feasible collections with these additional properties, that is, $W(\mathcal{P})=W_B(D,w,\mathcal{S})$. 
	
	The algorithm runs in time $\mathcal{O}(r^2\ell|V(D)|)$.
\end{corollary}
\begin{proof}
	Assume we are given as instance for our algorithm collections $\mathcal{P}_k^B, k=1,\ldots,\ell$ and $\mathcal{P}_M^B$ of directed paths and cycles in $D[M_k]$ respectively $D_M$ which are feasible and optimal for $(w_k^B,\mathcal{S}_k^B)$ respectively $(w_M^B,\mathcal{S}_M^B)$. Following the strategy of the $\Longleftarrow$-part of the proof of \cref{disjpathrecursion}, we can construct a collection $\mathcal{P}$ of directed cycles and paths in $D$ which is feasible for $(w_M^B, \mathcal{S}_M^B)$ and moreover (see the end of the proof of the $\Longleftarrow$-part of \cref{disjpathrecursion}) satisfies $W(\mathcal{P})=W_B(D,w,\mathcal{S})$. 
	
	It therefore remains to determine the runtime needed for executing the construction of $\mathcal{P}$ as described in \cref{disjpathrecursion}. The procedure consists of the following steps. 
	\begin{itemize}
		\item Initialising the collection $\mathcal{P}$ as the union of the collections $\mathcal{P}_k^B$, $k=1,\ldots,\ell$. As each collection contains at most $r$ paths and cycles in $D$, this certainly can be executed in time $\mathcal{O}(r\ell|V(D)|)$. 
		\item Initialising $I\coloneqq \emptyset, J\coloneqq B \cup \{r'+1,\ldots,r\}$, and the remaining capacities $R_{(I,J)}(z)\coloneqq w(z)$ for all $z \in V(D)$. This certainly requires time at most $\mathcal{O}(r+|V(D)|)$.
		\item In each of the $|B|+r-r' \leq r$ remaining steps of the procedure, selecting some $i \in J$ and constructing the path $P_i$ (given the path $Q_i \in \mathcal{P}_k^B$). This includes for every module $M_k$ with $v_k \in V(Q_i) \setminus \{\eta(s_i),\eta(t_i)\}$ the selection of a vertex $u \in M_k$ with the property 
		$$R_{(I,J)}(u)>|\CondSet{i' \in J \setminus \{i\}}{s_{i'}=u}|+|\CondSet{i' \in J \setminus \{i\}}{t_{i'}=u}|.$$ Computing the right hand side of this inequality for a given vertex $u$ requires only time $\mathcal{O}(r)$, the left hand side is known. Therefore, searching for a vertex $u$ with this property in the worst case can take time $\mathcal{O}(r|M_k|)$. Consequently, the total runtime required for constructing the path $P_i$ can be bounded as $\mathcal{O}(r|V(D)|)$. 
		\item In each of the at most $r$ steps, when adding a new path $P_i$ to the collection, updating $I\coloneqq I \cup \{i\}, J\coloneqq J \setminus \{i\}$, and updating the corresponding reduced capacities $R_{(I,J)}(z)$ for all $z \in V(D)$. We have to reduce $R_{(I,J)}(z)$ only at vertices $z \in V(P)$ by $1$ (respectively $2$, if $z=s_i=t_i$), which requires not more than $\mathcal{O}(\ell)$ operations (the length of $Q_i$ is at most $\ell$).
		\item Returning the collection $\mathcal{P}$ of directed paths and cycles. This requires at most $\mathcal{O}(r|V(D)|)$ operations.
			\end{itemize} 
Summing up, the algorithm therefore runs in time $\mathcal{O}(r^2\ell|V(D)|),$ which verifies the claim.

\end{proof}
We are now prepared for solving our weighted generalisation of \ref{prob:rVDDP} on arbitrary digraphs as claimed by \cref{algdisjpaths}. In order to ensure the claimed polynomial run-time for fixed parameters $r$ and $\omega$, we have to carefully set up our recursion, and in fact, solve a a slightly more general algorithmic problem, which shall be described in the following.

We need some further terminology, which is adapted to the one from \cref{disjpathrecursion}.

If $\mathcal{S}=[(s_1,t_1),\ldots,(s_r,t_r)]$ is a list of vertex-tuples in a digraph $D$ with vertex-capacities $w:V(D) \rightarrow \mathbb{Z}$, then for any $A \subseteq [r]$, we let $\mathcal{S}(A)\coloneqq [(s_i,t_i)~|~i \in A]$. We furthermore define a corresponding reduced capacity $w_A:V(D) \rightarrow \mathbb{Z}$ on the vertices by
$$w_A(z)\coloneqq w(z)-|\CondSet{i \notin A}{s_i=z}|-|\CondSet{i \notin A}{t_i=z}|.$$

\ProblemDefLabelled{$r$ Vertex-Disjoint Directed Paths (Capacity Version, Sub-Lists)}{$r$-VDDP-CS}{prob:rVDDP-CS}
{A digraph $D$ with non-negative integer capacities $w:V(D) \rightarrow \mathbb{N}_0$, a list $\mathcal{S}=[(s_1,t_1), \ldots, (s_r,t_r)]$ of not necessarily disjoint pairs of not necessarily distinct vertices in $D$ and a threshold $\tau \in \mathbb{N}$ such that $\sum_{z \in V(D)}{w(z)} \leq \tau$.}
{For each set $A \subseteq [r]$, test whether $w_A \ge 0$ and if there exists a feasible collection of directed paths and cycles in $D$ which is compatible with $(w_{A},\mathcal{S}(A))$.
	If so, output the value of $W(D,w_{A},\mathcal{S}(A))$ and a corresponding optimal collection.}

Since $w_{[r]}=w$, solving the above problem contains the weighted generalisation \ref{prob:rVDDP-C} of the \ref{prob:rVDDP} as the special case $A=[r]$, and therefore the following finally implies \cref{algdisjpaths}. 
\begin{theorem}\label{subcollections}
	There exists an algorithm that solves \ref{prob:rVDDP-CS} in time $$\mathcal{O}\left(n^3+f(r,\omega)(n^2+n\log \tau)\right),$$ where $n\coloneqq |V(D)|$, $\omega\coloneqq \dmw(D)$, and $f(r,\omega)=2^{\mathcal{O}(r \log r \cdot 2^\omega \omega)}$.
\end{theorem}
\begin{proof}
	In the following, we describe a recursive algorithm which makes use of \cref{disjpathrecursion} and \cref{recursivecall}. Afterwards, we analyse the run-time of the proposed algorithm. 
	
	\paragraph{Description and Correctness of the Algorithm}
	Assume we are given as input $D,w,\mathcal{S}=[(s_1,t_1),\ldots,(s_r,t_r)]$, and $\tau \in \mathbb{N}$ such that $\sum_{z \in V(D)}{w(z)} \leq \tau$.
	
	We first check whether $|V(D)|=1$. If this is the case, and $V(D)=\{v\}$ for a unique vertex $v$, there exist no directed $s_i$-$t_i$-paths in $D$ at all. We therefore return for every $A \subseteq [r]$ with $A \neq \emptyset$ that no feasible collection for $(D,\mathcal{S}[A],w_A)$ exists. If $A=\emptyset$, we test whether $w_\emptyset(v)=w(v)-2r \ge 0$. If so, the empty collection is the unique feasible collection for $(D,\mathcal{S}(\emptyset),w_\emptyset)$, and this is our output. Otherwise, we return that no feasible collection exists.
	
	If $|V(D)| \ge 2$, we start by applying the algorithm from \cref{computedecomposition} to $D$ and obtain a non-trivial module decomposition $\Set{M_1,\ldots,M_\ell}$, of $V(D)$, where $2 \leq \ell \leq \dmw(D)$. 
	
	We compute $r' \in \{0,\ldots,r\}$ and relabel the elements of $\mathcal{S}$ in such a way that $s_i,t_i$ are contained in the same module for all $i \in [r']$ and in different modules for all $i>r'$.  
	
	Next we determine the digraphs $D[M_k], k=1,\ldots,\ell$ as well as the module-digraph $D_M$. 
	For each $k \in [\ell]$, we then compute the sets $$I_k^\text{out}\coloneqq \CondSet{i \in [r]}{s_i \neq t_i \text{ and }
	|\{s_i,t_i\} \cap M_k|=1}, I_k^\text{in}\coloneqq \CondSet{i \in [r]}{s_i,t_i \in M_k}$$ as well as integer-capacities $w_k: M_k \rightarrow \mathbb{Z}$ defined by
	
	$$w_k(z)\coloneqq w(z)-|\{i \in I_k^\text{out}|s_i=z\}|-|\{i \in I_k^\text{out}|t_i=z\}|,$$
	for all $z \in M_k$. Finally, we compute the sub-collection of vertex pairs
	$$\mathcal{S}_k\coloneqq [(s_i,t_i)|i \in I_k^\text{in}]=\mathcal{S}(I_k^\text{in})$$ for each $k \in [\ell]$.
	
	We now make a recursive call of the algorithm to the each of the $\ell$ instances $$(D[M_k],w_k,\mathcal{S}_k, \tau), k \in [\ell].$$
	Each of these instances is indeed feasible, as we have
	$$\sum_{z \in M_k}{w_k(z)} \leq \sum_{z \in V(D)}{w(z)} \leq \tau$$ for all $k \in [\ell]$. From the recursive calls, for every $k \in [\ell]$ and every set $A \subseteq I_k^\text{in}$, we know whether $(w_k)_A \ge 0$, and if so, we know the value of $W(D[M_k],(w_k)_A,\mathcal{S}_k(A))$ as well as a corresponding optimal collection of directed paths and cycles in $D[M_k]$, which we denote by $\mathcal{P}_k(A)$.
	
	We now go through all (at most $2^{r}$) possible subsets $A \subseteq [r]$. For a fixed $A \subseteq [r]$, we compute the weighting $w_A$ defined as before by $$w_A(z)\coloneqq w(z)-|\CondSet{i \notin A}{s_i=z}|-|\CondSet{i \notin A}{t_i=z}|.$$ If it has a negative entry, we return that no collection feasible for the pair $(w_A,\mathcal{S}(A))$ exists, and move on to the next set $A$.
	
	Otherwise, we go through all (at most $2^{r'} \leq 2^r$) possible subsets $B \subseteq [r'] \cap A$. 
	
	For a fixed set $B$, we will determine the value of $W_B(D,w_A,\mathcal{S}(A))$ as defined in \cref{disjpathrecursion}, and if this value is finite, compute a collection $\mathcal{P}(A,B)=(P_i(A,B)~|~i \in A)$ of directed paths and cycles in $D$ which is compatible with $(w_A,\mathcal{S}(A))$, additionally has the property that the $s_i$-$t_i$-path $P_i(A,B)$ in $D$ leaves the module containing $s_i, t_i$ for all $i \in B$, stays within the module containing $s_i,t_i$ for all $i \in (A \cap [r'])\setminus B$, and is optimal among collections with these properties in the sense that $W(\mathcal{P}(A,B))=W_B(D,w_A,\mathcal{S}(A))$. In the following, we solve this task by solving an Integer Linear Program with bounded number of variables and constraints, as well as using the algorithm provided by \cref{ipsolving}. 
	
	To do so, for each $k \in [\ell]$, we compute the subset $A_k\coloneqq I_k^\text{in} \cap A$, which corresponds to the vertex-pairs $(s_i,t_i)$ in $\mathcal{S}(A)$ within the module $M_k$. 
	
	For a fixed but arbitrary $k \in [\ell]$, let $(w_A)_k^B: M_k \rightarrow \mathbb{Z}$ and the collection $(\mathcal{S}(A))_k^B$ be defined as in \cref{disjpathrecursion} with respect to the pairs $\mathcal{S}(A)$ in $D$, the module-decomposition $M_1,\ldots,M_\ell$, the weighting $w_A: V(D) \rightarrow \mathbb{N}_0$ and the set $B$.
	
	On the other hand consider the vertex-capacities $(w_k)_{A_k \setminus B}:M_k \rightarrow \mathbb{Z}$ defined with respect to the list of vertex-pairs $\mathcal{S}_k$ within $D[M_k]$, the weighting $w_k:M_k \rightarrow \mathbb{Z}$ of $D[M_k]$ and the subset $A_k \setminus B \subseteq I_k^\text{in}$. 
	
	\begin{claim}
		For any $k \in [\ell]$ we have $(w_A)_k^B=(w_k)_{A_k \setminus B}$ and $(\mathcal{S}(A))_k^B=\mathcal{S}_k(A_k \setminus B)$.
	\end{claim}
	\begin{proof}
		For the first claim, let $z \in M_k$ be arbitrary. We show that $w(z)-(w_A)_k^B(z)=w(z)-(w_k)_{A_k \setminus B}(z)$. By the definition of $w_A$ and the definition of $(w_A)_k^B(z)$ in \cref{disjpathrecursion}, we have
		\begin{align*}
		w(z)-(w_A)_k^B(z)&=(w(z)-w_A(z))+(w_A(z)-(w_A)_k^B(z))\\
		&=|\CondSet{i \in X_1 \cup X_2}{s_i=z}|+|\CondSet{i \in X_1 \cup X_2}{t_i=z}|,
		\end{align*}
		where $X_1\coloneqq (I_k^\text{out} \cup I_k^\text{in}) \setminus A, X_2\coloneqq (I_k^\text{out} \cup I_k^\text{in})) \cap (A \cap (B \cup (\{r'+1,\ldots,r\}))$.
		By the definition of $(w_k)_{A_k \setminus B}$, we have
		\begin{align*}
		w(z)-(w_k)_{A_k \setminus B}(z)&=(w(z)-w_k(z))+(w_k(z)-(w_k)_{A_k \setminus B}(z))\\
		&=|\CondSet{i \in Y_1 \cup Y_2}{s_i=z}|+|\CondSet{i \in Y_1 \cup Y_2}{t_i=z}|,
		\end{align*}
		where $Y_1\coloneqq I_k^\text{out}, Y_2\coloneqq I_k^\text{in} \setminus (A_k \setminus B)=(I_k^\text{in} \setminus A) \cup (I_k^\text{in} \cap A \cap B)$. The claim follows now from observing that $X_1 \cup X_2=Y_1 \cup Y_2$.
		
		The second part of the claim follows directly from observing that both $(\mathcal{S}(A))_k^B$ and $\mathcal{S}(A_k \setminus B)$ contain exactly those pairs $(s_i,t_i)$ with $i \in A, i \in I_k^\text{in}$, but $i \notin B$. 
	\end{proof}
	
	This claim implies that using the information obtained from the recursive calls, for every $k \in [\ell]$, we already know 
	\begin{itemize}
		\item the capacities $(w_A)_k^B=(w_k)_{A_k \setminus B}$,
		\item if $(w_A)_k^B$ is non-negative, the value $$W(D[M_k], (w_A)_k^B, (\mathcal{S}(A))_k^B)=W(D[M_k], (w_k)_{A_k \setminus B}, \mathcal{S}_k(A_k \setminus B)),$$ 
		\item and if this value is finite, a collection of directed cycles and paths in $D[M_k]$ which is optimal for the pair $((w_A)_k^B,\mathcal{S}(A_k))$.
	\end{itemize}
	We now proceed by going through the vertices of $D$ in the different modules and testing for non-negativity of the capacities $(w_A)_k^B(z), z \in M_k, k \in [\ell]$. If we find a negative entry, by \cref{disjpathrecursion}, no feasible collection $\mathcal{P}(A,B)$ with the required properties exists. We then close this branch of computations with respect to $B$ and proceed with the next subset of $A \cap [r']$.
	
	Otherwise, we compute for each vertex $v_k \in V(D_M)$ the corresponding integer capacity $(w_A)_M^B(v_k)$ as defined in \cref{disjpathrecursion}, that is, we set
	$$(w_A)_M^B(v_k)\coloneqq \sum_{z \in M_k}{w_A(z)}-W(D[M_k],(w_A)_k^B,(\mathcal{S}(A))_k(B))$$ for all $k \in [\ell]$. 
	
	Now we test for non-negativity. If we find that $(w_A)_M^B(v_k)<0$ for some $k \in [\ell]$,  we again conclude by \cref{disjpathrecursion} that there exists no collection $\mathcal{P}(A,B)$ with the required properties, and we again close this branch of computations with respect to the set $B \subseteq A \cap [r']$, and move on with another subset.
	
	If $(w_A)_M^B(v_k) \ge 0$ for all $k \in [\ell]$, we compute $$(\mathcal{S}(A))_M^B\coloneqq [(\eta(s_i),\eta(t_i))~|~i \in A \cap (B \cup (\{r'+1,\ldots,r\})]$$ in accordance with \cref{disjpathrecursion}. 
	
	Next we set up an ILP which computes the value of $W(D_M,(\mathcal{S}(A))_M^B,(w_A)_M^B)$. In order to model the problem, for every $i \in A \cap (B \cup (\{r'+1,\ldots,r\})$, we produce a list $\Pi_i$ of all $s_i$-$t_i$-paths in $D_M$. For every pair $(i,Q)$ with $Q \in \Pi_i$, we have a binary variable $x_{i,Q} \in \{0,1\}$ which encodes whether the $s_i$-$t_i$-path $Q$ is used for the collection or not.
	
	For each $Q \in \Pi_i$, we compute the size $|Q|$ of $Q$. Furthermore, for each vertex $v_k \in V(D_M), k \in [\ell]$, we compute the number of traversals 
	$$c_{Q,i,k}\coloneqq \begin{cases} 2, & \text{if }v_k=\eta(s_i)=\eta(t_i) \cr 1, & \text{if }v_k \in V(Q) \setminus \{\eta(s_i)=\eta(t_i)\} \cr 1, & \text{if } v_k \in V(Q), \eta(s_i) \neq \eta(t_i) \cr 0, & \text{if }v_k \notin V(Q) \end{cases}$$ 
	of the vertex $v_k$ by $Q$. The ILP is defined as follows.
	\begin{alignat}{3}\label{ILPdisjointpaths}
		~&&~\min \sum_{i}{\sum_{Q \in \Pi_i}{|Q| \cdot x_{i,Q}}}~&&~\\
		\text{  subj.\ to }~&&~\sum_{Q \in \Pi_i}{x_{i,Q}}&=1 ~&&~ \text{ for all } i \in A \cap (B \cup (\{r'+1,\ldots,r\}), \nonumber\\
		~&&~ \sum_{i}\sum_{Q \in \Pi_i}{c_{Q,i,k} \cdot x_{i,Q}} &\leq (w_A)_M^B(v_k) ~&&~ \text{ for all }k \in [\ell], \nonumber\\
		~&&~ 0 \leq x_{i,Q} &\leq 1 ~&&~ x_{i,Q} \in \mathbb{Z}.\nonumber
	\end{alignat}
	It is easy to see that the feasible variable assignments of this ILP are in bijection with the collections of directed paths and cycles in $D_M$ that are feasible for the pair $((w_A)_M^B, (\mathcal{S}(A))_M^B)$. Moreover, if such a collection exists, the optimum of the ILP is $W(D_M,(\mathcal{S}(A))_M^B,(w_A)_M^B)$, and a corresponding optimal variable assignment encodes an optimal collection of paths and cycles in $D_M$. The optimal value of the program (if existent) is clearly bounded from above by $U_1\coloneqq r(\ell+1)$, and the maximum absolute value appearing in any feasible variable assignment is bounded by $U_2\coloneqq 1$.
	
	We therefore, as a next step, run the algorithm from \cref{ipsolving} to determine whether the ILP (\ref{ILPdisjointpaths}) is feasible, and if so, to find an optimal variable assignment. If the algorithm returns that the ILP is not feasible, with \cref{disjpathrecursion} we conclude that no collection $\mathcal{P}(A,B)$ with the required properties exist. We therefore close this branch of computations with respect to $B$ and move on with the next subset of $A \cap [r']$. 
	
	Otherwise, we have found an optimal collection of paths an directed cycles in $D_M$ with respect to the pair $((w_A)_M^B, (\mathcal{S}(A))_M^B)$. By the above, we furthermore know for every $k \in [\ell]$ a collection of directed cycles and paths in $D[M_k]$ which is optimal for the pair $((w_A)_k^B,\mathcal{S}(A_k))$. We have now gathered all necessary information to apply the algorithm from \cref{recursivecall} to these instances. As an output, we therefore obtain a collection $\mathcal{P}(A,B)$ of directed paths and cycles in $D$ which is feasible for the pair $(w_A,\mathcal{S}(A))$ and attains the minimum value $W(\mathcal{P}(A,B))=W_B(D,\mathcal{S}_A,w_A)$ among all feasible collections which share with $\mathcal{P}(A,B)$ the property that for any $i \in A \cap [r']$, the $s_i$-$t_i$-path in the collection leaves the module containing $s_i$ and $t_i$ iff $i \in B$. We now close this branch of computations with respect to $B$ and move on with the next subset of $A \cap [r']$. 
	
	After having finished the described computations for each subset $B \subseteq A \cap [r']$, we first test whether there is at least one set $B \subseteq [r']$ for which a corresponding feasible collection $\mathcal{P}(A,B)$ exists. 
	
	If this is not the case, this means that there exists no feasible collection of directed paths and cycles in $D$ with respect to the pair $(w_A, \mathcal{S}(A))$, no matter which $s_i$-$t_i$-paths with $i \in A \cap [r']$ are required to leave or stay in the module containing $s_i,t_i$. We therefore output that no feasible collection for $(D,\mathcal{S}(A),w_A)$ exists at all, close the branch of computations with respect to the set $A \subseteq [r]$, and move on with the next subset of $[r]$.
	
	Otherwise, it follows from the definitions that we have the equality
	$$W(D,\mathcal{S}(A),w_A)=\min_{B \subseteq [r']}{W_B(D,\mathcal{S}(A),w_A)}.$$ 
	As we have already computed the values $W_B(D,\mathcal{S}(A),w_A)$ for all $B \subseteq A \cap [r']$, we can compute a set $B \subseteq [r']$ for which the minimum is attained. We now output $\mathcal{P}(A,B)$, which is a feasible collection for $(D,\mathcal{S}(A),w_A)$ satisfying $W(\mathcal{P}(A,B))=W(D,\mathcal{S}(A),w_A)$, and is therefore optimal. We close this branch of computations with respect to the set $A$ and move on with the next subset of $[r]$.
	
	After having finished the described computations, for each subset $A \subseteq [r]$ we have either returned an optimal collection of directed paths and cycles for $(D,\mathcal{S}(A),w_A)$ or concluded correctly that no feasible collection exists. 
	
	This closes the description of the algorithm. By the above argumentation, the algorithm, if it runs in finite time, correctly solves \ref{prob:rVDDP-CS} for the given instance. Indeed, the algorithm terminates: As long as $|V(D)| \ge 2$, the digraphs $D[M_1],\ldots,D[M_{\ell}]$ appearing in the recursive calls because of $\ell \ge 2$ each have strictly less vertices than $D$. Therefore, the described recursive algorithm, after finitely many steps, reduces the problem to the same problem on a set of digraphs consisting of single vertices. As described above, the algorithm then terminates in finite time for each of these singleton-digraphs.
	
	\paragraph{Run-Time Analysis}
	For our analysis, we again use a rooted decomposition tree $T$ which corresponds to the structure of the recursive calls during the algorithm. Again, the vertices correspond to the digraphs appearing in the instances of the recursive calls during the algorithm, the root corresponds to the input digraph $D$, and the children of some vertex associated with a digraph $D'$ correspond to the induced subdigraphs of the modules in the module-decomposition of $D'$. The leafs of the tree correspond to the digraphs induced by singletons of $V(D)$. Each digraph $D'$ in the tree by the description of the algorithm is an induced subdigraph of $D$, and therefore, by Fact \ref{inducedmonotonicity}, we have $\dmw(D') \leq \dmw(D)=\omega$. Furthermore note that in any instance $(D',\mathcal{S}',w',\tau)$ used for a recursive call the list $\mathcal{S'}$ is a sublist of $\mathcal{S}$ and therefore has size at most $r$. 
	
	In the following, for an arbitrary branching vertex of $T$ with corresponding instance $(D',\mathcal{S}',w',\tau)$, where $|V(D')| \ge 2$, let us estimate the run-time consumed by all operations made during the corresponding recursive call in the algorithm, disregarding the run-times needed for executing the recursive calls to its leafs. For the sake of simplicity, in the following analysis, we restrict to the root, however the upper bound for the run-time obtained in this way clearly applies to every other branching vertex as well, as it only depends on the parameters $n=|V(D)|$, $r$ and $\omega=\dmw(D)$, which can only get smaller for other branchings. 
	The operations that need to be executed can be listed as follows:
	\begin{itemize}
		\item First we compute a module-decomposition $M_1,\ldots,M_{\ell}$ of $D$ using the algorithm from \cref{computedecomposition}. This requires time $\mathcal{O}(|V(D)|+|E(D)|)=\mathcal{O}(n^2)$.
		\item Determining $r'$ and relabeling the list $\mathcal{S}$ accordingly. This needs time at most $\mathcal{O}(r^2)$.
		\item Computing the module-digraph $D_M$ and $D[M_1],\ldots,D[M_\ell]$. This can be executed in time $\mathcal{O}(|V(D)|^2)=\mathcal{O}(n^2)$.
		\item For each $k \in [\ell]$, computing the sets $I_k^\text{out}, I_k^\text{in}$. This needs time at most $\mathcal{O}(\ell r) \leq \mathcal{O}(\omega r)$.
		\item For each $k \in [\ell]$, computing the weightings $w_k:M_k \rightarrow \mathbb{Z}$ and the sub-collections $\mathcal{S}_k$. This can be done using $\mathcal{O}(\ell n+\ell r)\leq\mathcal{O}(\omega(n+r))$ operations.
		\item For each $A \subseteq [r]$, we need to compute the weightings $w_A$, the sub-collection $\mathcal{S}(A)$, and for each $k \in [\ell]$ the set $A_k$. This needs no more than $\mathcal{O}(2^r(nr+r+\ell r))\leq\mathcal{O}(r2^r(n+\omega))$ operations.
		\item For each of the at most $2^r \cdot 2^r=4^r$ pairs $(A,B)$ of sets with $A \subseteq [r]$, $B \subseteq A \cap [r']$, we need to execute the following:
		\begin{itemize}
			\item For each $k \in [\ell]$, computing the set $A_k \setminus B$ and testing for non-negativity of $(w_k)_{A_k \setminus B}$. This can be done using $\mathcal{O}(\ell (r+n))=\mathcal{O}(\omega(n+r))$ operations.
			\item If applicable, for each $k \in [\ell]$, computing the weights $(w_A)_M^B(v_k)$ and testing for non-negativity. This requires at most $\mathcal{O}(\ell n)=\mathcal{O}(\omega n)$ operations in total.
			\item If applicable, computing $(\mathcal{S}(A))_M^B$. This is in $\mathcal{O}(r)$. 
			\item Generating the list $\Pi_i$ of all $s_i$-$t_i$-paths in $D_M$ for all $i \in A \cap (B \cup \{r'+1,\ldots,r\})$. This is in $\mathcal{O}(r2^\ell)=\mathcal{O}(r2^\omega)$.
			\item Determining $|Q|$ for all $Q \in \Pi_i$ and $i \in [r]$, as well as $c_{Q,i,k}$ for all $Q \in \Pi_i, i \in [r]$ and $k \in [\ell]$. This is in $\mathcal{O}(r\ell 2^\ell)=\mathcal{O}(r\omega2^\omega)$. 
			\item Testing feasibility and solving the ILP \ref{ILPdisjointpaths} using the algorithm from \cref{ipsolving} with respect to the given bounds $U_1=r(\ell+1) \leq r(\omega+1)$ and $U_2=1$. The number of variables of the ILP is bounded by $r2^\ell \leq r2^\omega$, the number of constraints is $\mathcal{O}(r2^\omega)$ as well. By definition, we have
			$$(w_A)_M^B(v_k) \leq \sum_{z \in M_k}{w_A(z)} \leq \sum_{z \in V(D)}{w(z)} \leq \tau$$ for all $k \in [\ell]$. Therefore each entry can be encoded using $\mathcal{O}(\log \tau)$ bits. As the number of entries in the matrix and the vectors is upper bounded in terms of $r$ and $\omega$, the coding length of the matrix and vectors describing the ILP is $L=\mathcal{O}(f_1(r,\omega)\log \tau)$ for some function $f_1$. By \cref{ipsolving}, this step therefore requires running time $\mathcal{O}(f_2(r,\omega)\log \tau \log (U_1U_2))=\mathcal{O}(f_3(r,\omega) \log \tau)$ for fixed functions $f_2$ and $f_3$.
		\end{itemize}
		\item For each $A \subseteq [r]$, if applicable, choosing the best out of the at most $2^r$ collections $\mathcal{P}(A,B), B \subseteq A \cap [r']$ computed beforehand. This is $\mathcal{O}(2^r \log \tau)$.
	\end{itemize}
	Finally, we conclude that the running time consumed by the root (and thus an arbitrary branching vertex) is bounded from above by $\mathcal{O}(n^2+f(r,\omega)(n+\log \tau))$, where $f:\mathbb{N}^2 \rightarrow \mathbb{N}$ is some function. For each of the $n$ leafs of the tree $T$, the respective running time is constant. Again making use of Fact \ref{dectree}, we have $|V(T)| \leq 2n-1$ and therefore conclude that the total running time of the algorithm is bounded from above by
	$$\mathcal{O}\left(|V(T)|(n^2+f(r,\omega)(n+\log \tau))+n\right)=\mathcal{O}\left(n^3+f(r,\omega)(n^2+n\log \tau)\right),$$
	which verifies the bound claimed in the Theorem. From the above one can see that we have $f(r,\omega)=2^{\mathcal{O}(r \log r \cdot 2^\omega \omega)}$. This concludes the proof.
\end{proof}
\section{The Directed Subgraph Homeomorphism Problem and Topological Minors}\label{sec:homeo}
In this section, we want to shed light on the parametrised complexity of the so-called \emph{directed subgraph homeomorphism problem} with respect to the parameter directed modular width. Both the (undirected) subgraph homeomorphism problem \cite{lapaugh} as well as its directed version \cite{fortune} are concerned with a quite general routing problem, in which specified vertices in a digraph are supposed to be connected by internally vertex-disjoint paths according to a given pattern.

Formally, we consider a pair of an \emph{input digraph} $D$ and a \emph{pattern digraph} $\mathcal{H}$. The pattern digraph $\mathcal{H}$ (exceptionally in this paper) is allowed to admit loops, as well as multiple parallel or anti-parallel edges. A \emph{homeomorphism} from $\mathcal{H}$ into $D$ maps vertices of $\mathcal{H}$ to distinct vertices in $D$ and directed edges in $\mathcal{H}$ to directed paths in $D$ connecting the images of the corresponding end vertices, such that the paths only intersect at common endpoints. For a loop, this means that its image forms a directed cycle passing through the image of the incident vertex in $\mathcal{H}$. The directed homeomorphism problem is to decide whether a given digraph $D$ contains a homeomorphic image of a pattern digraph $\mathcal{H}$ using specified vertices. This can be seen as a generalised path finding problem. In fact, the $r$-VDDP is the special case of this problem where the pattern digraph $\mathcal{H}$ forms an oriented matching consisting of $r$ disjoint edges. To keep control of the complexity of this problem, one usually regards the pattern digraph $\mathcal{H}$ as part of the problem description rather than as part of the input.

\ProblemDefLabelled{Directed Subgraph Homeomorphism Problem (pattern digraph $\mathcal{H}$)}{$\mathcal{H}$-DSHP}{prob:DSH}
{A digraph $D$, and a list $s_1,s_2,s_3,\ldots,s_r$ of pairwise distinct vertices in $D$, where $r=|V(\mathcal{H})|$.}
{Decide whether $D$ contains a homeomorphic image of $\mathcal{H}$ with vertex set $V(\mathcal{H})=\{h_1,h_2,\ldots,h_r\}$, such that for every $i \in [r]$, the vertex $h_i$ is mapped to $s_i$. If so, find such a homeomorphic image in $D$.}

A classification of the pattern digraphs $\mathcal{H}$ for which the decision version of the $\mathcal{H}$-DSHP is NP-complete respectively polynomial-time solvable (for fixed $r$) was obtained in \cite{fortune}. Their main result shows that the problem is in $P$ for pattern digraphs $\mathcal{H}$ which admit a dominating source or a dominating sink, and NP-complete in every other case. On the positive side, they established polynomial-time algorithms to solve the $\mathcal{H}$-DSHP on DAGs for every fixed pattern digraph $\mathcal{H}$. 

As the main-result of this section, we show that as a direct Corollary of \cref{algdisjpaths}, for any fixed pattern digraph $\mathcal{H}$, we can obtain an FPT-algorithm solving the $\mathcal{H}$-DSHP with respect to directed modular width as the fixed parameter. 
\begin{theorem}\label{dshpsolution}
Let $\mathcal{H}$ be a pattern digraph with vertex set $\{h_1,\ldots,h_r\}$, and let $m\coloneqq |E(H)|$. There exists an algorithm that, given as input a digraph $D$ equipped with a list $s_1,s_2,\ldots,s_r$ of pairwise distinct vertices in $D$, solves the $\mathcal{H}$-DSHP with this instance in time $$\mathcal{O}(n^3+f(m,\omega)n^2),$$ where $n\coloneqq |V(D)|$, $\omega\coloneqq \dmw(D)$, and $f(m, \omega)=2^{\mathcal{O}(m \log m \cdot 2^\omega \omega)}$.
\end{theorem}
\begin{proof}
We apply Theorem \cref{algdisjpaths} to the instance $D, w, \mathcal{S}, \tau$, where we define 
$$w(z)\coloneqq \begin{cases} \deg_\mathcal{H}(h_i), & \text{if }z=s_i \text{ with }i \in [r], \cr
1, & \text{otherwise,}  \end{cases}$$ for all $z \in V(D)$. Here, $\deg_\mathcal{H}(h_i)$ is the degree of $h_i$ in the underlying graph of $\mathcal{H}$, where incident loops are counted twice. Furthermore, $$\mathcal{S}\coloneqq [(s_i,s_j)~|~e \in E(\mathcal{H}), \text{tail}(e)=h_i, \text{head}(e)=h_j].$$ Here, parallel edges lead to multiple occurrences of the same vertex-tuple in the list $\mathcal{S}$. Finally, we define $\tau\coloneqq 2m-r+n=\sum_{z \in V(D)}{w(z)}$. 

It is now easy to see that the homeomorphic images of $\mathcal{H}$ in $D$ are in bijection with collections $\mathcal{P}$ of directed paths and cycles in $D$ which are feasible with respect to the pair $(w,\mathcal{S})$: Let $\mathcal{P}$ be a feasible collection for $(w,\mathcal{S})$. Then for every directed edge from $h_i$ to $h_j$ in $\mathcal{H}$, we route a corresponding directed $s_i$-$s_j$-path in $D$. For any $i \in [r]$, since the accessible capacity $w(s_i)$ at a vertex $s_i$ is already saturated by the end-points of the paths corresponding to the incident edges in $\mathcal{H}$, $s_i$ is contained in no path whose corresponding edge is not incident with $h_i$. Furthermore, all the paths are internally vertex-disjoint, since $w(z)=1$ for $z \notin \{s_1,\ldots,s_r\}$. This shows that mapping any edge in $\mathcal{H}$ to the path connecting the corresponding vertex-tuple in $\mathcal{S}$ defines a homeomorphism from $\mathcal{H}$ into $D$ where $h_i$ is mapped to $s_i$, for $i=1,2,\ldots,r$. The converse of this statement is verified in the same fashion. Therefore, we can use the algorithm from \cref{algdisjpaths} to decide whether a homeomorphic image of $\mathcal{H}$ in $D$ mapping $h_i$ to $s_i$ for $i=1,2,\ldots,r$ exists, and in case it does, to obtain such an image as the union of the paths contained in a feasible collection for $(w,\mathcal{S})$. The running time of this algorithm is $$\mathcal{O}(n^3+f_1(m,\omega)(n^2+n \log (m-r+n)))=\mathcal{O}(n^3+f(m,\omega)n^2),$$ where $f_1$ is the function from \cref{algdisjpaths}, and $f(m,\omega)\coloneqq f_1(m,\omega) \cdot \frac{\log (m-r+n)}{\log n}=2^{\mathcal{O}(m \log m \cdot 2^\omega \omega)}$. This proves the claim.
\end{proof}
For a digraph $\mathcal{H}$, a \emph{subdivision} is a digraph obtained from $\mathcal{H}$ by replacing directed edges with directed paths of positive length connecting the respective end vertices.

If a digraph $D$ contains a subdivision of another digraph $\mathcal{H}$, this digraph $\mathcal{H}$ is also called a \emph{topological minor} of $D$. This notion has great importance in the minor structure theory of directed graphs, as topological minors form a canonical special case of the so-called \emph{butterfly-minors}. The latter are for instance used in the \emph{directed grid theorem} \cite{kawarabayashi2015directed}, a cornerstone of recent developments in structural digraph theory. 

For undirected graphs, it has been proven that testing for topological minors is fixed parameter-tractable \cite{grohe} when parametrising with the size of the minor. However, for directed graphs, there exist instances $\mathcal{H}$ for which it is NP-complete to test whether a given digraph $D$ contains a subdivision of $\mathcal{H}$ (see Theorem 33 in \cite{BANGJENSEN} for an example). We therefore believe that the following result is a relevant contribution to minor-testing in directed graphs.
\begin{theorem}
Let $\mathcal{H}$ be a multi-digraph (possibly containing loops, and multiple parallel and anti-parallel edges). There exists an algorithm that decides whether $\mathcal{H}$ is a topological minor of a given digraph $D$, and if so, returns a subdivision of $\mathcal{H}$ which is a subdigraph of $D$. This algorithm runs in time $$\mathcal{O}(f(m,\omega)n^{r+3}).$$ Here, we have $m\coloneqq |E(H)|, r\coloneqq |V(H)|, n\coloneqq |V(D)|$ and $\omega\coloneqq \dmw(D)$. We furthermore have $f(m,\omega)=2^{\mathcal{O}(m \log m \cdot 2^\omega \omega)}$. 
\end{theorem}
\begin{proof}
Let $\{h_1,h_2,\ldots,h_r\}$ be the labelled set of vertices of $\mathcal{H}$. To test whether $D$ contains a subdivision of $\mathcal{H}$, we apply the algorithm from \cref{dshpsolution} to all $\mathcal{O}(n^r)$ possible combinations $(s_1,s_2,\ldots,s_r)$ of pairwise distinct vertices in $D$. If we find a hoeomorphic image of $\mathcal{H}$ during this process, we output this image, and conclude that $\mathcal{H}$ is a topological minor of $D$, otherwise, we correctly conclude that $D$ does not contain a subdivision of $\mathcal{H}$.
\end{proof}
\section{Computing Other Directed Width Measures}\label{sec:otherwidth}

In the introduction we discussed that most directed width measures are not very powerful in the algorithmic context.
However, some of them like \emph{directed pathwidth} or \emph{cycle-rank} find very specialised applications, but still there are no FPT-algorithms known to compute their corresponding decompositions.
So in some cases it might still be desirable to compute these parameters, even in a setting where the directed modular width is bounded.

This section is dedicated to explore the power of directed modular width with regards to the computation of directed pathwidth and cycle-rank, an idea that was already pursued for directed pathwidth and directed treewidth on the special case of directed co-graphs \cite{gurski2018directed}.

\subsection{Directed Pathwidth}

The directed pathwidth was introduced by Reed, Thomas and Seymour around 1995 and was later tightly bound to a variant of the cops and robbers game by Bar{\'a}t \cite{barat2006directed}.
The game characterisation yields an XP-time algorithm for the computation of directed pathwidth, whose running time can further be improved by a more specialised approach \cite{tamaki2011polynomial}.

Directed pathwidth finds applications in \emph{boolean networks} \cite{tamaki2010directed}.
There also exists a class of so called \emph{FIFO stack-up} problems that is closely connected to the directed pathwidth of the input graph \cite{gurski2015directed,gurski2016complexity}.

Let $D$ be a directed graph and $A,B\subseteq\Fkt{V}{D}$.
We call $\Brace{A,B}$ a \emph{directed separation} if $A\cup B=\Fkt{V}{D}$ and there is no edge with tail in $B\setminus A$ and head in $A\setminus B$ in $D$.
The \emph{order} of a directed separation $\Brace{A,B}$ is $\Abs{A\cap B}$.

Let $P$ be a directed path from $u$ to $v$ and let $t\in\Fkt{V}{P}$.
We denote the subpath from $u$ to $t$ by $Pt$ and the subpath from $t$ to $v$ by $tP$.

\begin{definition}
	Let $D$ be a digraph, $P$ a directed path and $\beta\colon\Fkt{V}{P}\rightarrow 2^{\Fkt{V}{D}}$.
	For a directed subpath $P'$ of $P$, we use the notation $\Fkt{\beta}{P'}=\bigcup_{t\in V(P')}\Fkt{\beta}{t}$.
	
	The tuple $\Brace{P,\beta}$ is called a \emph{directed path decomposition} for $D$ if
	\begin{enumerate}
		
		\item $\bigcup_{t\in V(P)}\Fkt{\beta}{t}=\Fkt{V}{D}$,
		
		\item  $\Brace{\Fkt{\beta}{Pt},\Fkt{\beta}{t'P}}$ is a directed separation of $D$ for every $\Brace{t,t'}\in\Fkt{E}{P}$, and
		
		\item $B_v\coloneqq \CondSet{t\in\Fkt{V}{P}}{v\in\Fkt{\beta}{t}}$ induces a subpath of $P$.
		
	\end{enumerate}
	
	We call $\Fkt{\beta}{t}$ the \emph{bags} of $\Brace{P,\beta}$.
	The \emph{width} of $\Brace{P,\beta}$ is $\Fkt{\operatorname{width}}{P,\beta}=\max_{t\in V(P)}\Abs{\Fkt{\beta}{t}}-1$.
	The directed pathwidth of $D$, denoted by $\dpw{D}$, is the minimum width of a directed path decomposition for $D$.
	
	The task here is to determine the directed pathwidth of a given digraph $D$ and find a directed path decomposition of minimum width while doing so.
	
	\ProblemDefLabelled{Directed Pathwidth}{DPW}{prob:DPW}
	{A digraph $D$.}
	{Find a directed path decomposition of minimum width for $D$.}
	
	\begin{theorem} \label{FPTDPW}
		There exists an algorithm that given a digraph $D$ as input, outputs a directed path decomposition for $D$ of minimum width.
		The algorithm runs in time $\Fkt{\mathcal{O}}{\omega n^3+\omega^3 2^{\omega^2} n}$, where $n\coloneqq\Abs{\Fkt{V}{D}}$ and $\omega\coloneqq\Fkt{\dmw}{D}$.
	\end{theorem}
	
\end{definition}

For the following proofs we need some further notation.
Let $D$ be a digraph and $\Brace{P,\beta}$ a directed path decomposition for $D$.
For every $\Brace{t,t'}\in\Fkt{E}{P}$ we can classify every vertex $v\in\Fkt{V}{D}$ as follows.
\begin{enumerate}
	
	\item If $v\in\Fkt{\beta}{Pt}\setminus\Fkt{\beta}{t'P}$, then $v$ is called \emph{clean}.
	
	\item If $v\in\Fkt{\beta}{Pt}\cap\Fkt{\beta}{t'P}$, then $v$ is called a \emph{guard} at $\Brace{t,t'}$.
	
	\item If If $v\in\Fkt{\beta}{t'P}\setminus\Fkt{\beta}{Pt}$, then $v$ is called \emph{infected}.
	
	\item We say that $v$ is \emph{protected} from a set $X\subseteq\Fkt{V}{D}$ by the guards at $\Brace{t,t'}$ if every directed path from a vertex of $X$ to $v$ in $D$ contains a guard at $\Brace{t,t'}$.
	
\end{enumerate}

Similar to the previous sections we need to solve a weighted version of the DPW-Problem on the module graph which we then can transform into a directed path decomposition using the recursion.
The directed path decomposition we will end up with has certain properties we cannot avoid.
So before we go with the description of the algorithm we will prove that there always is a directed path decomposition of minimum width with these properties.

In the following let $D$ be a digraph and $\mathcal{M}=\Set{M_1,\dots,M_{\ell}}$ a partition of $\Fkt{V}{D}$ into modules.
We say that a directed path decomposition $\Brace{P,\beta}$ of $D$ \emph{respects} $\mathcal{M}$ if the subgraph $P_{M_i}$ of $P$ induced by $\CondSet{t\in\Fkt{V}{P}}{M_i\cap\Fkt{\beta}{t}\neq\emptyset}$ is a directed path for every $i\in [\ell]$.

Please note that it is well-known that directed pathwidth is monotone under taking subgraphs.
To see this let $\Brace{P,\beta}$ be any directed path decomposition of $D$ and $X\subseteq\Fkt{V}{D}$.
We define the the \emph{$X$-reduction} $\Brace{P,\beta}_X$ of $\Brace{P,\beta}$ to be the directed path decomposition of $\InducedSubgraph{D}{X}$ obtained by deleting every vertex $p$ of $P$ with $\Fkt{\beta}{p}\cap X=\emptyset$ and adding edges to the remaining graph in order to form a new directed path still respecting the vertex order induced by $P$.
At last we delete all vertices not contained in $X$ from the bags of the remaining vertices.

We start with two preliminary lemmas.

\begin{lemma}\label{lemma:nonproctectingmodules}
	Let $\Brace{P,\beta}$ be a directed path decomposition for $D$ and $\Brace{t,t'}\in\Fkt{E}{P}$.
	If there is some $i\in[\ell]$ such that $M_i\setminus\Brace{\Fkt{\beta}{Pt}\cap\Fkt{\beta}{t'P}}\neq\emptyset$, then $\Brace{\Fkt{\beta}{Pt}\cap\Fkt{\beta}{t'P}}\setminus M_i$ protects $\Fkt{\beta}{Pt}\setminus M_i$ from $\Fkt{\beta}{t'P}\setminus M_i$.
\end{lemma}

\begin{proof}
	Suppose $\Brace{\Fkt{\beta}{Pt}\cap\Fkt{\beta}{t'P}}\setminus M_i$ does not protect $\Fkt{\beta}{Pt}\setminus M_i$ from $\Fkt{\beta}{t'P}\setminus M_i$.
	Then there must exist a directed path $P'$ in $D$ from $\Fkt{\beta}{t'P}\setminus \Brace{M_i\cup\Fkt{\beta}{Pt}}$ to $\Fkt{\beta}{Pt}\setminus \Brace{M_i\cup\Fkt{\beta}{t'P}}$ avoiding $\Brace{\Fkt{\beta}{Pt}\cap\Fkt{\beta}{t'P}}\setminus M_i$.
	Thus every such path $P'$ must contain a vertex of $M_i\cap\Fkt{\beta}{Pt}\cap\Fkt{\beta}{t'P}$ since $\Brace{\Fkt{\beta}{Pt},\Fkt{\beta}{t'P}}$ is a directed separation of $D$.
	We may choose the path $P'$ to be a shortest one among these paths and we claim $\Abs{\Fkt{V}{P'}\cap M_i}=1$.
	Suppose $P'$ contains more than one vertex of $M_i$ and let $x$ be the first vertex and $y$ the last vertex of $P'$ in $M_i$.
	Let $x'$ be the direct predecessor of $x$ on $P'$, then $\Fkt{V}{P'x'}\cap M_i=\emptyset$ and, since $M_i$ is a module, the edge $\Brace{x',y}$ exists in $D$.
	Hence $P'x'+\Brace{x',y}+yP'$ is a shorter path than $P'$ still meeting all of our requirements.
	
	Since $M_i\setminus\Brace{\Fkt{\beta}{Pt}\cap\Fkt{\beta}{t'P}}\neq\emptyset$ there is a vertex $z\in M_i\setminus \Brace{\Fkt{\beta}{Pt}\cap\Fkt{\beta}{t'P}}$.
	Let $y$ be the unique vertex of $M_i$ on $P'$, $x$ its predecessor and $x'$ its successor.
	Then, with $M_i$ being a module, the edges $\Brace{x,z}$ and $\Brace{z,x'}$ also exist and thus we can find another directed path $P''$ by replacing $y$ with $z$ that meets all of our requirements and avoids all of $\Fkt{\beta}{Pt}\cap\Fkt{\beta}{t'P}$.
	This is a contradiction to $\Brace{\Fkt{\beta}{Pt},\Fkt{\beta}{t'P}}$ being a directed separation and so our assertion follows.
\end{proof}

\begin{lemma}\label{lemma:protectionfrommodules}
	Let $\Brace{P,\beta}$ be a directed path decomposition for $D$ and $\Brace{t,t'}\in\Fkt{E}{P}$.
	If there is $i\in[\ell]$ such that $M_i\setminus\Brace{\Fkt{\beta}{Pt}\cap\Fkt{\beta}{t'P}}\neq\emptyset$ and $\Brace{M_i\cap\Fkt{\beta}{t'P}}\setminus\Fkt{\beta}{Pt}\neq\emptyset$ then $\Brace{\Fkt{\beta}{Pt}\cap\Fkt{\beta}{t'P}}\setminus M_i$ protects $\Fkt{\beta}{Pt}\setminus M_i$ from $M_i$.
\end{lemma}

\begin{proof}
	Since $M_i$ is a module, if there is an edge $\Brace{x,v}$ in $D$ with $x\in M_i$ and $v\in\Fkt{\beta}{Pt}\setminus M_i$, then there is an edge $\Brace{y,v}$ for every $y\in M_i$.
	More generally, if there is a directed path from a vertex in $M_i$ to a vertex of $\Fkt{\beta}{Pt}\setminus M_i$, then, since $M_i\setminus\Brace{\Fkt{\beta}{Pt}\cap\Fkt{\beta}{t'P}}\neq\emptyset$, there is still such a path in $D-\Brace{M_i\cap\Fkt{\beta}{Pt}\cap\Fkt{\beta}{t'P}}$.
	Moreover, there also exists such a path which starts in a vertex of $\Brace{M_i\cap\Fkt{\beta}{t'P}}\setminus\Fkt{\beta}{Pt}\neq\emptyset$, and uses the same set of vertices from $(\beta(Pt) \cap \beta(t'P)) \setminus M_i$.
	So, with $\Brace{\Fkt{\beta}{Pt},\Fkt{\beta}{t'P}}$ being a directed separation, every directed path from $M_i$ to $\Fkt{\beta}{Pt}\setminus M_i$ must contain a vertex of $\Brace{\Fkt{\beta}{Pt}\cap\Fkt{\beta}{t'P}}\setminus M_i$.
\end{proof}

With these observations on how modules protect and are protected within a directed path decomposition we are able to transform any directed path decomposition for $D$ into one that respects $\mathcal{M}$.

\begin{lemma}\label{lemma:mrespecting}
	There exists a directed path decomposition of minimum width for $D$ that respects $\mathcal{M}$.
\end{lemma}

\begin{proof}
	Let $\Brace{P,\beta}$ be a directed path decomposition of $D$ of minimum width.
	If $\Brace{P,\beta}$ does not respect $\mathcal{M}$ there is at least one $i\in[\ell]$ such that $P_{M_i}$ is no directed path.
	We show that we can construct a directed path decomposition $\Brace{P',\beta'}$ from $\Brace{P,\beta}$ such that $P'_{M_i}$ is a directed path and for all $j\in[\ell]\setminus\Set{i}$ the graph $P'_{M_j}$ is a directed path if and only if $P_{M_j}$ is a directed path. Repeating this procedure eventually will yield the claim.
	
	First notice that there cannot exist a $t\in\Fkt{V}{P}$ such that $M_i\subseteq\Fkt{\beta}{t}$, since then $t\in B_x$ for all $x\in M_i$ and since every $B_x$ induces a directed path, $P_{M_i}$ must be a directed path.
	
	Let $t\in\Fkt{V}{P_{M_i}}$ be the vertex maximising $\Abs{M_i\cap\Fkt{\beta}{t}}$ closest to the starting point of $P$ and let $t_1$ be its predecessor on $P$ while $t_2$ is its successor on $P$.
	We distinguish several cases and the way we construct $\Brace{P',\beta'}$ varies slightly from case to case.
	To avoid repetition of arguments we will first give the construction of $\Brace{P',\beta'}$ by distinguishing cases and then show that $\Brace{P',\beta'}$ has the required properties in a separate step.
	In each of our constructions we will introduce a new directed path $P^i$ with $\Abs{\Fkt{V}{P^i}}=\Abs{\Fkt{V}{P_{M_i}}}$, and we treat the vertices of $P^i$ as copies of the vertices of $P_{M_i}$.
	Let $f\colon\Fkt{V}{P_{M_i}}\rightarrow\Fkt{V}{P^i}$ be the bijective function such that for every $t'\in\Fkt{V}{P_{M_i}}$ the set of vertices of $\Fkt{f}{t'}P^i$ is exactly the image of $\Fkt{V}{t'P}\cap\Fkt{V}{P_{M_i}}$ under $f$.
	We call $\Brace{P^i,f}$ the \emph{interval copy} of $P_{M_i}$.
	
	\textbf{Case 1:} If $t_1$ does not exist $t$ must be the starting point of $P$ and since $\Fkt{\beta}{t}$ does not contain all of $M_i$ there must exist a vertex of $M_i$ in $\Fkt{\beta}{t_2P}\setminus\Fkt{\beta}{t}$.
	Now consider the interval copy $\Brace{P^i,f}$ of $P_{M_i}$.
	Let $P'$ be the directed path obtained from $P$ and $P^i$ by deleting $t$ in $P$ and adding an edge from the endpoint of $P^i$ to $t_2$ which is the new starting point of $P-t$.
	
	\textbf{Case 2:} If $t_2$ does not exist, $t$ must be the endpoint of $P$.
	By our choice of $t$ this means that $\Brace{M_i\cap\Fkt{\beta}{t}}\setminus\Fkt{V}{Pt_1}=\Brace{M_i\cap\Fkt{V}{tP}}\setminus\Fkt{V}{Pt_1}\neq\emptyset$.
	Consider the interval copy $\Brace{P^i,f}$ of $P_{M_i}$.
	Let $P'$ be the directed path obtained from $P$ and $P^i$ by deleting $t$ from $P$ and introducing a directed edge from $t_1$ to the starting point of $P^i$.
	
	\textbf{Case 3:} Both $t_1$ and $t_2$ exist.
	Then the choice of $t$ and the fact that $\Fkt{\beta}{t}$ does not contain all of $M_i$ imply that $\Brace{M_i\cap\Fkt{\beta}{t}}\setminus\Fkt{V}{Pt_1}\neq\emptyset$ and so $\Brace{M_i\cap\Fkt{V}{tP}}\setminus\Fkt{V}{Pt_1}\neq\emptyset$.
	Now consider the interval copy $\Brace{P^i,f}$.
	Let $P'$ be the directed path obtained from $P$ and $P^i$ by deleting $t$ from $P$ and introducing an edge from $t_1$ to the starting point of $P^i$ together with an edge from the endpoint of $P^i$ to $t_2$.
	
	In all three cases we define $\beta'\colon\Fkt{V}{P'}\rightarrow 2^{V(D)}$ as follows:
	\begin{align*}
		\Fkt{\beta}'{t'}\coloneqq\TwoCases{\Fkt{\beta}{t'}\setminus M_i}{t'\in\Fkt{V}{P}}{\Brace{\Fkt{\beta}{t}\setminus M_i}\cup\Brace{M_i\cap\Fkt{f^{-1}}{t'}}}{t'\in\Fkt{V}{P^i}}
	\end{align*}
	Now we have to show that $\Brace{P',\beta'}$ is indeed a directed path decomposition, that it is of minimum width and that for all $j\in[\ell]\setminus\Set{i}$ the graph $P'_{M_j}$ is a directed path if and only if $P_{M_j}$ is a directed path.
	Clearly $P'_{M_i}$ is a directed path in all three cases.
	To see that $\Brace{P',\beta'}$ is indeed a directed path decomposition first note that the reduction of $\beta'$ to the vertices of $P^i$ and $M_i$ together with $P^i$ forms a directed path decomposition of $\InducedSubgraph{D}{M_i}$, since $\Brace{P,\beta}$ is a directed path decomposition of $D$.
	Similarly reducing $\beta'$ to the vertices of $\Fkt{V}{D}\setminus M_i$ yields a directed path decomposition of $\InducedSubgraph{D}{\Fkt{V}{D}\setminus M_i}$.
	
	In the second and third case by \cref{lemma:protectionfrommodules,lemma:protectionfrommodules} we have that $\Brace{\Fkt{\beta}{Pt_1}\cap\Fkt{\beta}{tP}}\setminus M_i$ protects $\Fkt{\beta}{Pt_1}\setminus M_i$ from $\Fkt{\beta}{tP}\cup M_i$.
	By definition we have $\Fkt{\beta'}{P'}=\Fkt{V}{D}$.
	For every $\Brace{p,p'}\in\Fkt{E}{P}$ with $p'\in\Fkt{V}{P^i}$ we have $\Fkt{\beta'}{P'p}\setminus\Brace{M_i\cup\Fkt{\beta}{t}}=\Fkt{\beta}{Pt_1}\setminus\Brace{M_i\cup\Fkt{\beta}{t}}$.
	Moreover, if $p$ is the starting point of $P^i$, then $\Fkt{\beta'}{P't_1}\cap\Fkt{\beta'}{pP'}=({\Fkt{\beta}{Pt_1}\cap\Fkt{\beta}{tP}})\setminus M_i$ and for all $\Brace{p,p'}\in\Fkt{E}{P^i}$ we have $\Brace{\Fkt{\beta'}{P'p}\cap\Fkt{\beta'}{p'P'}}\setminus M_i=\Fkt{\beta}{t}\setminus M_i$.
	So in particular $B'v\coloneqq\CondSet{p\in\Fkt{V}{P'}}{v\in\Fkt{\beta'}{p}}$ induces a subpath of $P'$ for all $v\in\Fkt{V}{D}$.
	Also for every $\Brace{p,p'}\in\Fkt{E}{P'}$ with $p'\in\Fkt{V}{P^i}$ the above observations show that $\Fkt{\beta'}{p}\cap\Fkt{\beta'}{p'}$ protects $\Fkt{\beta'}{P'p}\setminus M_i$ from $\Fkt{\beta'}{p'P'}$.
	Additionally we have observed that $\Brace{P',\beta'}$ can be reduced to a directed path decomposition of $\InducedSubgraph{D}{M_i}$ and thus $\Brace{P',\beta'}$ is a directed path decomposition of $D$.
	
	In the first case we again observe $\Brace{\Fkt{\beta'}{P'p}\setminus\Fkt{\beta'}{p'P'}}\cap M_i=\Fkt{\beta}{t}\setminus M_i$ for all $\Brace{p,p'}\in\Fkt{E}{P^i}$.
	Also note that $\Brace{\Fkt{\beta}{t}\cap\Fkt{\beta}{t_2}}$ protects $M_i$ from $\Fkt{\beta}{t_2P}\setminus M_i$ and so again $\Brace{P',\beta'}$ must be a directed path decomposition.
	
	For all $t'\in\Fkt{V}{P}\setminus\Set{t}$ we have $\Abs{\Fkt{\beta'}{\Fkt{f}{t'}}}\leq\Abs{\Fkt{\beta}{t'}}$ and for all $p\in\Fkt{V}{P^i}$ we have $\Abs{\Fkt{\beta'}{p}}\leq\Abs{\Fkt{\beta}{t}}$ since $\Abs{\Fkt{\beta}{t}\cap M_i}$ is maximal.
	Hence $\Fkt{\operatorname{width}}{P',\beta'}\leq\Fkt{\operatorname{width}}{P,\beta}$.
	
	Now let $j\in[\ell]\setminus\Set{i}$ be such that $P_{M_j}$ is a directed path.
	If $\Fkt{V}{P_{M_j}}$ does not contain $t$, then $P'_{M_j}$ must also be a directed path.
	So we may assume $t$ to be a vertex of $P_{M_j}$ and thus $M_j\cap\Fkt{\beta}{t}\neq\emptyset$.
	We have seen in all three cases that $\Brace{\Fkt{\beta'}{P'p}\cap\Fkt{\beta'}{p'P'}}\cap M_i=\Fkt{\beta}{t}\setminus M_i$ for all $\Brace{p,p'}\in\Fkt{E}{P^i}$ and since $M_i\cap M_j=\emptyset$, $M_j\cap\Fkt{\beta'}{p}=M_j\cap\Fkt{\beta}{t}$ for all $p\in\Fkt{V}{P^i}$.
	Hence $P'_{M_j}$ must be a directed path as well.
	The reverse direction follows with an analogue argument which completes the proof.
\end{proof}

So we can always find a directed path decomposition for $D$ respecting $\mathcal{M}$ and witnessing the directed pathwidth of $D$.
In the next step we aim for a further refinement of such a decomposition.
We call a directed path decomposition of $D$ \emph{$\mathcal{M}$-reduced} if
\begin{enumerate}
	
	\item it respects $\mathcal{M}$,
	
	\item if for some $i\in[\ell]$ there is a $t\in\Fkt{\beta}{t}$ with $M_i\subseteq\Fkt{\beta}{t}$, then $\Fkt{\beta}{t'}\cap M_i\neq\emptyset$ implies $M_i\subseteq\Fkt{\beta}{t'}$ for all $t'\in\Fkt{V}{P}$, and
	
	\item if $M_i\setminus\Fkt{\beta}{t}\neq\emptyset$ for all $t\in\Fkt{V}{P}$, then $\Brace{P,\beta}_{M_i}$ is a path decomposition of minimum width for $\InducedSubgraph{D}{M_i}$.
	
\end{enumerate}

\begin{lemma}\label{lemma:mreduced}
	There exists a directed path decomposition of minimum width for $D$ that is $\mathcal{M}$-reduced.
\end{lemma}

\begin{proof}
	Let $\Brace{P,\beta}$ a a path decomposition of minimum width for $D$ and $i\in[\ell]$.
	By \cref{lemma:mrespecting} we may assume $\Brace{P,\beta}$ to respect $\mathcal{M}$.
	For every $\Brace{t,t'}\in\Fkt{E}{P}$ with $M_i\setminus\Brace{\Fkt{\beta}{t}\cap\Fkt{\beta}{t'}}\neq\emptyset$ we know by \cref{lemma:nonproctectingmodules} that $\Brace{\Fkt{\beta}{t}\cap\Fkt{\beta}{t'}}\setminus M_i$ protects $\Fkt{\beta}{Pt}\setminus M_i$ from $\Fkt{\beta}{t'P}\setminus M_i$.
	Hence, if there is some $t\in\Fkt{V}{P}$ with $M_i\subseteq\Fkt{\beta}{t}$ we can delete the vertices of $M_i$ from every bag not containing $M_i$ as a whole and still maintain the property of being a directed path decomposition.
	Clearly we also do not increase the width with this operation.
	Hence for every $i\in[\ell]$ we may further assume that if there is a $t\in\Fkt{\beta}{t}$ with $M_i\subseteq\Fkt{\beta}{t}$, then $\Fkt{\beta}{t'}\cap M_i\neq\emptyset$ implies $M_i\subseteq\Fkt{\beta}{t'}$ for all $t'\in\Fkt{V}{P}$.
	
	Now let us assume that $M_i\setminus\Fkt{\beta}{t}\neq\emptyset$ for all $t\in\Fkt{V}{P}$.
	In this case, since $\Brace{P,\beta}$ respects $\mathcal{M}$ we know that $P_{M_i}$ is a directed path.
	Moreover, $P_{M_i}$ is exactly the path of $\Brace{P,\beta}_{M_i}$.
	Let $\Brace{P_i,\beta_i}$ be a directed path decomposition of minimum width for $\InducedSubgraph{D}{M_i}$.
	In case $P_i$ is longer than $P_{M_i}$ we replace the endpoint $p$ of $P_{M_i}$ in $P$ by a directed path consisting of $\Abs{\Fkt{V}{P_i}}-\Abs{\Fkt{V}{P_{M_i}}}+1$ copies of $p$.
	If $P_{M_i}$ is longer than $P_i$ we replace the endpoint $p$ of $P_i$ by a directed path consisting of  $\Abs{\Fkt{V}{P_{M_i}}}-\Abs{\Fkt{V}{P_i}}+1$ copies of $p$.
	In both cases we also adjust $\beta$ or $\beta_i$ respectively.
	It is easy to see that the result either way is again a directed path decomposition maintaining all properties required by our assumptions.
	Hence we may assume $P_i$ and $P_{M_i}$ to be of the same length.
	Let $f\colon\Fkt{V}{P_{M_i}}\rightarrow\Fkt{V}{P_i}$ such that if $p\in\Fkt{V}{P_{M_i}}$ is the $j$th vertex of $P_{M_i}$, then $\Fkt{f}{p}$ is the $j$th vertex of $P_i$.
	We now replace $\beta$ by a new function $\beta'$ defined for every $p\in\Fkt{V}{P}$ as follows:
	\begin{align*}
		\Fkt{\beta'}{p}\coloneqq\TwoCases{({\Fkt{\beta}{p}\setminus M_i})\cup\Fkt{\beta_i}{\Fkt{f}{p}}}{p\in\Fkt{V}{P_{M_i}}}{\Fkt{\beta}{p}}{\text{otherwise.}}
	\end{align*}
	Again by \cref{lemma:nonproctectingmodules} we know that $\Brace{P,\beta'}_{V(D)\setminus M_i}$ still is a directed path decomposition of $\InducedSubgraph{D}{\Fkt{V}{D}\setminus M_i}$.
	Moreover, $\Brace{P,\beta'}_{M_i}=\Brace{P_i,\beta_i}$ and thus $\Brace{P,\beta'}$ is in fact a directed path decomposition of $D$.
	Since $\Fkt{\operatorname{width}}{\Brace{P,\beta}_{M_i}}\geq\Fkt{\operatorname{width}}{P_i,\beta_i}$ we also know that $\Brace{P,\beta'}$ is of minimum width.
	By repeating this process for every $i\in[\ell]$ with $M_i\setminus\Fkt{\beta}{t}\neq\emptyset$ for all $t\in\Fkt{V}{P}$ we eventually obtain an $\mathcal{M}$-reduced directed path decomposition of minimum width for $D$.
\end{proof}

In order to make use of the recursive nature of modular width we will need a weighted version of directed pathwidth.
For our notion we will introduce two weight functions on the vertices of the module digraph and the width of the decomposition will be evaluated depending on the use of every vertex.
If, in the directed path decomposition, a vertex appears as a guard for at least one edge we will evaluate it with the first function and if it never is a guard, then it will be evaluated with the second function.
This way we model whether we need to place all vertices in the module represented by the given vertex within a single bag, or we can spread the vertices of  the module in question out via an optimum directed path decomposition of the graph induced by the module.
For an $\mathcal{M}$-reduced directed path decomposition of $D$ these two options are exactly those we can choose for the different modules.

\begin{definition}
	Let $D$ be a directed graph and $n\colon\Fkt{V}{D}\rightarrow\mathbb{N}$ and $w\colon\Fkt{V}{D}\rightarrow\mathbb{N}$ two weight functions for the vertices of $D$.
	Given a directed path decomposition $\Brace{P,\beta}$ of $D$ we define the \emph{$n$-$w$-evaluation} function for every vertex $v\in\Fkt{V}{D}$ as follows:
	\begin{align*}
		\Fkt{\operatorname{eval}_{n,w}^{\Brace{P,\beta}}}{v}\coloneqq\TwoCases{\Fkt{n}{v}}{\text{there is}~\Brace{t,t'}\in\Fkt{E}{D}~\text{such that}~v~\text{is a guard at}~\Brace{t,t'}}{\Fkt{w}{v}}{\text{otherwise.}}
	\end{align*}
	The \emph{$n$-$w$-width} of $\Brace{P,\beta}$ is then defined as
	\begin{align*}
		\Fkt{n\text{-}w\text{-}\operatorname{width}}{P,\beta}\coloneqq\max_{t\in V(P)}\sum_{v\in\beta(t)}\Fkt{\operatorname{eval}_{n,w}^{\Brace{P,\beta}}}{v}
	\end{align*}
	At last we define the \emph{$n$-$w$-directed pathwidth} of $D$ as the minimum $n$-$w$-width over all directed path decompositions of $D$.
\end{definition}

We will first show that we can produce a directed path decomposition of minimum $n$-$w$-width for a digraph with at most $\omega$ vertices in a running time depending only on $\omega$.

\begin{lemma}\label{lemma:boundnumberofbags}
	Given a digraph $D$ on exactly $\omega\in\mathbb{N}$ vertices together with two weight functions $n\colon\Fkt{V}{D}\rightarrow\mathbb{N}$ and $w\colon\Fkt{V}{D}\rightarrow\mathbb{N}$
	a directed path decomposition for $D$ of minimum $n$-$w$-width can be found in time $\Fkt{\mathcal{O}}{\omega^32^{\omega^2}}$.
\end{lemma}

\begin{proof}
	First we will show that there is always a directed path decomposition $\Brace{P,\beta}$ of minimum $n$-$w$-width for $D$ with $\Abs{\Fkt{V}{P}}\leq \omega$.
	To see this consider a directed path decomposition $\Brace{P,\beta}$ of minimum $n$-$w$-width with $\Fkt{\beta}{t}\neq\Fkt{\beta}{t'}$ for all distinct $t,t'\in\Fkt{V}{P}$.
	We will iteratively replace every vertex $p$ of $P$ by a directed path such that every vertex that is introduced by the bag $\Fkt{\beta}{p}$, meaning every vertex of $\Fkt{\beta}{p}\setminus\Fkt{\beta}{Pp-p}$, receives its own introductory bag.
	Let $P=\Brace{p_1,\dots,p_{k}}$.
	We define the path $P_1$ to be a directed path with $\Abs{\Fkt{\beta}{p_1}}$ vertices.
	Choose an arbitrary ordering of the vertices of $\Fkt{\beta}{p_1}$ and assign to the first vertex of $P_1$ a bag containing exactly the first vertex of $\Fkt{\beta}{p_1}$.
	Then assign to the second vertex of $P_1$ a bag containing the first and the second vertex of $\Fkt{\beta}{p_1}$.
	Continue in this fashion until every vertex of $P_1$ received a bag containing the vertices of the bags of all previous vertices plus a single new one.
	Clearly this gives us a directed path decomposition of $\InducedSubgraph{D}{\Fkt{\beta}{p_1}}$.
	Now suppose we have already constructed the paths $P_i$ together with their bags for all $i\leq j$ for some $j\leq k-1$.
	Now choose an arbitrary ordering of $\Fkt{\beta}{p_{j+1}}\setminus\Fkt{\beta}{Pp_j}$ and introduce a new directed path $P_{j+1}$ on $\Abs{\Fkt{\beta}{p_{j+1}}\setminus\Fkt{\beta}{Pp_j}}$ vertices.
	Assign to the first vertex of $P_{j+1}$ a bag containing the vertices of $\Fkt{\beta}{p_j}\cap\Fkt{\beta}{p_{j+1}}$ together with the first vertex of $\Fkt{\beta}{p_{j+1}}\setminus\Fkt{\beta}{Pp_j}$.
	Continue with this in the same fashion as before until every vertex of $P_{j+1}$ has been assigned a bag.
	In particular each of those bags now contains $\Fkt{\beta}{p_j}\cap\Fkt{\beta}{p_{j+1}}$ as a subset.
	At last introduce an edge from the endpoint of $P_j$ to the staring point of $P_{j+1}$.
	The result is again a directed path decomposition of $\InducedSubgraph{D}{\Fkt{\beta}{Pp_{j+1}}}$.
	In the end this procedure produces a directed path decomposition $\Brace{P',\beta'}$ with $\Abs{\Fkt{V}{P'}}\geq\Abs{\Fkt{V}{P}}$ and there exists a bijection between the vertices of $D$ and the vertices of $P'$ defined by the first appearance of a vertex of $D$ in a bag of $P'$.
	Hence $\Abs{\Fkt{V}{P}}\leq\omega$.
	
	Now if we are given a path $P$ together with a mapping $\beta\colon\Fkt{V}{P}\rightarrow 2^{V(D)}$ we can test in $\Fkt{\mathcal{O}}{\omega^3}$ whether for every edge $\Brace{t,t'}\in\Fkt{E}{P}$ the tuple $\Brace{\Fkt{\beta}{Pt},\Fkt{\beta}{t'P}}$ defines a directed separation.
	In time $\Fkt{\mathcal{O}}{\omega^2}$ we can check whether $\CondSet{t\in\Fkt{V}{P}}{v\in\Fkt{\beta}{t}}$ induces a subpath od $P$ for every $v\in\Fkt{V}{D}$.
	At last, if the first two answers were yes we can evaluate the $n$-$w$-width of the directed path decomposition $\Brace{P,\beta}$ of $D$ in $\Fkt{\mathcal{O}}{\omega}$ steps.
	
	All that is left to do is to enumerate all possible pairs $\Brace{P,\beta}$ of directed paths of length at most $\omega$ and mappings $\beta\colon\Fkt{V}{P}\rightarrow 2^{V(D)}$.
	There are $\sum_{i=1}^{\omega}\Brace{2^{\omega}}^i\in\Fkt{\mathcal{O}}{2^{\omega^2}}$ such pairs.
	As we have seen it takes $\Fkt{\mathcal{O}}{\omega^3}$ steps to test such a pair on the property of being a directed path decomposition and evaluating its $n$-$w$-width.
	So in total we can find a minimum $n$-$w$-width directed path decomposition of $D$ in time $\Fkt{\mathcal{O}}{\omega^32^{\omega^2}}$.
\end{proof}

\begin{lemma}\label{lemma:decomposemodulegraph}
	Let $D_M$ be the module digraph of $D$ corresponding to $\mathcal{M}$ with vertex set $\Set{v_1,\dots,v_{\ell}}$.
	Let $n\colon\Fkt{V}{D_M}\rightarrow\mathbb{N}$ be defined by $\Fkt{n}{v_i}=\Abs{M_i}$ and $w\colon\Fkt{V}{D_M}\rightarrow\mathbb{N}$ be defined by $\Fkt{w}{v_i}\coloneqq\dpw{\InducedSubgraph{D}{M_i}}+1$ for all $i\in[\ell]$.
	
	Then $\Fkt{n\text{-}w\text{-}\operatorname{dpw}}{D_M}=\dpw{D}+1$.
	Moreover, given a directed path decomposition $\Brace{P_i,\beta_i}$ of minimum width for $\InducedSubgraph{D}{M_i}$ for every $i\in[\ell]$, we can construct a directed path decomposition of minimum width for $D$ in time $\Fkt{\mathcal{O}}{\Abs{\Fkt{V}{D}}^2}$.
\end{lemma}

\begin{proof}
	We start by proving $\Fkt{n\text{-}w\text{-}\operatorname{dpw}}{D_M}\leq\dpw{D}+1$.
	To see this consider a directed path decomposition $\Brace{P,\beta}$ for $D$ of minimum width.
	By \cref{lemma:mreduced} we may assume $\Brace{P,\beta}$ to be $\mathcal{M}$-reduced.
	
	We will now construct a directed path decomposition $\Brace{P_M,\beta_M}$ for $D_M$ as follows:
	
	Let us call a module $M_i$ \emph{guarding} if there is at least one $t\in\Fkt{V}{P}$ with $M_i\subseteq\Fkt{\beta}{t}$ and let $I\subseteq[\ell]$ be the set of all $i$ such that $M_i$ is not guarding.
	For every $i\in I$ introduce a single vertex $p_i$.
	We now iterate over the vertices of $P$ in the order induced by the direction of $P$.
	While iterating we will sometimes mark vertices with elements of $I$ to indicate that the corresponding module has been considered.
	While iterating we will, piece by piece, construct a new directed path $P_M$ together with its bags $\beta_M$.
	Suppose we already constructed the path $P_Mp$ together with its bags $\beta_M$ and some vertices of $P$ have already been marked with elements of $I$.
	Let $t\in\Fkt{V}{P}$ be the next vertex in order that we consider.
	\begin{enumerate}
		\item First check whether there is an $i\in I$ such that $\Fkt{\beta}{t}\cap M_i\neq\emptyset$ and $t$ is not marked with $i$.
		\begin{enumerate}
			\item If this is the case introduce the vertex $p_i$ together with the edge $\Brace{p,p_i}$ to $P_Mp$ and call the result $P_Mp_i$.
			Then mark all vertices $t'\in\Fkt{V}{P}$ whose bag contains vertices of $M_i$ with $i$.
			At last add $v_i$ to the bag of $p_i$.
			\item If this is not the case create a new vertex $p'$ and introduce it together with the edge $\Brace{p,p'}$ to $P_Mp$ and call the result $P_Mp'$.
		\end{enumerate}
		\item In both cases a new vertex was introduced to our path, let us call this new vertex $q$ independent on the outcome of the above.
		For every $j\in[\ell]\setminus I$ with $\Fkt{\beta}{t}\cap M_j\neq\emptyset$ add the vertex $v_j$ to the bag of $q$.
		\item Now check again if there is a some $i\in I$ such that $\Fkt{\beta}{t}\cap M_i\neq\emptyset$ and $t$ is not marked with $i$.
		If so go back to step 1.a, otherwise continue with the next vertex in order.
	\end{enumerate}
	Let $\Brace{P_M,\beta_M}$ be the result of this procedure.
	Since for each of those $i\in I$ the module $M_i$ is not guarding we can deduce from \cref{lemma:protectionfrommodules,lemma:protectionfrommodules} that $\Fkt{\beta_M}{p_i}\setminus\Set{v_i}$ protects $\Fkt{\beta_M}{Pp_i-p_i}$ from $\Fkt{\beta_M}{p_iP}$.
	
	Since $\Brace{P,\beta}$ is a directed path decomposition of $D$, $\CondSet{p\in\Fkt{V}{P_M}}{v_i\in\Fkt{\beta_M}{p}}$ induces a subpath of $P_M$ for every $i\in[\ell]$.
	
	Now let $\Brace{p,p'}\in\Fkt{E}{P_M}$ be an arbitrary edge and suppose $\Brace{\Fkt{\beta_M}{P_Mp},\Fkt{\beta_M}{p'P_M}}$ is not a directed separation of $D_M$.
	In this case there must exist an edge $\Brace{v_i,v_j}$ in $D_M$ with $v_i\in\Fkt{\beta_M}{P_Mp}\setminus\Fkt{\beta_M}{p'}$ and $v_j\in\Fkt{\beta_M}{p'P_M}\setminus\Fkt{\beta}{p}$.
	But this implies that there are edges from the vertices of $M_i$ to the vertices of $M_j$ and there must exist an edge $\Brace{t,t'}\in\Fkt{E}{P}$ such that $M_i\cap\Brace{\Fkt{\beta}{Pt}\setminus\Fkt{\beta}{t'}}\neq\emptyset$ and $M_j\cap\Brace{\Fkt{\beta}{t'P}\setminus\Fkt{\beta}{t}}\neq\emptyset$.
	This contradicts the fact that $\Brace{P,\beta}$ is a directed path decomposition.
	Hence $\Brace{P_M,\beta_M}$ is a directed path decomposition of $D_M$.
	
	Note that for every $i\in I$ there is a unique vertex $p\in\Fkt{V}{P_M}$ with $v_i\in\Fkt{\beta_M}{p}$, namely $p_i$.
	So the vertices of $I$ will be evaluated with $w$ for the $n$-$w$-width of $\Brace{P_M,\beta_M}$, while all the vertices corresponding to guarding modules themself appear as guards of at least one edge of $P_M$.
	Thus these vertices will be evaluated with $n$.
	Since $\Brace{P_M,\beta_M}$ is $\mathcal{M}$-reduced we conclude $\Fkt{n\text{-}w\text{-}\operatorname{width}}{P_M,\beta_M}=\Fkt{\operatorname{width}}{P,\beta}+1=\dpw{D}+1$.
	
	Now we prove the reverse inequality.
	Let $\Brace{P_M,\beta_M}$ be a directed path decomposition of $D_M$ of minimum $n$-$w$-width.
	Let $I\subseteq[\ell]$ be the set of all $i$ for which $v_i$ is not guarding at any edge of $P_M$.
	For every $i\in I$ let $\Brace{P_i,\beta_i}$ be a directed path decomposition for $\InducedSubgraph{D}{M_i}$ of minimum width.
	Note that for every $i\in I$, since $v_i$ is no guarding at any edge of $P_M$, there is a unique vertex $p_i\in\Fkt{V}{P_M}$ whose bag contains $v_i$.
	We create a directed path decomposition of $D$ in several steps.
	\begin{enumerate}
		\item First label every vertex of $P_M$ whose bag contains $v_i$ for $i\in I$ with $i$.
		It can happen that a single vertex is marked several times.
		\item Let $p\in\Fkt{V}{P_M}$ be a vertex that was marked in the previous step and let $Y\subseteq I$ be the set of its labels.
		Replace $p$ by a directed path $L_p$ on $\max_{i\in Y}\Abs{\Fkt{V}{P_i}}$ vertices and copy the content of $\Fkt{\beta_M}{p}$ for every vertex of $L_p$.
		\item Let $P$ be path obtained by the previous step after all marked vertices of $P_M$ have been replaced by paths and let $\beta'$ be the newly obtained bag function.
		For very $p\in\Fkt{V}{P}$ let $\Fkt{\lambda}{p}\coloneqq\CondSet{i\in[\ell]}{v_i\in\Fkt{\beta'}{p}}$.
		\begin{enumerate}
			
			\item For every $p\in\Fkt{V}{P}\cap\Fkt{V}{P_M}$ let $\Fkt{\beta}{p}\coloneqq\bigcup_{i\in\lambda(p)}M_i$.
			
			\item Let $p\in\Fkt{V}{P_M}\setminus\Fkt{V}{P}$, then $P$ contains $L_p$ as a subpath and the $\lambda$ of all vertices of $L_p$ are the same, so let $t\in\Fkt{V}{L_p}$ be chosen arbitrarily.
			Let $\Fkt{V}{L_p}=\Set{t_1,\dots,t_k}$ be numbered with respect to the ordering of the vertices induced by the orientation of $L_p$.
			Moreover, for every $i\in\Fkt{\lambda}{t}\cap I$ let $\Fkt{V}{P_i}=\Set{t_1^i,\dots,t^i_{k_i}}$ be numbered with respect to the orientation of $P_i$.
			By choice of $L_p$ we know $k\geq k_i$.
			Now let $j\in[k]$, we define the bag of $t_j$ as follows:
			\begin{align*}
				\Fkt{\beta}{t_j}\coloneqq\bigcup_{\substack{i\in\lambda(t)\cap I\\t^i_j\in V(P_i)}}\Fkt{\beta_i}{t^i_j}\cup\bigcup_{i\in\lambda(t)\setminus I}M_i.
			\end{align*}
			
		\end{enumerate}
	\end{enumerate}
	Now for every $i\in I$ we have $\Brace{P,\beta}_{M_i}=\Brace{P_i,\beta_i}$ and for every $i\in[\ell]\setminus I$ all of $M_i$ is completely contained in every bag that has a non-empty intersection with the module.
	Hence $\Brace{P_M,\beta_M}$ being a directed path decomposition of $D_M$ of minimum $n$-$w$-width implies that $\Brace{P,\beta}$ is a directed path decomposition of $D$ of width at most $\Fkt{n\text{-}w\text{-}\operatorname{dpw}}{D_M}-1$.
	So in total we have $\Fkt{n\text{-}w\text{-}\operatorname{dpw}}{D_M}=\dpw{D}+1$.
	
	Note that the procedure that produced $\Brace{P,\beta}$ given $\Brace{P_M,\beta_M}$ and $\Brace{P_i,\beta_i}$ for every $i\in I$ uses $\Fkt{\mathcal{O}}{\Abs{\Fkt{V}{D}}^2}$ steps.
	This completes our proof.
\end{proof}

\begin{proof}[Proof of \Cref{FPTDPW}.]
	Assume we are given as an instance a directed graph $D$ of directed modular width at most $\omega$.
	
	If $D$ consists of a single vertex $v$, we output the triple $\Brace{\Brace{\Brace{\Set{p},\emptyset},\beta_v},1,0}$ where $\Fkt{\beta_v}{p}=\Set{v}$.
	Here $\Brace{\Brace{\Set{p},\emptyset},\beta_v}$ is a directed path decomposition of $D$ of minimum width, $1$ is the number of vertices in $D$ and $0$ is the directed pathwidth of $D$.
	Otherwise, we apply the algorithm from \cref{computedecomposition} in order to obtain a non-trivial decomposition of $\Fkt{V}{D}$ into modules $\mathcal{M}=\Set{M_1,\ldots,M_{\ell}}$, where $2 \leq \ell \leq \omega$.
	We compute the induced subdigraphs $\InducedSubgraph{D}{M_1},\ldots,\InducedSubgraph{D}{M_{\ell}}$ as well as the module-digraph $D_M$.
	
	Now we recursively apply the algorithm to each of the instances $\InducedSubgraph{D}{M_i}$, $i\in[\ell]$.
	For each $i \in [\ell]$, we thereby obtain the directed pathwidth of $\InducedSubgraph{D}{M_i}$ together with a directed path decomposition $\Brace{P_i,\beta_i}$ of minimum width and the size $\Abs{M_i}$ of the module.
	
	We now define two weight functions $n\colon\Fkt{V}{D}\rightarrow\mathbb{N}$ and $w\colon\Fkt{V}{D}\rightarrow\mathbb{N}$ for the vertices $v_i$, $i\in[\ell]$, of the module graph $D_M$ as by setting $\Fkt{n}{v_i}\coloneqq\Abs{M_i}$ and $\Fkt{w}{v_i}\coloneqq\dpw{\InducedSubgraph{D}{M_i}}+1$.
	
	By \cref{lemma:boundnumberofbags} we can now compute a directed path decomposition $\Brace{P_M,\beta_M}$ of minimum $n$-$w$-width for $D_M$ in time $\Fkt{\mathcal{O}}{\omega^3 2^{\omega^2}}$.
	Once we produced $\Brace{P_M,\beta_M}$ we can use the procedure described in the proof of \cref{lemma:decomposemodulegraph} to combine $\Brace{P_M,\beta_M}$ with the directed path decompositions of the graphs induced by the modules to obtain a directed path decomposition $\Brace{P,\beta}$ of minimum width for $D$.
	This step takes $\Fkt{\mathcal{O}}{\Abs{\Fkt{V}{D}}^2}$ elementary operations for every module.
	
	In total, the computation time spent on $D$ alone can be bounded by 
	\[
	\Fkt{\mathcal{O}}{n^2+\omega n^2+\omega^3 2^{\omega^2}+\omega n^2}=\Fkt{\mathcal{O}}{\Brace{2\omega+1}n^2+\omega^3 2^{\omega^2}}.
	\]
	We now recurse on every $\InducedSubgraph{D}{M_i}$ with at least $2$ vertices and each such subgraph has strictly less than $n$ vertices by Fact \ref{inducedmonotonicity}.
	The same bound as for $D$ applies for all branching vertices of the tree defined by the recursive calls of our algorithm except for the leaves, where we terminate in constant time.
	Therefore the total run-time required by our algorithm is bounded from above by
	\begin{align*}
		\mathcal{O}({\underbrace{n}_{\text{leafs}}+\underbrace{(2n-1)\Brace{\Brace{2\omega+1}n^2+\omega^3 2^{\omega^2}}}_{\text{branching vertices}}})=\Fkt{\mathcal{O}}{\omega n^3+\omega^3 2^{\omega^2} n}.
	\end{align*}

\end{proof}

\subsection{Cycle Rank}

The cycle-rank of a digraph can be seen as possible generalisation of the notion of \emph{tree-depth} for undirected graphs, however its introduction predates most (un)directed graph measures including tree-depth.
It was introduced in 1963 by Eggan as a tool to analyse regular languages \cite{eggan1963transition} and has found its place among many other digraph width parameters by a characterisation via a cop\&robbers game \cite{giannopoulou2012lifo} which implies an XP-time algorithm for its computation.

\begin{definition}
The \emph{cycle-rank} of a digraph $D$, denoted by $\Fkt{\operatorname{cr}}{D}$, is defined as follows
\begin{enumerate}
	\item If $D$ has no directed cycle, den $\Fkt{\operatorname{cr}}{D}=0$.
	\item If $D$ is strongly connected, then $\Fkt{\operatorname{cr}}{D}=1+\min_{v\in V(D)}\Fkt{\operatorname{cr}}{D-v}$.
	\item Otherwise $$\Fkt{\operatorname{cr}}{D}=\max_{\substack{\text{strongly connected component}~H\\\text{of}~D}}\Fkt{\operatorname{cr}}{H}.$$
\end{enumerate}
\end{definition}

Alternatively we can describe the cycle-rank of $D$ via an elimination tree.
Let $D$ be a digraph,  then $T_D$ is an \emph{elimination tree} for $D$ with root $r\in\Fkt{V}{T_D}$ if
\begin{enumerate}
	
	\item $\Fkt{V}{D}=\Fkt{V}{T_D}$,
	
	\item $\InducedSubgraph{D}{T'}$ is strongly connected for every component of $T'\subseteq T_D-r$, and
	
	\item every component $T'\subseteq T_d-r$ is an elimination tree for $\InducedSubgraph{D}{\Fkt{V}{T'}}$.
	
\end{enumerate}
The \emph{depth} of a rooted tree $T$ is the maximum number of vertices of a path starting in the root of $T$ and ending in a leaf of $T$.
It is straight forward to see that the cycle-rank of $D$ equals the minimum depth of an elimination tree for $D$.

Another way of encoding the cycle-rank of $D$ is via an ordering $\sigma$ of $\Fkt{V}{D}$.
Let $X\subseteq\Fkt{V}{D}$ be a set of vertices and $\sigma$ an ordering of $\Fkt{V}{D}$, then we denote by $\sigma_X$ the ordering $\sigma$ induces on $X$.
Given an ordering $\sigma$ of $\Fkt{V}{D}$ we can derive an elimination tree $T_{\sigma}$ from $\sigma$ as follows
\begin{enumerate}
	
	\item the smallest vertex $v$ of $D$ with respect to $\sigma$ is the root of $T_{\sigma}$, and
	
	\item for every strongly connected component $H\subseteq D-v$ there is a component $T_H$ of $T_{\sigma}-v$ such that $T_H=T_{\sigma_{V(H)}}$.
	
\end{enumerate}
Then the \emph{rank} of $\sigma$, denoted by $\Fkt{\operatorname{rank}}{\sigma}$ is the depth of $T_{\sigma}$ and again it is straight forward to see that the cycle-rank of $D$ is equal to the minimum rank of an ordering $\sigma$ of $\Fkt{V}{D}$.
Hence in order to determine the cycle-rank of $D$ it suffices to find an ordering of minimum rank.

The following is straight forward as we just have to recursively find the strongly connected components of $D-v$ where $v$ is the smallest vertex of $D$ with respect to $\sigma$.
Therefore we omit the proof.

\begin{lemma}\label{lemma:determinerank}
Let $D$ be a digraph and $\sigma$ an ordering of $\Fkt{V}{D}$, then we can compute $T_{\sigma}$ and thus determine the rank of $\sigma$ in time $\Fkt{\mathcal{O}}{\Abs{\Fkt{V}{D}}^3}$.
\end{lemma}

The task at hand is, given a directed graph $D$, to determine its cycle-rank and find an ordering of $\Fkt{V}{D}$ whose rank witnesses the cycle-rank.

\ProblemDefLabelled{Cycle-Rank}{CR}{prob:CR}
{A digraph $D$.}
{Find an ordering of minimum rank for $\Fkt{V}{D}$.}

\begin{theorem}\label{thm:CR}
There exists an algorithm that, given a digraph $D$ as input, outputs an ordering $\sigma$ for $\Fkt{V}{D}$ such that $\Fkt{\operatorname{rank}}{\sigma}=\Fkt{\operatorname{cr}}{D}$.
The algorithm runs in time $\Fkt{\mathcal{O}}{n^3+\omega^3\omega!n}$, where $n\coloneqq\Abs{\Fkt{V}{D}}$ and $\omega\coloneqq\Fkt{\dmw}{D}$.
\end{theorem}

Similar to our approach for directed pathwidth we need to make sure that we can always find an ordering of minimum rank that is consistent with a module decomposition.
So in the following $D$ will always be a strongly connected digraph and $\mathcal{M}=\Set{M_1,\dots,M_{\ell}}$ is a decomposition of $D$ into modules.
We say that an ordering $\sigma$ of $\Fkt{V}{D}$ \emph{respects} $\mathcal{M}$ if for all $i\in[ell]$, $u,w\in M_i$ and $v\in\Fkt{V}{D}$, $u\leq_{\sigma} v\leq_\sigma w$ implies $v\in M_i$.
In other words, $\sigma$ respects $\mathcal{M}$ if every module appears as an interval on $\sigma$.

\begin{lemma}\label{lemma:respectufulcrordering}
Let $D$ be a digraph and $\mathcal{M}=\Set{M_1,\dots,M_{\ell}}$ a decomposition of $D$ into $\ell\in\mathbb{N}^+$ modules, then there exists an ordering $\sigma$ for $\Fkt{V}{D}$ of minimum rank that respects $\mathcal{M}$.
\end{lemma}

\begin{proof}
	We prove the statement via induction over the number of modules $\ell$ in $\mathcal{M}$.
	Note that we allow $\mathcal{M}$ to be trivial in the sense that $\mathcal{M}$ consists of a single module.
	For $\ell=1$ this is trivial since with just one module any ordering respects $\mathcal{M}$.
	
	Now let $\ell\geq 2$ and $\sigma$ be an ordering for $\Fkt{V}{D}$ of minimum rank.
	Let $r$ be the root of the elimination tree $T_{\sigma}$ corresponding to $\sigma$ and let $v\in\Fkt{V}{T_{\sigma}}$ be the vertex closest to $r$ with more than one successor in $T_{\sigma}$.
	Note that $v=r$ in the case where $r$ itself has more than one neighbour.
	Let $P$ be the path from $r$ to $v$ in $T_{\sigma}$.
	Moreover let $T_1,\dots,T_k$ be the components of $T_{\sigma}-\Fkt{V}{P}$ where $r_i$ is the root of $T_i$ for each $i\in[k]$, furthermore please note that every $r_i$ is a successor of $v$ by choice of $v$.
	At last let $j\in[\ell]$ be such that $v\in M_j$.
	
	Suppose there exist two distinct integers $n,m\in[k]$ and some $h\in[\ell]\setminus\Set{j}$ such that $M_j\cap\Fkt{V}{T_n}\neq\emptyset$ while $M_j\cap\Fkt{V}{T_m}=\emptyset$ and $\Fkt{V}{T_m}\cap M_h\neq\emptyset$.
	Since $v$ is the first vertex for which $T_{\sigma}$ branches, we know that $\Set{v}\cup\bigcup_{p=1}^k\Fkt{V}{T_p}$ induces a strongly connected subgraph of $D$.
	Hence there exists a directed cycle containing $v$ and at least one vertex of $\Fkt{V}{T_m}$ in $D'\coloneqq\InducedSubgraph{D}{\Set{v}\cup\bigcup_{p=1}^k\Fkt{V}{T_p}}$.
	Since $M_j$ is a module and $T_m$ does not contain a vertex of $M_j$ we may assume $C$ to contain no other vertex of $M_j$.
	Let $\Brace{t_1,v},\Brace{v,t_2}\in\Fkt{E}{C}$ be the two edges incident with $v$ in $C$.
	By our assumption there exists a vertex $u\in M_j\cap\Fkt{V}{T_n}$ and thus the edges $\Brace{t_1,u}$ and $\Brace{u,t_2}$ must exist in $D$ as well.
	So by replacing $v,\Brace{t_1,v}$, and $\Brace{v,t_2}$ in $C$ with $u,\Brace{t_1,u}$, and $\Brace{u,t_2}$ we obtain a new directed cycle $C'$ in $D'-v$.
	This means that $\Fkt{V}{T_n}\cup \Fkt{V}{T_m}$ is contained in the vertex set of a strongly connected component of $D'-v$.
	This however is a contradiction to the definition of $T_{\sigma}$.
	
	Therefore one of two cases must be true: either for all $h\in[k]$ we have that $\Fkt{V}{T_h}\subseteq M_j$, or for all $h\in[k]$ we have $\Fkt{V}{T_h}\cap M_j=\emptyset$ since $k\geq 2$ by choice of $v$.
	
	In the first case $P$ must contain the vertices of all other modules.
	In this case we may simply consider $D-M_j$ together with $\sigma_{V(D)\setminus M_j}$ and the module decomposition $\Set{M_1,\dots,M_{j-1},M_{j+1},\dots,M_{\ell}}$ into $\ell-1$ modules.
	By our induction hypothesis there exists an ordering $\sigma''$ for $D-M_j$ of minimum rank that respects $\mathcal{M}$.
	Let $\sigma'''$ be an ordering of minimum rank for $\InducedSubgraph{D}{M_j}$ and let $\sigma'$ be the ordering of $\Fkt{V}{D}$ obtained by concatenating $\sigma''$ by $\sigma'''$.
	It follows that $\Fkt{\operatorname{rank}}{\sigma'}\leq\Abs{\Fkt{V}{D}\setminus M_j}+\Fkt{\operatorname{rank}}{\sigma'''}$.
	We claim that $\sigma'$ has rank $\Fkt{\operatorname{cr}}{D}$.
	Since $\Fkt{\operatorname{rank}}{\sigma}\geq\Abs{\Fkt{V}{P}\setminus M_j}+\Fkt{\operatorname{rank}}{\sigma_{M_j}}\geq\Abs{\Fkt{V}{P}\setminus M_j}+\Fkt{\operatorname{rank}}{\sigma'''}$ this implies our claim.
	
	In the second case $P$ must contain all vertices of $M_j$.
	Like in the first case we may assume by induction that there exists an ordering $\sigma''$ of minimum rank for $D-M_j$ that respects $\Set{M_1,\dots,M_{j-1},M_{j+1},\dots,M_{\ell}}$.
	Let $\sigma'''$ be an ordering of minimum rank for $\InducedSubgraph{D}{M_j}$.
	This time we construct our new ordering, which respects $\mathcal{M}$ by concatenating $\sigma'''$ by $\sigma''$, so this time the vertices of $M_j$ are smaller than all vertices of $\Fkt{V}{D}\setminus M_j$.
	We obtain $\Fkt{\operatorname{rank}}{\sigma'}\leq \Abs{M_j}+\Fkt{\operatorname{rank}}{\sigma''}$.
	And again we can obtain $\Fkt{\operatorname{rank}}{\sigma}\geq\Abs{M_j}+\Fkt{\operatorname{rank}}{\sigma_{V(D)\setminus M_j}}\geq\Abs{M_j}+\Fkt{\operatorname{rank}}{\sigma''}$ which concludes our proof.
\end{proof}

As in the case of directed pathwidth we want to construct an ordering of minimum rank for $D$ by finding an ordering for the module graph $D_M$ and combining this ordering with minimum rank orderings of the modules.
To do this we again introduce a weighted version of the cycle-rank problem which evaluates the rank of an ordering depending on the position of a vertex within the corresponding elimination tree.
The idea behind this is, that any vertex which is no a leaf in the elimination tree corresponds to a module which has to be completely deleted before the module of any of its successors can be decomposed.
Hence every inner vertex of the elimination tree has to be counted with the entire cardinality of its corresponding module, while the modules corresponding to leaf vertices may be eliminated in an optimal fashion.

\begin{definition}
Let $D$ be a directed graph and $n\colon\Fkt{V}{D}\rightarrow\mathbb{N}$ and $w\colon\Fkt{V}{D}\rightarrow\mathbb{N}$ two weight functions for the vertices of $D$.
Given an ordering $\sigma$ for $\Fkt{V}{D}$ we define the \emph{$n$-$w$-$\sigma$-evaluation} function for every vertex $v\in\Fkt{V}{D}$ as follows:
\begin{align*}
\Fkt{\operatorname{eval}^{\sigma}_{n,w}}{v}\coloneqq\TwoCases{\Fkt{n}{v}}{v~\text{is an inner vertex of}~T_{\sigma}}{\Fkt{w}{v}}{v~\text{is a leaf of}~T_{\sigma}.}
\end{align*}
The \emph{$n$-$w$-rank} is then defined as
\begin{align*}
\Fkt{n\text{-}w\text{-}\operatorname{rank}}{\sigma}\coloneqq\max_{\substack{P\subseteq T_{\sigma}\\P~\text{path from the root of}~T_{\sigma}~\text{to a leaf}}}~\sum_{v\in V(P)}\Fkt{\operatorname{eval}^{\sigma}_{n,w}}{v}
\end{align*}
At last we define the \emph{$n$-$w$-cycle-rank} of $D$, denoted by $\Fkt{n\text{-}w\text{-}\operatorname{cr}}{D}$, as the minimum $n$-$w$-rank of an ordering $\sigma$ for $\Fkt{V}{D}$.
\end{definition}

In the next step we need to show that given a digraph on a constant number of vertices together with the functions $n$ and $w$, we can find an ordering of minimum $n$-$w$-rank in constant time.
Moreover, we have to relate the $n$-$w$-cycle-rank, with special choices for $n$ and $w$, of the module digraph $D_M$ to the cycle-rank of $D$.

\begin{lemma}\label{lemma:nwrankconstantvertices}
Given $D$ a digraph on exactly $\omega\in\mathbb{N}$ vertices together with two weight functions $n\colon\Fkt{V}{D}\rightarrow\mathbb{N}$ and $w\colon\Fkt{V}{D}\rightarrow\mathbb{N}$, an ordering for $\Fkt{V}{D}$ of minimum $n$-$w$-rank can be found in time $\Fkt{\mathcal{O}}{\omega^3\omega!}$.
\end{lemma}

\begin{proof}
To find an ordering of minimum $n$-$w$-rank for $\Fkt{V}{D}$ we can simply check all $\omega!$ orderings of $\Fkt{V}{D}$.
For each of them we can construct the corresponding elimination tree in $\Fkt{\mathcal{O}}{\omega^3}$ time by \cref{lemma:determinerank}.
It at most the same time we can enumerate for each such elimination tree all paths from the root to a leaf and thus we can determine the $n$-$w$-rank for each ordering in time $\Fkt{\mathcal{O}}{\omega^3}$.
In total, a minimum $n$-$w$-rank ordering for $\Fkt{V}{D}$ can be found in time $\Fkt{\mathcal{O}}{\omega^3\omega!}$.
\end{proof}

\begin{lemma}\label{lemma:nwranktocyclerank}
Let $D_M$ be the module digraph of $D$ corresponding to $\mathcal{M}$ with vertex set $\Set{v_1,\dots,v_{\ell}}$.
Let $n\colon\Fkt{V}{D_M}\rightarrow\mathbb{N}$ be defined by $\Fkt{n}{v_i}\coloneqq\Abs{M_i}$ and $w\colon\Fkt{V}{D_M}\rightarrow\mathbb{N}$ be defined by $\Fkt{w}{v_i}\coloneqq\Fkt{\operatorname{cr}}{\InducedSubgraph{D}{M_i}}$ for all $i\in[\ell]$.

Then $\Fkt{n\text{-}w\text{-}\operatorname{cr}}{D_M}=\Fkt{\operatorname{cr}}{D}$.
Moreover, given an ordering $\sigma_i$ of minimum rank for $\InducedSubgraph{D}{M_i}$ for every $i\in[\ell]$ we can construct an ordering of minimum rank for $\Fkt{V}{D}$ in time $\Fkt{\mathcal{O}}{\ell\Fkt{\log}{\ell}}$.
\end{lemma}

\begin{proof}
We start by showing $\Fkt{n\text{-}w\text{-}\operatorname{cr}}{D_M}\leq\Fkt{\operatorname{cr}}{D}$.
To do this let $\sigma$ be an ordering of minimum rank for $\Fkt{V}{D}$.
By \cref{lemma:respectufulcrordering} we may assume that $\sigma$ respects $\mathcal{M}$.
We construct an ordering $\sigma_M$ for $D_M$ from $\sigma$ by letting $v_i\leq_{\sigma_M}v_j$ if and only if for all $x\in M_i$ and $y\in M_j$ we have $x\leq_{\sigma}y$ for all distinct integers $i,j\in[\ell]$.
Since $\sigma$ respects $\mathcal{M}$, $\sigma_M$ is well defined.
Now let $i\in[\ell]$ be chosen such that $v_i$ is not a leaf of $T_{\sigma_M}$.
Let $r_\sigma$ be the root of $T_{\sigma}$ and thus the smallest vertex with respect to $\sigma$.
We need to prove that $M_i$ induces a path in $T_{\sigma}$ and only one vertex of $M_i$ can have more than one successor in $T_{\sigma}$.

So first suppose $M_i$ does not induce a path in $T_{\sigma}$, then there exist vertices $x_0,x_1,x_2\in M_i$ such that $x_0$ has at least two successors and $x_i$ is the root of a subtree below $x_0$ that does not contain another vertex of $M_i$ for $i\in[2]$.
Moreover, $x_1$ and $x_2$ belong to different components of $T_{\sigma}-x_0$
Since $v_i$ is not a leaf of $T_{\sigma_M}$, $x_1$ can be chosen such that the subtree $T_1$ of $T_{\sigma}$ rooted in $x_1$ contains vertices of another module $M_j$ where $j\in[\ell]\setminus\Set{i}$ and $v_i <_{\sigma_M} v_j$.
Hence there exists a directed cycle $C$ in $D$ containing only vertices of $\Fkt{V}{T_1}$ and $x_1$ in particular.
Let $\Brace{t,x_1}$ and $\Brace{x_1,t'}$ be the two edges of $C$ incident with $x_1$.
By choice of $x_1$ we know $t,t'\notin M_i$ and thus, by $M_i$ being a module, the edges $\Brace{t,x_2}$ and $\Brace{x_2,t'}$ must also exist in $D$.
Thus we can create a new directed cycle $C'$ by replacing $x_1$, $\Brace{t,x_1}$, and $\Brace{x_1,t'}$ in $C$ with $x_2$, $\Brace{t,x_2}$, and $\Brace{x_2,t'}$.
This however means hat $x_1$ and $x_2$ cannot belong to different components of $T_{\sigma}-x_0$ by the definition of elimination trees and so $M_i$ must induce a subpath of $T_{\sigma}$.
Let $P_i$ be this path and let $x\in M_i$ be the endpoint of $P_i$ lowest in $T_{\sigma}$, i.e.\@ furthest away from $r_{\sigma}$.
Assume there is a vertex $y\in M_i\setminus\Set{x}$ such that $y$ has another successor $z$, which, since $M_i$ induces a path, cannot be a member of $M_i$.
Let $T_z$ be the subtree of $T_{\sigma}$ rooted at $z$, then $\Fkt{V}{T_z}\cap M_i=\emptyset$.
Since $M_i$ is a module and $z$ a successor of $y$ in the elimination tree $T_{\sigma}$ we again find a directed cycle $C$ that contains $y$, no other vertex of $M_i$ and has all other vertices in $T_z$.
However, by using the module property of $M_i$ we can now find, as we did above, another directed cycle $C'$ be replacing $y$ in $C$ with $x$, which again contradicts $x$ and $z$ to be contained in different components of $T_{\sigma}-y$.
Thus no vertex of $M_i$ except $x$ can have more than one successor.
Therefore every path from $r_{\sigma}$ to a leaf of $T_{\sigma}$ that contains a vertex of $M_i$ must contain the whole path $P_i$.
Now let $P$ be a path in $T_{\sigma}$ from $r_{\sigma}$ to a leaf $t$ of $T_{\sigma}$ that witnesses the rank of $\sigma$.
Let $i\in[\ell]$ be such that $t\in M_i$.
Let $I\subseteq[\ell]$ be the set of indices such that $M_j\cap\Fkt{V}{P-M_i}\neq\emptyset$.
By our observations above this means $M_j\subseteq\Fkt{V}{P-M_i}$ for all $j\in I$.
At last note that the subpath of $P$ induced by vertices of $M_i$ must have $\Fkt{\operatorname{rank}}{\sigma_{M_i}}$ vertices by the maximality of $P$.
Hence $\Fkt{rank}{\sigma}=\Abs{\bigcup_{j\in I}M_j}+\Fkt{\operatorname{rank}}{\sigma_{M_i}}$.
Now let $P_M$ be the subpath of $T_{\sigma_M}$ corresponding to $P$.
Note that, by the maximality of $P$, $P_M$ also has maximum weight with respect to the $n$-$w$-$\sigma_M$-evaluation function.
By definition and the construction of $\sigma_M$ we also have
\begin{align*}
	\Fkt{n\text{-}w\text{-}\operatorname{cr}}{D_M}\leq\Fkt{n\text{-}w\text{-}\operatorname{rank}}{\sigma_M}=\sum_{t\in V(P_M-v_i)}\Fkt{n}{t}+\Fkt{w}{v_i}=\Fkt{\operatorname{cr}}{D}.
\end{align*}

Now for $\Fkt{n\text{-}w\text{-}\operatorname{cr}}{D_M}\geq\Fkt{\operatorname{cr}}{D}$ let us assume $\sigma_M$ to be an ordering for $\Fkt{V}{D_M}$ of minimum $n$-$w$-rank.
For every $i\in[ell]$ let $\sigma_i$ be an ordering for $M_i$ of minimum rank.
Let $\sigma'_M$ be the ordering $\sigma_M$ induces on $[\ell]$, then we construct the ordering $\sigma$ for $\Fkt{V}{D}$ by concatenating the $\sigma_i$ in order of $\sigma'_M$.
Then $\sigma$, in particular, respects $\mathcal{M}$.
Let $P$ be a path in $T_{\sigma}$ from the root $r_{\sigma}$ to a leaf of the elimination tree with a maximum number of vertices.
Moreover let $i\in[\ell]$ be such that the endpoint of $P$ that is not $r_{\sigma}$ lies in $M_i$ and let $I\subseteq[\ell]$ be the set of indices such that $j\in I$ if and only if $M_i\neq M_i$ and $M_j\cap\Fkt{V}{P}\neq\emptyset$.
Hence we may conclude $\Abs{\Fkt{V}{P}}\leq\sum_{j\in I}\Abs{M_j}+\Fkt{\operatorname{rank}}{\sigma_i}$.
Therefore we have
\begin{align*}
	\Fkt{\operatorname{rank}}{\sigma}\leq\sum_{j\in I}\Abs{M_j}+\Fkt{\operatorname{rank}}{\sigma_i}=\Fkt{n\text{-}w\text{-}\operatorname{cr}}{D_M}
\end{align*}
by the choice of $\sigma_M$.
Since the construction of $\sigma$ is the concatenation of the $\sigma_i$ for $i\in[\ell]$ in order, this can be done in time $\Fkt{\mathcal{O}}{\ell\Fkt{\log}{\ell}}$.
\end{proof}

\begin{proof}[Proof of \cref{thm:CR}]
Assume we are given as an instance a directed graph $D$ of directed modular width at most $\omega$.
If $D$ consists of a single vertex $v$, we output the triple $\Brace{\Brace{v},1,1}$.
The sequence $\Brace{v}$ is an ordering of the vertex set of our digraph, $1$ is the number of its vertices and $1$ is its cycle-rank and in particular the rank of the ordering $\Brace{v}$.
If $D$ has at least two vertices we apply the algorithm from \cref{computedecomposition} in order to obtain a non-trivial decomposition of $\Fkt{V}{D}$ into modules $\mathcal{M}=\Set{M_1,\dots,M_{\ell}}$ where $2\leq\ell\leq\omega$.
We compute the induced subgraphs $\InducedSubgraph{D}{M_1},\dots,\InducedSubgraph{D}{M_{\ell}}$ as well as the corresponding module-digraph $D_M$.
Now we recursively apply the algorithm to each of the instances $\InducedSubgraph{D}{M_i}$, $i\in[\ell]$.
By doing this, for each of these graphs we obtain as the output of our algorithm the triple $\Brace{\sigma_i,n_i,w_i}$ where $\sigma_i$ is an ordering of $M_i$ of minimum rank, $n_i=\Abs{M_i}$ and $w_i=\Fkt{\operatorname{rank}}{\sigma_i}=\Fkt{\operatorname{cr}}{\InducedSubgraph{D}{M_i}}$ for all $i\in[\ell]$.
We now define two weight functions $n\colon\Fkt{V}{D_M}\rightarrow\mathbb{N}$ and $w\colon\Fkt{V}{D_M}\rightarrow\mathbb{N}$ for the vertices $v_i\in\Fkt{V}{D}$, $i\in[\ell]$,
of the module-digraph $D_M$ by setting $\Fkt{n}{v_i}\coloneqq n_i$ and $\Fkt{w}{v_i}=w_i$.
By \cref{lemma:nwrankconstantvertices} we can now compute an ordering $\sigma_M$ of minimum $n$-$w$-rank for $\Fkt{V}{D_M}$ in time $\Fkt{\mathcal{O}}{\omega^3\omega!}$.
Once we produced $\sigma_M$ we apply the procedure from \cref{lemma:nwranktocyclerank} to obtain an ordering of minimum rank for $\Fkt{V}{D}$ in time $\Fkt{\mathcal{O}}{\omega\Fkt{\log}{\omega}}$.
In total, the computation time spent on $D$ alone can be bounded by $\Fkt{\mathcal{O}}{n^2+\omega\Fkt{\log}{\omega}+\omega^3\omega!}=\Fkt{\mathcal{O}}{n^2+\omega^3\omega!}$.

As we recurse on the digraphs induced by the modules, which all have less than $n$ vertices, the same bound as for $D$ applies for all such digraphs induced by modules an at least two vertices.
Therefore the total run-time required by our algorithm is bounded from above by
\begin{align*}
\Fkt{\mathcal{O}}{n+n\Brace{n^2+\omega^3\omega!}}=\Fkt{\mathcal{O}}{n^3+\omega^3\omega!n}.
\end{align*}
\end{proof}

\section{Conclusion}
In this paper, we have initiated the study of the \emph{directed modular width} in parametrised algorithmics, which is a natural structural parameter for digraphs. We have seen that by combining dynamic programming with the strong tool of bounded-variable ILP solving, one can obtain FPT-algorithms for many intrinsically hard problems on directed graphs, which are $W[1]$-hard, intractable or still unsolved for classes of bounded directed tree-width, DAG-width or clique-width. In fact, while no FPT-algorithms are known for the \ref{prob:rVDDP} on digraphs of bounded directed tree-width and the Hamiltonian cycle problem is $W[1]$-hard for digraphs of bounded clique-width. Moreover, the recursive nature of module-decompositions allows us to find fast FPT-algorithms for natural weighted generalisations of these problems. Weighted generalisations such as \ref{prob:rVDDP-C} for \ref{prob:rVDDP} naturally appear in real-world problems, as routing problems often involve capacities and costs, which may depend on the location (respectively the vertices in the network).

Our results show that the directed modular width covers a nice niche in the landscape of directed width measures, as it can be computed efficiently, is small on dense but structured networks and avoids the algorithmic price of generality paid by most other width measures for directed graphs. We want to emphasize that although the directed modular width is very restrictive, it still covers non-trivial special cases such as directed co-graphs, which have been investigated previously.

Furthermore, when faced with essentially any hard algorithmic problem on digraphs, the results and techniques developed in this paper are worthwhile to be used as preprocessing or intermediate steps in other algorithms to achieve substantial speed-ups. Because a module-decomposition with a minimal number of modules can be computed in linear time, from a practical point of view, this comes at a relatively small price, but, depending on the instance, possibly with a huge pay-off.

As explained, our algorithmic solutions all follow a common general strategy. We wonder whether it is possible to use these ideas to obtain a strong algorithmic meta-theorem, which allows for more general problems than the one for clique-width. Another possible direction of future research would be to investigate generalisations of the notion of directed modular width, where instead of bounding the number of vertices of the module-digraph $D_M$ by a constant, we impose more general structural properties (for instance bounded directed treewidth or planarity). As we have investigated the complexity of computing other structural digraph parameters in \cref{sec:otherwidth}, we conclude with the following open problem:

Is it possible to compute or approximate the DAG-width of a given digraph in polynomial time for bounded directed modular width? 

It is not clear that such an algorithm should exist. In fact, the gadgets involved in the proof of the PSPACE-completeness of computing the DAG-width in \cite{pspace} are quite structured and admit decompositions into few modules. It is easily seen from their results that there exist digraphs on $n$ vertices of directed modular width at most $6$, which require DAG-decompositions of super-polynomial size in $n$.

\bibliography{references}
\bibliographystyle{alpha}

\end{document}